\documentclass[a4paper,UKenglish,cleveref,autoref]{lipics-v2019}

\usepackage[utf8]{inputenc}
\usepackage[ruled,vlined,linesnumbered,noresetcount]{algorithm2e}
\usepackage{mathpartir}
\usepackage{dashbox}
\usepackage{arydshln}
\usepackage{float}
\usepackage[numbers,sort]{natbib}
\usepackage{pdflscape}

\usepackage{pgf}
\usepackage{tikz}
\usetikzlibrary{arrows,automata,decorations.pathreplacing}

\usepackage{xcolor}
\definecolor{light-gray}{gray}{0.85}

\title{Deconstructing Stellar Consensus\newline(Extended Version)}
\titlerunning{Deconstructing Stellar Consensus}

\author{Álvaro García-Pérez}{IMDEA Software Institute, Madrid, Spain}{}{}{}
\author{Maria A. Schett}{University College London, United Kingdom}{}{}{}

\authorrunning{A. García-Pérez and M. A. Schett}
\Copyright{Álvaro García-Pérez and Maria A. Schett}



\EventEditors{Pascal Felber, Roy Friedman, Seth Gilbert, and Avery Miller}
\EventNoEds{4}
\EventLongTitle{23rd International Conference on Principles of Distributed Systems (OPODIS 2019)}
\EventShortTitle{OPODIS 2019}
\EventAcronym{OPODIS}
\EventYear{2019}
\EventDate{December 17--19, 2019}
\EventLocation{Neuch\^{a}tel, Switzerland}
\EventLogo{}
\SeriesVolume{153}
\ArticleNo{4}

\hideLIPIcs
\nolinenumbers

\floatstyle{plain}
\newfloat{algo}{pthb}{loa}
\floatname{algo}{Algorithm}

\newif\ifcomments
\newcommand{\mute}[1]{\ifcomments{#1}\fi}

\definecolor{msgreen}{rgb}{0.53, 0.66, 0.42}
\definecolor{agblue}{rgb}{0.13, 0.67, 0.8}
\definecolor{alred}{rgb}{1.0, 0.44, 0.37}

\newcommand{\al}[1]{\mute{\textcolor{alred}{(Alvaro: {#1})}}}

\newcommand{\ie}{\emph{i.e.}\xspace}

\newcommand{\VV}{\mathbf{A}}
\newcommand{\Tag}{\mathbf{Tag}}

\newcommand{\Value}{\mathbf{Val}} 
\newcommand{\Nat}{\mathbb{N}} 
\newcommand{\Counters}{\Nat^+} 
\newcommand{\Ballot}{\mathbf{Ballot}} 
\newcommand{\bllt}[2]{\langle #1, #2 \rangle}
\newcommand{\lic}{\lnsim} 
\newcommand{\lc}{\lesssim} 
\newcommand{\cl}{\gtrsim} 
\newcommand{\nbllt}{\bllt{0}{\bot}}
\newcommand{\round}{n}
\newcommand{\val}{x}

\newcommand{\AbstractConsensus}{Abstract SCP}
\newcommand{\abstractConsensus}{abstract SCP}
\newcommand{\AFCP}{ASCP}

\newcommand{\ConcreteConsensus}{Concrete SCP}
\newcommand{\concreteConsensus}{concrete SCP}
\newcommand{\CFCP}{CSCP}

\SetKwBlock{SubAlgoBlock}{}{end}
\newcommand{\SubAlgo}[2]{#1 \SubAlgoBlock{#2}}
\SetKw{Procc}{process}
\SetKw{Send}{send}
\SetKw{Upon}{when}
\SetKw{Received}{received}
\SetKw{Triggered}{triggered}
\SetKw{Trigger}{trigger}
\SetKw{From}{from}
\SetKw{Every}{every}
\SetKw{Each}{each}
\SetKw{And}{and}
\SetKw{Not}{not}
\SetKw{Either}{either}
\SetKw{Or}{or}
\SetKw{To}{to}
\SetKw{With}{with}
\SetKw{SuchThat}{such that}
\SetKw{Exists}{exists}
\SetKw{Exist}{exist}
\SetKw{Forall}{for all}
\SetKw{Do}{do}
\SetKw{For}{for}
\SetKw{Some}{some}
\SetKw{Then}{then}
\SetKw{Where}{where}
\SetKw{New}{new}
\SetKw{Time}{time}
\SetKw{Set}{set}
\SetKw{After}{after}
\SetKw{Expires}{expires}
\SetKw{Wait}{wait}
\SetKw{Until}{until}

\newcommand{\V}{\mathbf{V}}
\newcommand{\Qs}{\mathcal{Q}}
\newcommand{\Ss}{\mathcal{S}}

\newcommand{\ok}{\mathrm{ok}}

\newcommand{\powerset}[1]{2^{#1}}

\newcommand{\Bool}{{\sf Bool}}
\newcommand{\true}{\mathit{true}}
\newcommand{\false}{\mathit{false}}

\newcommand{\readied}{{\sf readied}}
\newcommand{\ready}{{\sf ready}}
\newcommand{\delivered}{{\sf delivered}}

\newcommand{\deliver}{{\sf deliver}}

\newcommand{\READY}{{\tt READY}}

\newcommand{\federatedVoting}{{\sf federated\text{-}voting}}

\newcommand{\voted}{{\sf voted}}

\newcommand{\vote}{{\sf vote}}

\newcommand{\VOTE}{{\tt VOTE}}

\newcommand{\abstractconsensus}{{\sf abstract\text{-}consensus}}

\newcommand{\brs}{{\sf brs}}
\newcommand{\fvs}{{\sf ballots}}
\renewcommand{\b}{{\sf candidate}}
\newcommand{\h}{{\sf prepared}}
\newcommand{\n}{{\sf round}}

\newcommand{\commit}{{\sf commit}}

\newcommand{\starttimer}{{\sf start\text{-}timer}}
\newcommand{\timeout}{{\sf timeout}}

\newcommand{\propose}{{\sf propose}}
\newcommand{\decide}{{\sf decide}}

\newcommand{\bunchedVoting}{{\sf bunched\text{-}voting}}
\newcommand{\prepare}{{\sf prepare}}
\newcommand{\committed}{{\sf committed}}
\newcommand{\prepared}{{\sf prepared}}

\newcommand{\he}{{\sf max\text{-}{voted\text{-}prep}}}
\newcommand{\hr}{{\sf max\text{-}{readied\text{-}prep}}}
\newcommand{\hd}{{\sf max\text{-}{delivered\text{-}prep}}}

\newcommand{\Ce}{{\sf ballots\text{-}{voted\text{-}cmt}}}
\newcommand{\Cr}{{\sf ballots\text{-}{readied\text{-}cmt}}}
\newcommand{\Cd}{{\sf ballots\text{-}{delivered\text{-}cmt}}}

\newcommand{\M}{{\tt M}}

\newcommand{\CMT}{\textup{\sc cmt}}
\newcommand{\PREP}{\textup{\sc prep}}

\newcommand{\STA}{\textup{\sc s}}

\newcommand{\fh}{\phi}

\newcommand{\concreteconsensus}{{\sf concrete\text{-}consensus}}

\newcommand{\send}{{\sf send}}
\newcommand{\sendBatch}{{\sf send\text{-}batch}}
\newcommand{\receive}{{\sf receive}}
\newcommand{\receiveBatch}{{\sf receive\text{-}batch}}
\newcommand{\transmitBatch}{{\sf \transmit\text{-}batch}}
\newcommand{\voteBatch}{{\sf vote\text{-}batch}}
\newcommand{\deliverBatch}{{\sf deliver\text{-}batch}}

\newcommand{\concat}{\cdot}

\newcommand{\maptrace}{\sigma}
\newcommand{\transmit}{{\sf op}}

\newcommand{\emptrace}{[\,]}

\newcommand{\hist}{H}

\def\SmallTabColSep{\setlength{\tabcolsep}{1pt}}

\newcommand{\A}{\mathbf{1}}
\newcommand{\B}{\mathbf{2}}
\newcommand{\C}{\mathbf{3}}

\makeatletter
\renewcommand{\section}{\@startsection{section}{1}{0pt}%
{-3ex plus -1ex minus -.2ex}{1.5ex plus.2ex}%
{\normalfont\large\bfseries}}
\renewcommand{\subsection}{\@startsection{subsection}{1}{0pt}%
{-2ex plus -1ex minus -.2ex}{1ex plus.2ex}%
{\bfseries}}
\def\@listI{\leftmargin\leftmargini
            \parsep 0\p@ \@plus1\p@ \@minus\p@
            \topsep 6\p@ \@plus2\p@ \@minus0\p@
            \itemsep 0\p@}
\let\@listi\@listI
\@listi
\makeatother

\makeatletter
\renewcommand{\paragraph}{%
  \@startsection%
    {paragraph}%
    {4}%
    {0pt}%
    {7pt plus 2pt minus 2pt}
    {-1em}%
    {\normalsize \bfseries}}
\makeatother

\commentsfalse

\begin{document}

\maketitle

\begin{abstract}
  Some of the recent blockchain proposals, such as Stellar and Ripple, allow for open
membership while using quorum-like structures typical for classical Byzantine consensus
with closed membership. This is achieved by constructing quorums in a decentralised way:
each participant independently chooses whom to trust, and quorums arise from these
individual decisions. Unfortunately, the consensus protocols underlying such blockchains
are poorly understood, and their correctness has not been rigorously investigated. In this
paper we rigorously prove correct the Stellar Consensus Protocol (SCP), with our proof
giving insights into the protocol structure and its use of lower-level abstractions.  To
this end, we first propose an abstract version of SCP that uses as a black box Stellar's
federated voting primitive (analogous to reliable Byzantine broadcast), previously
investigated by García-Pérez and Gotsman~\cite{GG18}. The abstract consensus protocol
highlights a modular structure in Stellar and can be proved correct by reusing the
previous results on federated voting. However, it is unsuited for realistic
implementations, since its processes maintain infinite state. We thus establish a
refinement between the abstract protocol and the concrete SCP that uses only finite state,
thereby carrying over the result about the correctness of former to the latter. Our
results help establish the theoretical foundations of decentralised blockchains like
Stellar and gain confidence in their correctness.


\end{abstract}

\section{Introduction}\label{sec:introduction}

Permissioned blockchains are becoming increasingly popular due to the low-energy
consumption and hard guarantees they provide on when a transaction can be considered
successfully committed. Such blockchains are often based on classical Byzantine
fault-tolerant (BFT) consensus protocols, like PBFT~\cite{CL02}. In these protocols
consensus is reached once a {\em quorum} of participants agrees on the same
decision. Quorums can be defined as sets containing enough nodes in the system (e.g.,
$2f+1$ out of $3f+1$, assuming at most $f$ failures) or by a more general structure of a
Byzantine quorum system (BQS)~\cite{MR98}.  Unfortunately, defining quorums in this way
requires fixing the number of participants in the system, which prevents decentralisation.

Some of the recent blockchain proposals, such as Stellar~\cite{M2015} %
and Ripple~\cite{SYB14}, allow for open membership while using quorum-like structures
typical for classical Byzantine consensus with closed membership. This is achieved by
constructing quorums in a decentralised way: each protocol participant independently
chooses whom to trust, and quorums arise from these individual decisions. In particular,
in Stellar trust assumptions are specified using a {\em federated Byzantine quorum system
  (FBQS)}, where each participant selects a set of {\em quorum slices}---sets of nodes
each of which would convince the participant to accept the validity of a given statement
(\S\ref{sec:background}). Quorums are defined as sets of nodes $U$ such that each node in
$U$ has some quorum slice fully within $U$, so that the nodes in a quorum can potentially
reach an agreement. Consensus is then implemented by a fairly intricate protocol whose key
component is {\em federated voting}---a protocol similar to Bracha's protocol for reliable
Byzantine broadcast~\cite{Bra87,CGR11}. Unfortunately, even though Stellar has been
deployed as a functioning blockchain, the structure of the consensus protocol underlying
it is poorly understood and its correctness has not been rigorously investigated. In this
paper we aim to close this gap, rigorously defining and proving correct the Stellar
Consensus Protocol (SCP). Apart from giving more confidence in the correctness of the
protocol, our proof is structured in such a way as to give insights into its structure and
its use of lower-level abstractions.

In more detail, the guarantees provided by SCP are nontrivial. When different
participants in an FBQS choose different slices, only a subset of the participants may
take part in a subsystem in which every two quorums intersect in a correct node---a
property required for achieving consensus. The system may partition into such subsystems,
and SCP will guarantee agreement within each of them. In blockchain terms, the blockchain
may fork, but in this case each fork will be internally consistent, a property that is
enough for business applications of the Stellar blockchain. The subsystems where agreement
is guaranteed are characterised by Mazières et al.~\cite{MLG19} through the notion of
\emph{intact sets}. Our proof of correctness establishes safety and liveness properties of
SCP relative to such intact sets (\S\ref{sec:spec}).

As a stepping stone in the proof, we first propose an {\em abstract version} of SCP that
uses as a black box Stellar's federated voting primitive (analogous to reliable Byzantine
broadcast) previously investigated by García-Pérez and Gotsman~\cite{GG18}
(\S\ref{sec:abstract-federated-consensus}). This abstract formulation allows specifying
the protocol concisely and highlights the modular structure present in it. This allows
proving the protocol by reusing the previous results on federated voting~\cite{GG18}
(reviewed in \S\ref{sec:stellar-broadcast}). However, the abstract protocol is unsuited
for realistic implementations, since its processes maintain infinite state. To address
this, we formulate a realistic version of the protocol---a {\em concrete SCP}---that uses
only finite state. We then prove a refinement between the abstract and concrete SCP,
thereby carrying over the result about the correctness of former to the latter
(\S\ref{sec:concrete-federated-consensus}).

A subtlety in SCP is that its participants receive information about quorum slices of
other participants directly from them. Hence, Byzantine participants may lie to others
about their choices of quorum slices, which may cause different participants to disagree
on what constitutes a quorum. Our results also cover this realistic case
(\S\ref{sec:lying}).

Overall, our results help establish the theoretical foundations of decentralised
blockchains like Stellar and gain confidence in their correctness. %
Proofs of the lemmas and theorems in the paper are given in the appendices.


\section{Background: System Model and Federated Byzantine Quorum Systems}
\label{sec:background}

\paragraph*{System model.}
We consider a system consisting of a finite {\em universe} of {\em nodes} $\V$ and assume
a Byzantine failure model where {\em faulty} nodes can deviate arbitrarily from their
specification. All other nodes are called {\em correct}. Nodes that are correct, or that
only deviate from their specification by stopping execution, are called
\emph{honest}. Nodes that deviate from their specification in ways other than stopping are
called \emph{malicious}. We assume that any two nodes can communicate over an
authenticated perfect link. We assume a partial synchronous network, which guarantees that
messages arrive within bounded time after some unknown, finite \emph{global stabilisation
  time} (\emph{GST}). Each node has a local timer and a timeout service that can be
initialised with an arbitrary delay $\Delta$. We assume that after GST the clock skew of
correct nodes is bounded, \ie, after GST two correct nodes can only disagree in the
duration of a given delay $\Delta$ by a bounded margin.

\paragraph*{Federated Byzantine quorum systems.}

Given a finite universe $\V$ of nodes, a \emph{federated Byzantine quorum system}
(\emph{FBQS}) \cite{M2015,GG18} is a function
$\Ss : \V \to \powerset{\powerset{\V}} \setminus \{\emptyset\}$ that specifies a non-empty
set of {\em quorum slices} for each node, ranged over by $q$. We require that a node
belongs to all of its own quorum slices:
$\forall v\in\V.\, \forall q\in \Ss(v).\, v \in q$. Quorum slices reflect the trust
choices of each node. A non-empty set of nodes $U\subseteq \V$ is a \emph{quorum} in an
FBQS $\Ss$ iff $U$ contains a slice for each member, \ie,
$\forall v\in U.\, \exists q\in \Ss(v).\, q\subseteq U$.

For simplicity, for now we assume that faulty nodes do not equivocate about their quorum
slices, so that all the nodes share the same FBQS. In \S\ref{sec:lying} we consider the
more realistic \emph{subjective FBQS} \cite{GG18}, where malicious nodes may lie about
their slices and different nodes have different views on the FBQS. There we also lift the
results on the subsequent sections of the paper to subjective FBQSes.

\begin{example}\label{ex:three-f-plus-one}
  Consider a universe $\V$ with $3f+1$ nodes, and consider the FBQS $\Ss$ where for every
  node $v\in \V$, the set of slices $\Ss(v)$ consists of every set of $2f+1$ nodes that
  contains $v$ itself. $\Ss$ encodes the classical cardinality-based quorum system of $3f+1$
  nodes with failure threshold $f$, since every set of $2f+1$ or more nodes is a quorum.
\end{example}

\begin{example}\label{ex:split}
  Let the universe $\V$ contain four nodes $v_1$ to $v_4$, and consider the FBQS $\Ss$
  in the diagram below.
  \begin{center}
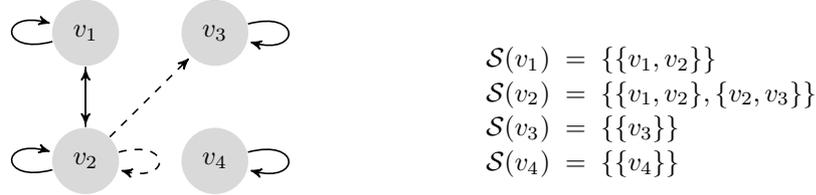
    \begin{minipage}{.4\linewidth}
      \begin{tikzpicture}[->,>=stealth',shorten >=1pt,auto,node distance=1.7cm,
        semithick]
        \tikzstyle{every state}=[fill=light-gray,draw=none,text=black]

        \node[state] (1)              {$v_1$};
        \node[state] (2) [below of=1] {$v_2$};
        \node[state] (3) [right of=1]  {$v_3$};
        \node[state] (4) [below of=3] {$v_4$};

        \path[dashed] (2) edge (3);
        \path (1) edge (2)
              (2) edge (1);
        \path (1) edge [loop left] (1);
        \path (2) edge [loop left] (2);
        \path[dashed] (2) edge [loop right] (2);
        \path (3) edge [loop right] (3);
        \path (4) edge [loop right] (4);
      \end{tikzpicture}
    \end{minipage}
    \begin{minipage}{.4\linewidth}
      \begin{displaymath}
        \begin{array}{l}
          \Ss(v_1)\;=\;\{\{v_1,v_2\}\}\\
          \Ss(v_2)\;=\;\{\{v_1,v_2\},\{v_2,v_3\}\}\\
          \Ss(v_3)\;=\;\{\{v_3\}\}\\
          \Ss(v_4)\;=\;\{\{v_4\}\}
        \end{array}
      \end{displaymath}
    \end{minipage}
  \end{center}
  For each node, all the outgoing arrows with the same style determine one slice. Node
  $v_2$ has two slices, determined by the solid and dashed arrow styles respectively. The
  rest of the nodes have one slice. $\Ss$ has the following set of quorums $\Qs = {}$
  \begin{center}
    $\{\{v_1,v_2\},\{v_2,v_3\},\{v_3\},\{v_4\}, \{v_1,v_2,v_3\},\{v_3,v_4\},
    \{v_1,v_2,v_4\},\{v_2,v_3,v_4\}, \{v_1,v_2,v_3,v_4\}\}$.
  \end{center}
\end{example}
A consensus protocol that runs on top of an FBQS may not guarantee global agreement,
because when nodes choose slices independently, only a subset of the nodes may take part
in a subsystem in which every two quorums intersect in at least one correct node---a basic
requirement of a Byzantine quorum system \cite{MR98} to ensure agreement. To formalise
which parts of the system may reach agreement internally, we borrow the notions of
\emph{intertwined nodes} and of \emph{intact set} from \cite{MLG19}. Two nodes $v_1$ and
$v_2$ are \emph{intertwined} iff they are correct and every quorum containing $v_1$
intersects every quorum containing $v_2$ in at least one correct node. Consider an FBQS
$\Ss$ and a set of nodes $I$. The {\em projection} $\Ss|_I$ of $\Ss$ to $I$ is the FBQS
over universe $I$ given by $\Ss|_I(v) = \{q \cap I \mid q \in \Ss(v)\}$. For a given set
of faulty nodes, a set $I$ is an \emph{intact set} iff $I$ is a quorum in $\Ss$ and every
member of $I$ is intertwined with each other in the projected FBQS $\Ss|_I$. The intact
sets characterise those sets of nodes that can reach consensus, which we later show using
the following auxiliary result.
\begin{lemma}
  \label{lem:quorums-intersect-intact-set}
  Let $\Ss$ be an FBQS and assume some set of faulty nodes. Let $I$ be an intact set in
  $\Ss$ and consider any two quorums $U_1$ and $U_2$ in $\Ss$ such that
  $U_1\cap I\not= \emptyset$ and $U_2\cap I\not= \emptyset$. Then the intersection
  $U_1\cap U_2$ contains some node in $I$.
\end{lemma}

The maximal intact sets are disjoint with each other:
\begin{lemma}\label{lem:intact-sets-closed-union}
  Let $\Ss$ be an FBQS and assume some set of faulty nodes. Let $I_1$ and $I_2$ be two
  intact sets in $\Ss$. If $I_1\cap I_2\not=\emptyset$ then $I_1\cup I_2$ is an intact set
  in $\Ss$.
\end{lemma}
In SCP the system may split into different partitions---\ie, the maximal intact
sets---that may be inconsistent with each other, but which constitute independent systems
each of which can reach consensus.

Consider the $\Ss$ from Example~\ref{ex:three-f-plus-one}, which encodes the
cardinality-based quorum system of $3f+1$ nodes, and let $f=1$, so that the universe $\V$
contains four nodes $v_1$ to $v_4$. If we assume that node $v_3$ is faulty, then the set
$I=\{v_1,v_2,v_4\}$ is the only maximal intact set: $I$ is a quorum in $\Ss$, and $\Ss|_I$
contains the quorums $\{\{v_1,v_2\},\{v_2,v_4\},\{v_1,v_4\},\{v_1,v_2,v_4\}\}$, which
enjoy quorum intersection. This ensures that every two nodes in $I$ are intertwined in the
projected system $\Ss|_I$.

Now consider the $\Ss$ from Example~\ref{ex:split}. If we assume that node $v_3$ is
faulty, then the sets $I=\{v_1,v_2\}$ and $I'=\{v_4\}$ are the maximal intact sets: $I$
and $I'$ are quorums in $\Ss$, and the projected systems $\Ss|_I$ and $\Ss|_{I'}$ enjoy
quorum intersection---$\Ss_I$ contains quorums $\{v_1,v_2\}$ and $\{v_2\}$, and
$\Ss|_{I'}$ contains quorum $\{v_4\}$---which ensures that every two nodes in either $I$
or $I'$ are intertwined in the projected systems $\Ss|_I$ and $\Ss_{I'}$ respectively. It
is easy to check that adding any set of correct nodes to either $I$ or $I'$ results in
sets that are not quorums in $\Ss$, or in projected systems that contain some pairs of
nodes that are not intertwined.


\section{Specifications}
\label{sec:spec}

Assume a set $\Value$ of \emph{consensus values}. In the consensus protocols that we study
in \S\ref{sec:abstract-federated-consensus}--\ref{sec:concrete-federated-consensus}, each
correct node proposes some $x\in\Value$ through an invocation $\propose(x)$, and each
node may decide some $x'\in \Value$ through an indication $\decide(x')$. We consider a
variant of the \emph{weak Byzantine consensus} specification in \cite{CGR11} that we call
\emph{non-blocking Byzantine consensus for intact sets}, which is defined as
follows. Given a maximal intact set $I$,
\begin{description}
\item[\rm\color{black}(\emph{Integrity})] no correct node decides twice,
\item[\rm\color{black}(\emph{Agreement for intact sets})] no two nodes in $I$ decide
  differently,
\item[\rm\color{black}(\emph{Weak validity for intact sets})] if all nodes
    are honest and every node proposes $x$, then no node in $I$ decides a consensus value
    different from $x$; furthermore, if all nodes are honest and some node in $I$ decides
    $x$, then $x$ was proposed by some node, and \label{pag:weak-validity}
\item[\rm\color{black}(\emph{Non-blocking for intact sets})] if a node $v$ in $I$ has not
  yet decided in some run of the protocol, then for every continuation of that run in
  which all the malicious nodes stop, node $v$ eventually decides some consensus value.
\end{description}
\noindent
The usual \emph{Weak validity} property of consensus~\cite{CGR11} ensures that if all
nodes are correct and they propose the same consensus value, then no node can decide a
consensus value different from the proposed one; and that if all nodes are correct, then a
node can only decide a consensus value proposed by some node. Our \emph{Weak validity for
  intact sets} above adapts this requirement to the nodes in a maximal intact set, and
weakens its condition by assuming that all nodes are honest instead of correct. Notice
that if every two quorums intersect our property entails the usual one because a correct
node is also honest, and because if all nodes are correct then the maximal intact set is
the universe. For instance, this condition holds in the cardinality-based quorum systems
($3f+1$).

The consensus protocols we consider in this paper specify the behaviour of SCP's ballot
protocol~\cite{M2015,MLG19} with one of its suggested strategies for managing timeouts
(Strategy~1 from \cite{MLG19}). As discussed in \cite{MLG19}, in SCP malicious nodes with
good network timing could permanently delay the termination of the nodes in an intact set,
and thus the protocol does not provide the usual \emph{Termination} guarantee that every
correct node eventually decides some consensus value~\cite{CGR11}. Instead, we consider
the weaker liveness guarantee of \emph{Non-blocking for intact sets}, which we have
obtained by adapting the \emph{Non-blocking} property in
\cite{SRSD08}. \emph{Non-blocking} requires that some continuation of a given run
\emph{exists} in which every correct node terminates. Our \emph{Non-blocking for intact
  sets} adapts this requirement to the nodes in a maximal intact set, and requires that
they terminate in \emph{every} continuation of the run in which malicious nodes are
stopped. It is easy to check that if every correct node is in an intact set, then
\emph{Non-blocking for intact sets} entails \emph{Non-blocking} in \cite{SRSD08}. For
instance, this condition holds in the cardinality-based quorum systems ($3f+1$). Besides,
if every correct node is honest, then \emph{Non-blocking for intact sets} entails the
usual {\em Termination} property that guarantees that every correct node eventually
decides some consensus value.

The \emph{non-blocking Byzantine consensus for intact sets} above entails the \emph{weak
  Byzantine consensus} specification \cite{CGR11} in the cardinality-based quorum systems
($3f+1$), which guarantees the \emph{Integrity} property above, as well as the usual
\emph{Agreement} property that ensures that no two correct nodes decide differently, and
the usual \emph{Weak validity} and \emph{Termination} properties that we have recalled in
the paragraphs above.

One of the core components of the consensus protocol in
\S\ref{sec:abstract-federated-consensus} is \emph{federating voting} (\emph{FV})
\cite{M2015,MLG19}. Assume a set of \emph{voting values} $\VV$ that could be disjoint with
the set $\Value$ of consensus values (we typically let $\VV$ be the set of Booleans
$\Bool\equiv\{\true,\false\}$). FV allows each correct node to vote for some $a\in\VV$
through an invocation $\vote(a)$, and each node may deliver some $a'\in\VV$ through an
indication $\deliver(a')$. The interface of FV is akin to that of consensus, where each
node activates itself through the primitive $\vote(a)$. However, FV has weaker liveness
guarantees than consensus, which are reminiscent to those of \emph{Byzantine reliable
  broadcast} from \cite{CGR11} and \emph{weakly reliable Byzantine broadcast} from
\cite{GG18}. Here, we consider a variant of the latter specification that we call
\emph{reliable Byzantine voting for intact sets}, which is defined as follows. Given a
maximal intact set $I$,
\begin{description}
\item[\rm\color{black}(\emph{No duplication})] every correct node delivers at most one
  voting value,
\item[\rm\color{black}(\emph{Totality for intact sets})] if a node in $I$ delivers a
  voting value, then every node in $I$ eventually delivers a voting value,
\item[\rm\color{black}(\emph{Consistency for intertwined nodes})] if two intertwined nodes
  $v$ and $v'$ deliver $a$ and $a'$ respectively, then $a = a'$, and
\item[\rm\color{black}(\emph{Validity for intact sets})] if all nodes in $I$ vote for $a$,
  then all nodes in $I$ eventually deliver $a$.
\end{description}

The ability of each node to activate itself independently in the specification above
simulates a malicious sender that may send different voting values to each node in the
specification of \emph{weakly reliable Byzantine broadcast} from \cite{GG18}.


\section{Federated Voting}
\label{sec:stellar-broadcast}

In this section we recall \emph{federated voting} (\emph{FV}) from~\cite{M2015}, which
also corresponds to the {\em Stellar broadcast} considered in~\cite{GG18}. We prove that
FV implements the specification of reliable Byzantine voting for intact sets, thereby
generalising the results of~\cite{GG18} to the case of multiple intact sets within the
system. The consensus protocol that we study in the next section uses multiple instances
of FV independent from each other. This is done by letting each node run a distinct
process for each instance of FV, which is identified by a \emph{tag} $t$ from some
designated set $\Tag$ of tags.

Algorithm~\ref{alg:broadcast} below depicts FV over an FBQS $\Ss$ with set of quorums
$\Qs$. A node $v$ runs a process $\federatedVoting(v,t)$ for each tag $t$. The messages
exchanged by such a process are also tagged with $t$, in order to distinguish them from
the messages exchanged for instances of FV associated with tags different from $t$.

\begin{algo}
  \begin{algorithm}[H]
    \setcounter{AlgoLine}{0}
    \SubAlgo{\Procc $\federatedVoting(v\in\V,t\in\Tag)$}
    {
      \smallskip
      $\voted,\ready,\delivered \leftarrow \false \in \Bool$\;\label{lin:bc-init}
      \smallskip
      \SubAlgo{$\vote(a)$\label{lin:echo}}
      {
        \If{\Not $\voted$\label{lin:if-echoed}}
        {
          $\voted \leftarrow \true$\; \label{lin:echoed}
          \Send $\VOTE(t,a)$ \To \Every $v'\in\V$; \label{lin:send-echo}
        }
      }
      \smallskip
      \SubAlgo{\Upon \Received $\VOTE(t,a)$
        \From \Every $u\in U$ \For \Some $U\in \mathcal{Q}$ \SuchThat $v\in U$
        \And \Not $\ready$ \label{lin:ready}}
      {
        $\ready \leftarrow \true$\; \label{lin:bc-readied}
        \Send $\READY(t,a)$ \To \Every $v'\in\V$;\label{lin:send-ready}
      }
      \smallskip
      \SubAlgo{\Upon \Received
        $\READY(t,a)$ \From \Every $u\in B$ \For \Some
        {\bf $v$-blocking} $B$ \And \Not $\ready$ \label{lin:bc-readied-v-blocking}
        \label{lin:v-blocking}}
      {
        $\ready \leftarrow \true$\;
        \Send $\READY(t,a)$ \To \Every $v'\in\V$;\label{lin:send-v-blocking}
      }
      \smallskip
      \SubAlgo{\Upon \Received
        $\READY(t,a)$ \From \Every $u\in U$ \For \Some $U\in \mathcal{Q}$
        \SuchThat $v\in U$ \And \Not $\delivered$
        \label{lin:deliver}}
      {
        $\delivered \leftarrow \true$\;
        \Trigger $\deliver(a)$;\label{lin:send-deliver}
      }
    }
    \caption{Federated voting (FV) over an FBQS $\Ss$ with set of quorums $\Qs$.}
    \label{alg:broadcast}
  \end{algorithm}
\end{algo}

FV adapts \emph{Bracha's protocol} for reliable Byzantine broadcast \cite{Bra87}, which
works over the cardinality-based quorum systems of $3f+1$ nodes, to the federated setting
of the FBQSs. In FV nodes process each other's messages in several stages, where for each
tag $t$ progress is denoted by several Boolean flags (line~\ref{lin:bc-init} of
Algorithm~\ref{alg:broadcast}). When a node $v$ votes $a$ for tag $t$ for the first time,
the node sends $\VOTE(t,a)$ to every node (including itself, for uniformity;
lines~\ref{lin:echo}--\ref{lin:send-echo}). When a node $v$ receives a $\VOTE(t,a)$
message from a quorum to which $v$ itself belongs, it sends a $\READY(t,a)$ message to
every node, signalling its willingness to deliver the value $a$ for tag $t$
(lines~\ref{lin:ready}--\ref{lin:send-ready}). Note that, for each tag $t$, two nodes in
the same intact set $I$ cannot send $\READY$ messages with two different voting values
through the rule in lines~\ref{lin:ready}--\ref{lin:send-ready}. Indeed, this would
require two quorums of $\VOTE$ messages, each with a node in $I$, with different voting
values for the same tag.  But by Lemma~\ref{lem:quorums-intersect-intact-set} these
quorums would intersect in a node in $I$, which is by definition correct and cannot send
contradictory $\VOTE$ messages for the same tag. When a node $v$ receives the message
$\READY(t,a)$ from a quorum to which $v$ itself belongs, it delivers $a$ for tag $t$
(lines~\ref{lin:deliver}--\ref{lin:send-deliver}).

The exchange of $\READY$ messages in the protocol is necessary to establish liveness
guarantees. It ensures that, if a node in an intact set $I$ delivers a voting value for
some tag, other nodes in $I$ have enough information to also deliver a voting value for
the same tag. This relies on the rule in
lines~\ref{lin:v-blocking}--\ref{lin:send-v-blocking}, which uses the notion of
\emph{$v$-blocking set}~\cite{M2015}. Given a node $v$, a set $B$ is \emph{$v$-blocking}
iff $B$ overlaps each of $v$'s slices, \ie, $\forall q\in \Ss(v).~q\cap B\not= \emptyset$.
(To illustrate this notion, in Example~\ref{ex:three-f-plus-one} every set of $f+1$ nodes
is $v$-blocking for every $v$, and in Example~\ref{ex:split} the set $\{v_1,v_3\}$ is
$v_2$-blocking and the set $\{v_2\}$ is $v_1$-blocking.)
Lines~\ref{lin:v-blocking}--\ref{lin:send-v-blocking} allow a node to send a $\READY(t,a)$
message even if it previously voted for a different voting value for tag $t$: this is done
if $v$ receives $\READY(t,a)$ from each member of a $v$-blocking set. If $v$ is in an
intact set~$I$, the following lemma guarantees that in this case $v$ has received at least
one $\READY(t,a)$ message from some node in $I$.

\begin{lemma}\label{lem:empty-I-not-v-blocking}
  Let $\Ss$ be an FBQS and assume a set of faulty nodes. Let $I$ be an intact set in $\Ss$
  and $v\in I$. Then, no $v$-blocking set $B$ exists such that $B\cap I = \emptyset$.
\end{lemma}
By Lemma~\ref{lem:empty-I-not-v-blocking}, the first node in $I$ to ever send a
$\READY(t,a)$ message for a tag $t$ has to do it through the rule in
lines~\ref{lin:ready}--\ref{lin:send-ready}, and hence the value $a$ has been
cross-checked by a quorum.

If the condition $v\in U$ in lines~\ref{lin:ready} and \ref{lin:deliver} of
Algorithm~\ref{alg:broadcast} was dropped, this could violate \emph{Agreement for intact
  sets} as follows. Take the $\Ss$ from Example~\ref{ex:split} and consider a run of FV
for some tag $t$ where $v_3$ is malicious. Node $v_3$ could respectively send
$\READY(t,a)$ and $\READY(t,a')$ with $a\not=a'$ to correct nodes $v_1$ and $v_2$.  Since
$\{v_3\}\in \Qs$, these nodes will respectively deliver $a$ and $a'$ by
lines~\ref{lin:deliver}--\ref{lin:send-deliver} of Algorithm~\ref{alg:broadcast} without
condition $v\in U$.

Our first contribution is to generalise the results of \cite{GG18} to establish the
correctness of FV within each of the maximal intact sets of an FBQS, as captured by
Theorem~\ref{thm:reliable-byzantine-broadcast-intact-sets} below.

\begin{theorem}\label{thm:reliable-byzantine-broadcast-intact-sets}
  Let $\Ss$ be an FBQS and $t$ be a tag. The instance for $t$ of FV over $\Ss$ satisfies
  the specification of reliable Byzantine voting for intact sets.
\end{theorem}

FV also guarantees the property stated by the following lemma, which helps establish the
liveness properties of the consensus protocol that we introduce in
\S\ref{sec:abstract-federated-consensus}.
\begin{lemma}\label{lem:bounded-totality}
  Let $\Ss$ be an FBQS and $t$ be a tag. Consider an execution of the instance for $t$ of
  FV over $\Ss$. Let $I$ be an intact set in $\Ss$ and assume that GST has expired. If a
  node $v\in I$ delivers a voting value then every node in $I$ will deliver a voting value
  within bounded time.
\end{lemma}
We write $\delta_I$ for the time that a node in $I$ takes to deliver some voting value
after GST and provided that some other node in $I$ already delivered some voting
value. The delay $\delta_I$---which is determined by $\Ss$ and $I$---is unknown, but
Lemma~\ref{lem:bounded-totality} guarantees that it is finite.

\begin{example}\label{ex:federated-voting}
  Consider the $\Ss$ from Example~\ref{ex:three-f-plus-one}, which encodes the
  cardinality-based quorum system $3f+1$, and let $f=1$ such that the universe $\V$
  contains four nodes $v_1$ to $v_4$. Every set of three or more nodes is a quorum, and
  every set of two or more nodes is $v$-blocking for every $v\in \V$. Let us fix a tag $t$
  and consider an execution of the instance of FV for tag $t$ where we let the voting
  values be the Booleans. Assume that nodes $v_1$, $v_2$ and $v_4$ are correct, which
  constitute the maximal intact set. In the execution, nodes $v_1$ and $v_2$ vote
  $\false$, and node $v_4$ votes $\true$. Malicious node $v_3$ sends the message
  $\VOTE(t,\false)$ to every node (highlighted in red) thus helping the correct nodes to
  deliver $\false$.

  \begin{figure}[t]
    \centering

    \begingroup\SmallTabColSep{\small\begin{tabular}[t]{|l|l|l|l|}
    Node $v_1$ &  Node $v_2$ & {\color{red} Node $v_3$} & Node $v_4$\\
    \hline\hline

    $\vote(\false)$&

    $\vote(\false)$&&

    $\vote(\true)$\\

    \hline

    $\VOTE(t,\false)$&

    $\VOTE(t,\false)$&

    {\color{red}$\VOTE(t,\false)$}&

    $\VOTE(t,\true)$\\

    \hline

    $\READY(t,\false)$&

    $\READY(t,\false)$&&\\

    \hline

    &&&$\READY(t,\false)$\\

    \hline

    $\deliver(\false)$&

    $\deliver(\false)$&&

    $\deliver(\false)$\\

    \hline

    \end{tabular}}\endgroup

  \caption{Execution of the instance of FV for tag $t$.}
    \label{fig:trace-federated-voting}
  \end{figure}

  Figure~\ref{fig:trace-federated-voting} depicts a possible execution of FV described
  above, from which a trace can be constructed as follows: all the events in each row may
  happen concurrently, and any two events in different rows happen in real time, where
  time increases downwards; in those cells that are tagged with a message, the node sends
  the message to every node, and in a given cell a node has received all the messages from
  every node in the rows above it. (These conventions are only for presentational
  purposes, and should not be mistaken with the \emph{perfectly synchronised round-based
    model} of \cite{DFGP07}, which we do not use.) The quorum $\{v_1,v_2,v_3\}$ sends
  $\VOTE(t,\false)$ to every node, which makes nodes $v_1$ and $v_2$ send
  $\READY(t,\false)$ to every node through lines~\ref{lin:ready}--\ref{lin:send-ready} of
  Algorithm~\ref{alg:broadcast}. However, there exists not a quorum $U$ such that
  $v_4\in U$ and every member of $U$ sends a message $\VOTE(t,a)$ with the same Boolean
  $a$, and thus node $v_4$ sends $\READY(t,\false)$ through
  lines~\ref{lin:bc-readied-v-blocking}--\ref{lin:send-v-blocking} of
  Algorithm~\ref{alg:broadcast}, only after receiving corresponding ready messages from
  the $v_4$-blocking set $\{v_1,v_2\}$. Observe how node $v_4$ changes its original vote
  $\true$ and sends $\false$ in the $\READY$ message. After every correct node receives
  $\READY(t,\false)$ from the quorum $\{v_1,v_2,v_4\}$, they all deliver $\false$.
\end{example}


\section{Abstract Stellar Consensus Protocol}
\label{sec:abstract-federated-consensus}

In this section we introduce the \emph{\abstractConsensus} (\emph{\AFCP}), which concisely
specifies the mechanism of SCP~\cite{M2015,MLG19} and highlights the modular structure
present in it\footnote{More precisely, in this paper we focus on Stellar's core {\em
    balloting} protocol, which aims to achieve consensus. We abstract from Stellar's {\em
    nomination} protocol---which tries to converge (best-effort) on a value to
  propose---by assuming arbitrary proposals to consensus.}. Like \emph{Paxos}
\cite{Lam98}, \AFCP\ uses {\em ballots}---pairs $\bllt{n}{x}$, where $n\in\Counters$ a
natural positive \emph{round number} and $x\in\Value$ a \emph{consensus value}. We assume
that $\Value$ is totally ordered, and we consider a special \emph{null ballot} $\nbllt$
and let $\Ballot = (\Nat^+\times \Value)\cup\{\nbllt\}$ be the set of ballots. (We write
$b.n$ and $b.x$ respectively for the round and consensus value of ballot $b$.) The set
$\Ballot$ is totally ordered, where we let $b < b'$ iff either $b.n < b'.n$, or
$b.n = b'.n$ and $b.\val < b'.\val$.

To better convey SCP's mechanism, we let the abstract protocol use FV as a black box where
nodes may hold a binary vote on each of the ballots: we let the set of voting values $\V$
be the set of Booleans and the set of tags $\Tag$ be the set of ballots, and let the
protocol consider a separate instance of FV for each ballot. A node voting for a Boolean
$a$ for a ballot $b$ that carries the consensus value $b.\val$ encodes the aim to either
\emph{abort} the ballot (when $a=\false$) or to \emph{commit} it (when $a = \true$) thus
deciding the consensus value $b.\val$. From now on we will unambiguously use `Booleans',
`ballots' and `values' instead of `voting values', `tags' and `consensus values',
respectively.

We have dubbed \AFCP\ `abstract' because, although it specifies the protocol concisely, it
is unsuited for realistic implementations. On the one hand, each node $v$ maintains
infinite state, because it stores a process $\federatedVoting(v,b)$ for each of the
infinitely many ballots~$b$ in the array $\fvs$ (line~\ref{lin:brs-init} of
Algorithm~\ref{alg:abstract-federated-consensus}). On the other hand, each node $v$ may
need to send or receive an infinite number of messages in order to progress
(lines~\ref{lin:send-echo}, \ref{lin:prepared}, \ref{lin:quorum-round} and
\ref{lin:prepare-increased} of Algorithm~\ref{alg:abstract-federated-consensus}, which are
explained in the detailed description of \AFCP\ below). This is done by assuming a
\emph{batched network semantics} (\emph{BNS}) in which the network exchanges
\emph{batches}, which are (possibly infinite) sequences of messages, instead of exchanging
individual messages: the sequence of messages to be sent by a node when processing an
event is batched per recipient, and each batch is sent at once after the atomic processing
of the event; once a batch is received, the recipient node atomically processes all the
messages in the batch in sequential order. By convention, we let the statement
$\mathbf{for all}$ in lines~\ref{lin:send-prepare} and \ref{lin:prepare-increased} of
Algorithm~\ref{alg:abstract-federated-consensus} consider the ballots $b'$ in ascending
ballot order. In \S\ref{sec:concrete-federated-consensus} we introduce a `concrete'
version of SCP that is amenable to implementation, since nodes in it maintain finite state
and exchange a finite number of messages; however, this version does not use FV as a black
box.

\AFCP\ uses the following \emph{below-and-incompatible-than} relation on ballots. We say
ballots~$b$ and $b'$ are \emph{compatible} (written $b \sim b'$) iff $b.\val = b'.\val$,
and \emph{incompatible} (written $b \not\sim b'$) otherwise, where we let $\bot \not= x$
for any $x\in \Value$. We say ballot $b$ is \emph{below and incompatible than} ballot $b'$
(written $b\lic b'$) iff $b<b'$ and $b\not\sim b'$.  In a nutshell, \AFCP\ works as
follows: each node uses FV to \emph{prepare} a ballot $b$ which carries the candidate
value $b.\val$, this is, it aborts every ballot $b'\lic b$, which prevents any attempt to
decide a value different from $b.\val$ at a round smaller than $b.\round$; once $b$ is
prepared, the node uses FV again to commit ballot $b$, thus deciding the candidate value
$b.\val$.

\begin{algo}[t]
  \begin{algorithm}[H]
    \setcounter{AlgoLine}{0}
    \SubAlgo{\Procc $\abstractconsensus(v\in\V)$}
    {
      \smallskip
      $\fvs \leftarrow [$\New \Procc $\federatedVoting(v,b)]^{b\in\Ballot}$\;
      \label{lin:brs-init}
      $\b, \h \leftarrow \nbllt \in \Ballot$\;\label{lin:candidate-highest-init}
      $\n\leftarrow 0\in \Nat^+ \cup \{0\}$\;\label{lin:current-round-init}
      \smallskip
      \SubAlgo{$\propose(x)$\label{lin:propose}}
      {
        $\b \leftarrow \bllt{1}{x}$\;\label{lin:init-b}
        \Forall $b' \lic \b$ \Do $\fvs[b'].\vote(\false)$\;\label{lin:send-prepare}
      }
      \smallskip
      \SubAlgo{\Upon \Triggered $\fvs[b'].\deliver(\false)$ \For \Every $b' \lic b$
        \And $\h< b$ \label{lin:prepared}}
      {
        $\h\leftarrow b$\;\label{lin:assign-h}
        \If{$\b \leq \h$\label{lin:h-reached-b}}
        {
          $\b\leftarrow \h$\;\label{lin:assign-b}
          $\fvs[\b].\vote(\true)$\;\label{lin:send-commit}
        }
      }
      \smallskip
      \SubAlgo{\Upon \Triggered $\fvs[b].\deliver(\true)$\label{lin:decided}}
      {
        \Trigger $\decide(b.\val)$;\label{lin:send-decided}
      }
      \smallskip
      \SubAlgo{\Upon \Exists $U\in \mathcal{Q}$ \SuchThat $v\in U$ \And
        \For \Each $u\in U$ \Exist $\M_u\in\{\VOTE,\READY\}$
        \And $b_u\in \Ballot$ \SuchThat $\n < b_u.\round$ \And \Either
        \Received $\M_u(b_u,\true)$ \From $u$ \Or \Received $\M_u(b',\false)$
        \From $u$ \For \Every $b'\in [z_u,b_u)$ \With $z_u< b_u$
        \label{lin:quorum-round}}
      {
        $\n \leftarrow \min\{b_u.\round\mid u \in U\}$\;\label{lin:update-current-round}
        $\starttimer(F(\n))$\;\label{lin:start-timer}
      }
      \smallskip
      \SubAlgo{\Upon \Triggered $\timeout$\label{lin:timeout}}
      {
        \lIf{$\h = \nbllt{}$}
        {$\b \leftarrow \bllt{\n + 1}{\b.\val}$\label{lin:increase-candidate}}
        \lElse{$\b \leftarrow \bllt{\n + 1}{\h.\val}$\label{lin:increase-prepared}}
        \Forall $b' \lic \b$ \Do $\fvs[b'].\vote(\false)$\;\label{lin:prepare-increased}
      }
    }
    \caption{\AbstractConsensus\ (\AFCP) over an FBQS $\Ss$ with set of quorums $\Qs$.}
    \label{alg:abstract-federated-consensus}
  \end{algorithm}
\end{algo}

\AFCP\ is depicted in Algorithm~\ref{alg:abstract-federated-consensus} above. We assume
that each node $v$ creates a process $\federatedVoting(v,b)$ for each ballot $b$, which is
stored in the infinite array $\fvs[b]$ (line~\ref{lin:brs-init}). The node keeps fields
$\b$ and $\h$, which respectively contain the ballot that $v$ is trying to commit and the
highest ballot prepared so far. Both $\b$ and $\h$ are initialised to the null ballot
(line~\ref{lin:candidate-highest-init}). The node also keeps a field $\n$ that contains
the current round, initialised to $0$ (line~\ref{lin:current-round-init}). Once $v$
proposes a value $x$, the node assigns the ballot $\bllt{1}{x}$ to $\b$ and tries to
prepare it by invoking FV's primitive $\vote(\false)$ for each ballot below and
incompatible than $\b$ (lines~\ref{lin:propose}--\ref{lin:send-prepare}). This may involve
sending an infinite number of messages, which by BNS requires sending finitely many
batches. Once $v$ prepares some ballot $b$ by receiving FV's indication $\deliver(\false)$
for every ballot below and incompatible than $b$, and if $b$ exceeds $\h$, the node
updates $\h$ to $b$ (lines~\ref{lin:prepared}--\ref{lin:assign-h}). The condition in
line~\ref{lin:prepared} may concern an infinite number of ballots, but it may hold after
receiving a finite number of batches by BNS. If $\h$ reaches or exceeds $\b$, then the
node updates $\b$ to $\h$, and tries to commit it by voting $\true$ for that ballot
(lines~\ref{lin:h-reached-b}--\ref{lin:send-commit}). Once $v$ commits some ballot $b$ by
receiving FV's indication $\deliver(\true)$ for ballot $b$, the node decides the value
$b.\val$ (lines~\ref{lin:decided}--\ref{lin:send-decided}) and stops execution.

If the candidate ballot of a node $v$ can no longer be aborted nor committed, then $v$
resorts to a timeout mechanism that we describe next. The primitive $\starttimer(\Delta)$
starts the node's local timer, such that a $\timeout$ event will be triggered once the
specified delay $\Delta$ has expired. (Invoking $\starttimer(\Delta')$ while the timer is
already running has the effect of restarting the timer with the new delay $\Delta'$.) In
order to start the timer, a node $v$ needs to receive, from each member of a quorum that
contains $v$ itself, messages that endorse either committing or preparing ballots with
rounds bigger than $\n$ (line~\ref{lin:quorum-round} of
Algorithm~\ref{alg:abstract-federated-consensus}). Since the domain of values can be
infinite, the condition in line~\ref{lin:quorum-round} requires that for each node $u$ in
some quorum $U$ that contains $v$ itself, there exists a ballot $b_u$ with round
$b_u.\round>\n$, and either $v$ receives from $u$ a message endorsing to commit $b_u$, or
otherwise $v$ receives from $u$ messages endorsing to abort every ballot in some
non-empty, right-open interval $[z_u,b_u)$, whose upper bound is $b_u$. This condition may
require receiving an infinite number of ballots, but it may hold after receiving a finite
number of batches by BNS. Once the condition in line~\ref{lin:quorum-round} holds, the
node updates $\n$ to the smallest $n$ such that every member of the quorum endorses to
either commit or prepare some ballot with round bigger or equal than $n$, and (re-)starts
the timer with delay $F(\n)$, %
where $F$ is an unbound function that doubles its value with each increment of $n$
(lines~\ref{lin:update-current-round}--\ref{lin:start-timer}). %
If the candidate ballot can no longer be aborted or committed, then $\timeout$ will be
eventually triggered (line~\ref{lin:timeout}) and the node considers a new candidate
ballot with the current round increased by one, and with the value $\b.\val$ if the node
never prepared any ballot yet (line~\ref{lin:increase-candidate}) or the value $\h.\val$
otherwise (line~\ref{lin:increase-prepared}). Then $v$ tries to prepare the new candidate
ballot by voting $\false$ for each ballot below and incompatible than it
(line~\ref{lin:prepare-increased}). This may involve sending an infinite number of
messages, which by BNS requires sending finitely many batches.

The condition for starting the timer in line~\ref{lin:quorum-round} does not strictly use
FV as a black box. However, this use is warranted because line~\ref{lin:quorum-round} only
`reads' the state of the network. \AFCP\ makes every other change to the network through
FV's primitives.

\AFCP\ guarantees the safety properties of \emph{non-blocking Byzantine consensus} in
\S\ref{sec:spec}. Since a node stops execution after deciding some value, \emph{Integrity
  for intact sets} holds trivially. The requirement in
lines~\ref{lin:prepared}--\ref{lin:send-commit} of
Algorithm~\ref{alg:abstract-federated-consensus} that a node prepares the candidate ballot
before voting for committing it, enforces that if a voting for committing some ballot
within the nodes of an intact set $I$ succeeds, then some node in $I$ previously prepared
that ballot:
\begin{lemma}\label{lem:prepared-before-commit-intact-set}
  Let $\Ss$ be an FBQS and consider an execution of \AFCP\ over $\Ss$. Let $I$ be an
  intact set in $\Ss$. If a node $v_1\in I$ commits a ballot $b$, then some node
  $v_2\in I$ prepared $b$.
\end{lemma}
Aborting every ballot below and incompatible than the candidate one prevents that one node
in an intact set $I$ prepares a ballot $b_1$, and concurrently another node in $I$ sends
$\READY(b_2,\true)$ with $b_2$ below and incompatible than $b_1$:
\begin{lemma}\label{lem:ready-commit-prepare-intact-set}
  Let $\Ss$ be an FBQS and consider an execution of \AFCP\ over $\Ss$. Let $I$ be an
  intact set in $\Ss$. Let $v_1$ and $v_2$ be nodes in $I$ and $b_1$ and $b_2$ be ballots
  such that $b_2\lic b_1$. The following two things cannot both happen: node $v_1$
  prepares $b_1$ and node $v_2$ sends $\READY(b_2,\true)$.
\end{lemma}
\emph{Agreement for intact sets} holds as follows: assume towards a contradiction that two
nodes in $I$ respectively commit ballots $b_1$ and $b_2$ with different values. A node in
$I$ prepared the bigger of the two ballots by
Lemma~\ref{lem:prepared-before-commit-intact-set}, which results in a contradiction by
Lemma~\ref{lem:ready-commit-prepare-intact-set}.

Lemma~\ref{lem:prepared-propose-intact-set} below ensures that in
line~\ref{lin:increase-prepared} it is safe to take as the new candidate value that of the
largest prepared ballot, which helps to establish \emph{Weak validity for intact sets}.
\begin{lemma}\label{lem:prepared-propose-intact-set}
  Let $\Ss$ be an FBQS and consider an execution of \AFCP\ over $\Ss$. Let $b_1$ be the
  largest ballot prepared by some node $v_1$ at some moment in the execution. If all nodes
  are honest, then some node $v_2$ proposed $b_1.\val$.
\end{lemma}

Now we examine the liveness properties of \emph{non-blocking Byzantine consensus} in
\S\ref{sec:spec}, which \AFCP\ also meets. Recall from \S\ref{sec:stellar-broadcast} the
bounded interval $\delta_I$ that a node in an intact set $I$ takes to deliver some Boolean
for a given ballot, provided that some other node in $I$ has already delivered a Boolean
for the same ballot. %
Let $v$ be a node in $I$ that prepares some ballot $b$ such that no other node in $I$ has
ever prepared a ballot with round bigger or equal than $b.\round$. We call the interval of
duration $\delta_I$ after $v$ prepares $b$ the \emph{window for intact set $I$ of round
  $b.\round$}. Lemma~\ref{lem:window-no-overlap} below guarantees that after some moment
in the execution, no two consecutive windows ever overlap.
\begin{lemma}\label{lem:window-no-overlap}
  Let $\Ss$ be an FBQS and consider an execution of \AFCP\ over $\Ss$. Let $I$ be an
  intact set in $\Ss$ and assume that all faulty nodes eventually stop. There exists a
  round $n$ such that either every node in $I$ decides some value before reaching round
  $n$, or otherwise the windows for $I$ of all the rounds $m \geq n$ never overlap with
  each other, and in each window of round $m$ the nodes in $I$ that have not decided yet
  only prepare ballots with round $m$.
\end{lemma}
\noindent Lemma~\ref{lem:window-no-overlap} helps to establish \emph{Non-blocking for
  intact sets} as follows. After the moment where no two consecutive windows overlap,
either every node in $I$ has the same candidate ballot at the beginning of the window of
some round, or otherwise the highest ballots prepared by each node in $I$ during that
window coincide with each other. In either case all the nodes in $I$ will eventually have
the same candidate ballot, and they will decide a value in bounded time.

Correctness of \AFCP\ is captured by
Theorem~\ref{thm:non-blocking-byzantine-consensu-intact-sets} below:
\begin{theorem}\label{thm:non-blocking-byzantine-consensu-intact-sets}
  Let $\Ss$ be an FBQS. The \AFCP\ protocol over $\Ss$ satisfies the specification of
  non-blocking Byzantine consensus for intact sets.
\end{theorem}


\section{Concrete Stellar Consensus Protocol}
\label{sec:concrete-federated-consensus}

In this section we introduce \emph{\concreteConsensus} (\emph{\CFCP}) which is amenable to
implementation because each node~$v$ maintains finite state and only needs to send and
receive a finite number of messages in order to progress. \CFCP\ relies on \emph{bunched
  voting} (\emph{BV}) in Algorithm~\ref{alg:bunch-voting}, which generalises FV and
embodies all of FV's instances for each of the ballots. \CFCP\ considers a single instance
of BV, and thus each node~$v$ keeps a single process $\bunchedVoting(v)$
(line~\ref{lin:pbd-brs-init} of Algorithm~\ref{alg:concrete-federated-consensus}).
In BV, nodes exchange messages that contain two kinds of statements: a \emph{prepare
  statement} $\PREP\ b$ encodes the aim to abort the possibly infinite range of ballots
that are lower and incompatible than $b$; and a \emph{commit statement} $\CMT\ b$ encodes
the aim to commit ballot $b$.


\begin{algo}[p]
\begin{algorithm}[H]
  \setcounter{AlgoLine}{0}
  \SubAlgo{\Procc $\bunchedVoting(v\in\V)$}
  {
    \smallskip
    $\he,\hr,\hd\leftarrow \nbllt \in \Ballot$\;\label{lin:h-init}
    $\Ce,\Cr,\Cd \leftarrow \emptyset \in  \powerset{\Ballot}$\;
    \label{lin:set-init}
    \smallskip
    \SubAlgo{$\prepare(b)$\label{lin:prepare}}
    {
      \If{$\he<b$}
      {
        $\he\leftarrow b$\;
        \Send $\VOTE(\PREP\ \he)$ \To \Every $v'\in\V$; \label{lin:send-vote-prepared}
      }
    }
    \smallskip
    \SubAlgo{\Upon \Exists {\bf maximum} $b$ \SuchThat $\he < b$ \And
      \Exists $U\in \mathcal{Q}$ \SuchThat $v\in U$ \And \For \Every $u\in U$
      \Received $\VOTE(\PREP\ b_u)$ \Where $b'\lic b_u$ \For \Every $b'\lic b$
      \label{lin:ready-prepare}}
    {
      $\hr \leftarrow b$ \label{lin:update-hr} \;
      \Send $\READY(\PREP\ \hr)$ \To \Every $v'\in\V$; \label{lin:send-ready-prepare}
    }
    \smallskip \SubAlgo{\Upon \Exists {\bf maximum} $b$ \SuchThat $\hr<b$ \And
      \Exists {\bf $v$-blocking} $B$ \SuchThat
      \For \Every  $u\in B$ \Received $\READY(\PREP\ b_u)$
      \Where $b'\lic b_u$ \For \Every $b'\lic b$\label{lin:ready-prepare-v-blocking}}
    {
      $\hr \leftarrow b$ \label{lin:update-hr-v-blocking} \;
      \Send $\READY(\PREP\ \hr)$ \To \Every $v'\in\V$;
      \label{lin:send-ready-prepare-v-blocking}
    }
    \smallskip
    \SubAlgo{\Upon \Exists {\bf maximum} $b$ \SuchThat $\hd<b$ \And
      \Exists $U\in\mathcal{Q}$ \SuchThat $v\in U$ \And \For \Every $u\in U$
      \Received $\READY(\PREP\ b_u)$ \Where $b'\lic b_u$ \For \Every $b'\lic b$
      \label{lin:b-prepared}}
    {
      $\hd \leftarrow b$\;\label{lin:hd}
      $\prepared(\hd)$;\label{lin:b-send-prepared}
    }
    \smallskip
    \SubAlgo{$\commit(b)$\label{lin:commit}}
    {
      \If{$b\not\in\Ce$ \And $\he = b$\label{lin:condition-commit}}
      {
        $\Ce \leftarrow \Ce \cup \{b\}$\;
        \Send $\VOTE(\CMT\ b)$ \To \Every $v'\in\V$ \label{lin:send-vote-commit};
      }
    }
    \smallskip
    \SubAlgo{\Upon \Received $\VOTE(\CMT\ b)$ \From \Every $u\in U$
      \For \Some $U\in \mathcal{Q}$ \SuchThat $v\in U$ \And $b\not\in\Cr$
      \label{lin:ready-commit}}
    {
      $\Cr \leftarrow \Cr \cup \{b\}$ \label{lin:readied} \;
      \Send $\READY(\CMT\ b)$ \To \Every $v'\in\V$ \label{lin:send-ready-commit};
    }
    \smallskip
    \SubAlgo{\Upon \Received
      $\READY(\CMT\ b)$ \From \Every $u\in B$ \For \Some
      {\bf $v$-blocking} $B$ \And $b\not\in\Cr$
       \label{lin:ready-commit-v-blocking}}
    {
      $\Cr \leftarrow \Cr \cup \{b\}$ \label{lin:readied-v-blocking} \;
      \Send $\READY(\CMT\ b)$ \To \Every $v'\in\V$
      \label{lin:send-ready-commit-v-blocking};
    }
    \smallskip
    \SubAlgo{\Upon \Received
      $\READY(\CMT\ b)$ \From \Every $u\in U$ \For \Some $U\in \mathcal{Q}$
      \SuchThat $v\in U$ \And $b\not\in\Cd$ \label{lin:committed}}
    {
      $\Cd \leftarrow \Cd \cup \{b\}$\label{lin:delivered}\;
      $\committed(b)$ \label{lin:send-committed};
    }
  }
  \caption{Bunched voting (BV) over an FBQS $\Ss$ with set of quorums  $\Qs$.}
  \label{alg:bunch-voting}
\end{algorithm}
\end{algo}
Algorithm~\ref{alg:bunch-voting} depicts BV over an FBQS $\Ss$ with set of quorums $\Qs$. A
node~$v$ stores the highest ballot for which $v$ has respectively voted, readied, or
delivered a prepare statement in fields $\he$, $\hr$, and $\hd$
(line~\ref{lin:h-init}). It also stores the set of ballots for which $v$ has respectively
voted, readied, or delivered a commit statement in fields $\Ce$, $\Cr$, and $\Cd$
(line~\ref{lin:set-init}). All these fields are finite and thus $v$ maintains only finite
state. %
When a node $v$ invokes $\prepare(b)$, if $b$ exceeds the highest ballot for which $v$ has
voted a prepare, then the node updates $\he$ to $b$ and sends $\VOTE(\PREP\ b)$ to every
other node (lines~\ref{lin:prepare}--\ref{lin:send-vote-prepared}). The protocol then
proceeds with the usual stages of FV, with the caveat that at each stage of the protocol
only the maximum ballot is considered for which the node can send a message---or deliver
an indication---with a prepare statement. In particular, when there exists a ballot~$b$
that exceeds $\hr$ and such that $v$ received a message $\VOTE(\PREP\ b_u)$ from each
member $u$ of some quorum to which $v$ belongs, then the node proceeds as follows: it
checks that each $b'$ lower and incompatible than $b_u$ is also lower and incompatible
than~$b$ (line~\ref{lin:ready-prepare}). If $b$ is the maximum ballot passing the previous
check for every member $u$ of the quorum, then the node updates the field $\hr$ to $b$ and
sends the message $\READY(\PREP\ b)$ to every other node
(lines~{\ref{lin:update-hr}--\ref{lin:send-ready-prepare}}). The node $v$ checks similar
conditions for the case when it receives messages $\READY(\PREP\ b_u)$ from each member
$u$ of a $v$-blocking set, and proceeds similarly by updating $\hr$ to $b$ and sending
$\READY(\PREP\ b)$ to every other node
(lines~\ref{lin:ready-prepare-v-blocking}--\ref{lin:send-ready-prepare-v-blocking}). The
node will update $\hd$ and trigger the indication $\prepared(b)$ when the same conditions
are met after receiving messages $\READY(\PREP\ b_u)$ from each member $u$ of a quorum to
which $v$ belongs (lines~\ref{lin:b-prepared}--\ref{lin:b-send-prepared}).
%
%
When a node $v$ invokes $\commit(b)$ then the protocol proceeds with the usual stages of
FV with two minor differences (lines~\ref{lin:commit}--\ref{lin:send-committed}). First, a
node $v$ only votes commit for the highest ballot for which $v$ has voted a prepare
statement (condition $\he = b$ in line~\ref{lin:condition-commit}). Second, the protocol
uses the sets of ballots $\Ce$, $\Cr$ and $\Cd$ in order to keep track of the stage of the
protocol for each ballot.

The structure of \CFCP\ in Algorithm~\ref{alg:concrete-federated-consensus} directly
relates to \AFCP\ in Algorithm~\ref{alg:abstract-federated-consensus}. A node proposes a
value~$x$ in line~\ref{lin:pbd-propose}. A node tries to prepare a ballot~$b$ by invoking
$\prepare(b)$ in line~\ref{lin:pbd-send-prepare}, and receives the indication
$\prepared(b)$ in line~\ref{lin:pbd-prepared}. A node tries to commit a ballot $b$ by
invoking $\commit(b)$ in line~\ref{lin:pbd-send-commit}, and receives the indication
$\committed(b)$ in line~\ref{lin:pbd-decided}. A node decides a value $x$ in
line~\ref{lin:pbd-send-decided}. Timeouts are set in
lines~\ref{lin:pbd-quorum-round}--\ref{lin:pbd-start-timer} and triggered in
line~\ref{lin:pbd-timeout}.
\begin{algo}[h]
  \begin{algorithm}[H]
    \setcounter{AlgoLine}{0}
    \SubAlgo{\Procc $\concreteconsensus(v\in\V)$}
    {
      \smallskip
      $\brs \leftarrow$ \New \Procc $\bunchedVoting(v)$\;
      \label{lin:pbd-brs-init}
      $\b, \h \leftarrow \nbllt \in \Ballot$\;\label{lin:pbd-candidate-highest-init}
      $\n\leftarrow 0\in \Counters \cup \{0\}$\;\label{lin:pbd-current-round-init}
      \smallskip
      \SubAlgo{$\propose(x)$\label{lin:pbd-propose}}
      {
        $\b \leftarrow \bllt{1}{x}$\;\label{lin:pbd-init-b}
        $\brs.\prepare(\b)$\;\label{lin:pbd-send-prepare}
      }
      \smallskip
      \SubAlgo{\Upon \Triggered $\brs.\prepared(b)$ \And $\h< b$ \label{lin:pbd-prepared}}
      {
        $\h\leftarrow b$\;\label{lin:pbd-assign-h}
        \If{$\b \leq \h$\label{lin:pbd-h-reached-b}}
        {
          $\b\leftarrow \h$\;\label{lin:pbd-assign-b}
          $\brs.\commit(\b)$\;\label{lin:pbd-send-commit}
        }
      }
      \smallskip
      \SubAlgo{\Upon \Triggered $\brs.\committed(b)$\label{lin:pbd-decided}}
      {
        \Trigger $\decide(b.\val)$;\label{lin:pbd-send-decided}
      }
      \smallskip
      \SubAlgo{\Upon \Exists $U\in\mathcal{Q}$ \SuchThat $v\in U$
        \And \For \Each $u \in U$ \Exist $\M_u\in\{\VOTE,\READY\}$
        \And $b_u \in \Ballot$ \SuchThat $\n < b_u.\round$
        \And \Received $~\M_u(\STA_u\ b_u)$ \From $u$ \With
        $\STA_u\in \{\CMT,\PREP\}$
        \label{lin:pbd-quorum-round}}
      {
        $\n \leftarrow \min\{b_u.\round \mid u\in U\}$\;\label{lin:pbd-update-current-round}
        $\starttimer(F(\n))$\;\label{lin:pbd-start-timer}
      }
      \smallskip
      \SubAlgo{\Upon \Triggered $\timeout$\label{lin:pbd-timeout}}
      {
        \lIf{$\h = \nbllt{}$}
        {$\b \leftarrow \bllt{\n + 1}{\b.\val}$\label{lin:pbd-increase-candidate}}
        \lElse{$\b \leftarrow \bllt{\n + 1}{\h.\val}$\label{lin:pbd-increase-round}}

        $\brs.\prepare(\b)$\;\label{lin:pbd-prepare-increased}
      }
    }
  \caption{\ConcreteConsensus\ (\CFCP) over an FBQS $\Ss$ with set of quorums $\Qs$.}
  \label{alg:concrete-federated-consensus}
  \end{algorithm}
\end{algo}%

Next we establish a correspondence between \CFCP\ in and \AFCP\ in
\S\ref{sec:abstract-federated-consensus}: the concrete protocol observationally refines
the abstract one, which means that any externally observable behaviour of the former can
also be produced by the latter \cite{FORY10}. Informally, the refinement shows that for an
FBQS~$\Ss$ and an intact set~$I$, for every execution of \CFCP\ over $\Ss$ there exists an
execution of \AFCP\ over $\Ss$ (with some behaviour of faulty nodes) such that each node
in $I$ decides the same value in both of the executions. The refinement result allows us
to carry over the correctness of \AFCP\ established in
\S\ref{sec:abstract-federated-consensus} to \CFCP.

We first define several notions required to formalise our refinement result.
%
A \emph{history} is a sequence of the events $v.\propose(x)$ and $v.\decide(x)$, where $v$
is a correct node and $x$ a value.
%
The specification of consensus assumes that $v$ triggers an event $v.\propose(x)$, thus a
history will have $v.\propose(x)$ for every correct node~$v$.
A \emph{concrete trace} $\tau$ is a sequence of events that subsumes histories, and
contains events $v.\prepare(b)$, $v.\commit(b)$, $v.\prepared(b)$, $v.\committed(b)$,
$v.\starttimer(n)$, $v.\timeout$, $v.\send(m, v')$, and $v.\receive(m, v')$, where $v$ is
a correct node and $v'$ is any node, $b$ is a ballot, $m$ is a message in
$\{\VOTE(s),\READY(s)\}$ with $s$ a statement in $\{\PREP\ b,\CMT\ b\}$, and $n$ is a
round.
An \emph{abstract trace} $\tau$ is a sequence of events that subsumes histories, and
contains events $v.\starttimer(n)$, $v.\timeout$, and batched events
$v.\voteBatch([b_i], a)$, $v.\deliverBatch([b_i], a)$, $v.\sendBatch([m_i], v')$, and
$v.\receiveBatch([m_i], v')$,%
where $v$ is a correct node and $v'$ is any node, $n$ is a round, $[b_i]$ is a sequence of
ballots, $a$ is a Boolean, and $[m_i]$ is a sequence of messages in
$\{\VOTE(b, a),\READY(b,a)\}$. The sequences of ballots and messages above, which
represent a possibly infinite number of `batched' events, ensure that the length of any
abstract trace is bounded by $\omega$.
We may omit the adjective `concrete/abstract' from `trace' when it is clear from the
context. Given a trace $\tau$, a history $\hist(\tau)$ can be uniquely obtained from
$\tau$ by removing every event in $\tau$ different from $v.\propose(x)$ or $v.\decide(x)$.

An execution of \CFCP\ (respectively, \AFCP) \emph{entails} a \emph{concrete trace}
(respectively, \emph{abstract trace})~$\tau$ iff for every invocation and indication as
well as for every send or receive primitive in an execution of the protocol in
Algorithm~\ref{alg:concrete-federated-consensus} (respectively, for every invocation,
indication and primitive in an execution of the protocol in
Algorithm~\ref{alg:abstract-federated-consensus}, where the $\vote$, $\deliver$, $\send$
and $\receive$ events are batched together), $\tau$ contains corresponding events in the
same order.

We are interested in traces that are relative to some intact set $I$. Given a trace
$\tau$, the \emph{$I$-projected} trace $\tau|_I$ is obtained by removing the events
$v.{\sf ev}\in\tau$ such that $v\not\in I$.
\begin{theorem}\label{thm:refinement}
  Let $\Ss$ be an FBQS and $I$ be an intact set. For every execution of \CFCP\ over $\Ss$
  with trace~$\tau$, there exists an execution of \AFCP\ over $\Ss$ with trace $\rho$ and
  $\hist(\tau|_I) = \hist(\rho|_I)$.
\end{theorem}
\begin{proof}[Proof sketch]
  We define a \emph{simulation function}~$\maptrace$ from concrete to abstract
  traces. Theorem~\ref{thm:refinement} can be established by showing that, for every
  finite prefix~$\tau$ of a trace entailed by \CFCP, the simulation $\maptrace(\tau)$ is a
  prefix of a trace entailed by \AFCP.
\end{proof}

Every execution of \AFCP\ enjoys the properties of \emph{Integrity}, \emph{Agreement for
  intact sets}, \emph{Weak validity for intact sets} and \emph{Non-blocking for intact
  sets}, and so does every execution of \CFCP\ by refinement. %
%
%
\begin{corollary}\label{cor:cscp-correct}
  Let $\Ss$ be an FBQS. The \CFCP\ protocol over $\Ss$ satisfies the specification of
  non-blocking Byzantine consensus for intact sets.
\end{corollary}


\section{Lying about Quorum Slices}
\label{sec:lying}

So far we have assumed the unrealistic setting where faulty nodes do not equivocate their
quorum slices, so all nodes share the same FBQS $\Ss$. We now lift this assumption. To
this end, we use a generalisation of FBQS called \emph{subjective FBQS}~\cite{GG18}, which
allows faulty nodes to lie about their quorum slices. Assuming that $\V_\ok$ is the set of
correct nodes, the \emph{subjective FBQS} $\{\Ss_v\}_{v\in \V_{\ok}}$ is an indexed family
of FBQSes where the different FBQSes agree on the quorum slices of correct nodes, \ie,
$\forall v_1,v_2,v\in \V_\ok.~\Ss_{v_1}(v)=\Ss_{v_2}(v)$. For each correct node~$v$, the
FBQS $\Ss_v$ is the \emph{view} of node~$v$, which reflects the choices of trust
communicated to $v$. We can run either \AFCP\ or \CFCP\ over a subjective FBQS
$\{\Ss_v\}_{v\in\V_{\ok}}$ by letting each correct node~$v$ act according to its view
$\Ss_v$.

We generalise the definition of intact set to subjective FBQSes, and we lift our results
so far to the subjective FBQSes. Let $\{\Ss_v\}_{v\in \V_{\ok}}$ be a subjective FBQS. A
set $I$ is an \emph{intact set} iff for each $v\in \V_{\ok}$ the set $I$ is a quorum in
$\Ss_v$ that only contains correct nodes, and every member of $I$ is intertwined with each
other in the projected FBQS $\Ss_v|_I$.

\begin{lemma}\label{lem:subjective-intact-set}
  Let $\{\Ss_v\}_{v\in \V_{\ok}}$ be a subjective FBQS. For any node~$v\in \V_\ok$, a
  set $I$ is an intact set in $\Ss_v$ iff $I$ is an intact set in
  $\{\Ss_v\}_{v\in \V_{\ok}}$.
\end{lemma}
\noindent Since Lemma~\ref{lem:subjective-intact-set} above guarantees that every view has
the same intact sets, which also coincide with the intact sets of the subjective FBQS,
from now on we may say `an intact set~$I$' and omit to which system (a particular view, or
the subjective FBQS) $I$ belongs .

Using the fact that nodes agree on the slices of correct nodes, we can prove
Lemma~\ref{lem:subjective-quorum-intersection} below, which is the analogue to
Lemma~\ref{lem:quorums-intersect-intact-set} and states sufficient safety conditions for
the nodes in an intact set $I$ to reach agreement when each node acts according to its own
view.
\begin{lemma}\label{lem:subjective-quorum-intersection}
  Let $\{\Ss_v\}_{v\in \V_{\ok}}$ be a subjective FBQS and for each correct node $v$ let
  $\Qs_v$ be the set of quorums in the view $\Ss_v$. Let $I$ be an intact set and consider
  two quorums $U_1$ and $U_2$ in $\bigcup_{v\in\V_\ok}\Qs_v$. If
  $U_1\cap I\not= \emptyset$ and $U_2\cap I\not=\emptyset$, then
  $U_1\cap U_2\cap I\not= \emptyset$.
\end{lemma}

Using arguments similar to those in the previous sections, we can establish the
correctness of \AFCP\ and \CFCP\ over subjective FBQSes.

\begin{theorem}\label{thm:non-blocking-byzantine-consensus-intact-sets-lying}
  Let $\{\Ss_v\}_{v\in \V_{\ok}}$ be a subjective FBQS. The \AFCP\ protocol over
  $\{\Ss_v\}_{v\in \V_{\ok}}$ satisfies the specification of non-blocking Byzantine
  consensus for intact sets.
\end{theorem}

\begin{theorem}\label{thm:refinement-lying}
  Let $\{\Ss_v\}_{v\in \V_{\ok}}$ be a subjective FBQS and $I$ be an intact set. For every
  execution of \CFCP\ over $\{\Ss_v\}_{v\in \V_{\ok}}$ with trace~$\tau$, there exists an
  execution of \AFCP\ over $\{\Ss_v\}_{v\in \V_{\ok}}$ with trace $\rho$ and
  $\hist(\tau|_I) = \hist(\rho|_I)$.
\end{theorem}

\begin{corollary}\label{cor:cscp-correct-lying}
  Let $\{\Ss_v\}_{v\in \V_{\ok}}$ be a subjective FBQS. The \CFCP\ protocol over
  $\{\Ss_v\}_{v\in \V_{\ok}}$ satisfies the specification of non-blocking Byzantine
  consensus for intact sets.
\end{corollary}


\section{Related Work}\label{sec:related-future-work}

García-Pérez and Gotsman~\cite{GG18} have previously investigated Stellar's federated
voting and its relationship to Bracha's broadcast over classical Byzantine quorum
systems. They did not address the full Stellar consensus protocol. Our proof of SCP
establishes the correctness of federated voting by adjusting the results in~\cite{GG18} to
multiple intact sets within the system.

Losa et al.~\cite{GLM19} have also investigated consensus over FBQSs. They propose a
generalisation of Stellar's quorums that does not prescribe constructing them from slices,
yet allows different participants to disagree on what constitutes a quorum. They then
propose a protocol solving consensus over intact sets in this setting that provides better
liveness guarantees than SCP, but is impractical. Losa et al.'s work is orthogonal to
ours: they consider a more general setting than Stellar's and a theoretical protocol,
whereas we investigate the practical protocol used by Stellar.

The advent of blockchain has given rise to a number of novel proposals of BFT protocols;
see~\cite{CV17} for a survey. Out of these, the most similar one to Stellar is
Ripple~\cite{SYB14}. In particular, Ripple have recently proposed a protocol called Cobalt
that allows for a federated setting similar to Stellar's~\cite{McB18}. We hope that our
work will pave the way to investigating the correctness of this and similar protocols.

\subparagraph*{Acknowledgements.} We thank Ilya Sergey, Giuliano Losa and Alexey Gotsman
for their comments. This work was supported by ERC Starting Grant RACCOON, Spanish State
Research Agency project BOSCO, and Madrid Regional Government project BLOQUES.


\bibliography{paper}

\begin{thebibliography}{10}

\bibitem{Bra87}
Gabriel Bracha.
\newblock Asynchronous {B}yzantine agreement protocols.
\newblock {\em Information and Computation}, 75(2):130--143, 1987.

\bibitem{CGR11}
Christian Cachin, Rachid Guerraoui, and Lu{\'i}s E.~T. Rodrigues.
\newblock {\em Introduction to Reliable and Secure Distributed Programming
  (2.~ed.)}.
\newblock Springer, 2011.

\bibitem{CV17}
Christian Cachin and Marko Vukolic.
\newblock Blockchain consensus protocols in the wild.
\newblock In {\em International Symposium on Distributed Computing (DISC)},
  2017.

\bibitem{CL02}
Miguel Castro and Barbara Liskov.
\newblock Practical {B}yzantine fault tolerance and proactive recovery.
\newblock {\em ACM Transactions on Computer Systems}, 20(4):398--461, 2002.

\bibitem{DFGP07}
Carole Delporte-Gallet, Hugues Fauconnier, Rachid Guerraoui, and Bastian
  Pochon.
\newblock The perfectly synchronized round-based model of distributed
  computing.
\newblock {\em Information and Computation}, 205(5):783--815, 2007.

\bibitem{FORY10}
Ivana Filipovic, Peter~W. O'Hearn, Noam Rinetzky, and Hongseok Yang.
\newblock Abstraction for concurrent objects.
\newblock {\em Theoretical Computer Science}, 411(51-52):4379--4398, 2010.

\bibitem{GG18}
{\'A}lvaro García-Pérez and Alexey Gotsman.
\newblock Federated {B}yzantine quorum systems.
\newblock In {\em International Conference on Principles of Distributed Systems
  (OPODIS)}, volume 125 of {\em LIPIcs}, pages 17:1--17:16. Schloss Dagstuhl,
  2018.

\bibitem{GG18b}
{\'A}lvaro García-Pérez and Alexey Gotsman.
\newblock Federated {B}yzantine quorum systems (extended version).
\newblock {\em CoRR}, abs/1811.03642, 2018.
\newblock \url{http://arxiv.org/abs/1811.03642}.

\bibitem{Lam98}
Leslie Lamport.
\newblock The part-time parliament.
\newblock {\em ACM Transactions on Computer Systems}, 16(2):133--169, 1998.

\bibitem{GLM19}
Giuliano Losa, Eli Gafni, and David Mazi{\`e}res.
\newblock {S}tellar consensus by instantiation.
\newblock In {\em 33rd International Symposium on Distributed Computing (DISC
  2019)}, volume 146 of {\em LIPIcs}, pages 27:1--27:15. Schloss Dagstuhl,
  2019.

\bibitem{McB18}
Ethan MacBrough.
\newblock Cobalt: {BFT} governance in open networks.
\newblock {\em CoRR}, abs/1802.07240, 2018.
\newblock URL: \url{http://arxiv.org/abs/1802.07240}.

\bibitem{MR98}
Dahlia Malkhi and Michael~K. Reiter.
\newblock {B}yzantine quorum systems.
\newblock {\em Distributed Computing}, 11(4):203--213, 1998.

\bibitem{M2015}
David Mazières.
\newblock The {Stellar} consensus protocol: a federated model for
  internet-level consensus, 2015.
\newblock \url{https://www.stellar.org/papers/stellar-consensus-protocol.pdf}.

\bibitem{MLG19}
David Mazières, Giuliano Losa, and Eli Gafni.
\newblock Simplifed {SCP}, 2019.
\newblock \url{http://www.scs.stanford.edu/~dm/blog/simplified-scp.html}.

\bibitem{SYB14}
David Schwartz, Noah Youngs, and Arthur Britto.
\newblock The {Ripple} protocol consensus algorithm, 2014.
\newblock \url{https://ripple.com/files/ripple_consensus_whitepaper.pdf}.

\bibitem{SRSD08}
Yee~Jiun Song, Robbert van Renesse, Fred~B. Schneider, and Danny Dolev.
\newblock The building blocks of consensus.
\newblock In {\em International Conference on Distributed Computing and
  Networking (ICDCN)}, 2008.

\end{thebibliography}

\appendix

\section{Proofs in \S\ref{sec:background}}
\label{ap:background}

\begin{lemma}[Lemma~33 in \cite{GG18b}]
  \label{lem:quorum-project}
  Let $U$ be a quorum in an FBQS $\Ss$, let $I$ be a set of nodes, and let $U'=U\cap
  I$. If $U'\not=\emptyset$ then $U'$ is a quorum in $\Ss|_I$.
\end{lemma}
\begin{proof}
  Straightforward by the definition of the projection operation.
\end{proof}

\begin{lemma}
  \label{lem:quorum-intersection}
  Let $\Ss$ be an FBQS and assume some set of faulty nodes. Let $I$ be an intact set in
  $\Ss$. Every two quorums in $\Ss|_I$ have non-empty intersection.
\end{lemma}
\begin{proof}
  Straightforward, since every node in $\Ss|_I$ is intertwined with each other by
  definition of intact set.
\end{proof}

\begin{proof}[Proof of Lemma~\ref{lem:quorums-intersect-intact-set}]
  $U_1\cap I$ and $U_2\cap I$ are quorums in $\Ss|_I$ by
  Lemma~\ref{lem:quorum-project}. Since every two quorums in $\Ss|_I$ have non-empty
  intersection by Lemma~\ref{lem:quorum-intersection}, we have
  $(U_1\cap I)\cap(U_2\cap I) = (U_1\cap U_2)\cap I \not= \emptyset$. Therefore the
  intersection $U_1\cap U_2$ contains some node in $I$.
\end{proof}

\begin{proof}[Proof of Lemma~\ref{lem:intact-sets-closed-union}]
  Assume $\V_\ok$ is the set of correct nodes, and let $I_1$ and $I_2$ be intact sets such
  that $I_1\cap I_2\not=\emptyset$. Assume towards a contradiction that $I_1\cup I_2$ is
  not an intact set, and therefore there exist nodes $v_1\in I_1$ and $v_2\in I_2$ that
  are not intertwined. By definition of intertwined, this assumption entails that there
  exist quorums $U_1$ and $U_2$ such that $v_1\in U_1$ and $v_2\in U_2$, and
  $(U_1\cap U_2) \cap \V_\ok = \emptyset$. Since both $U_1$ and $I_2$ are quorums in $\Ss$
  and both have non-empty intersection with the intact set $I_1$, we have
  $(U_1 \cap I_2) \cap I_1 \not= \emptyset$ by
  Lemma~\ref{lem:quorums-intersect-intact-set}, and we can conclude that $U_1$ has
  non-empty intersection with the intact set $I_2$.  Thus, we know that
  $(U_1\cap U_2) \cap I_2 \not= \emptyset$ again by
  Lemma~\ref{lem:quorums-intersect-intact-set}. But this results in a contradiction
  because $(U_1\cap U_2)\cap I_2=\emptyset$ by assumptions, since $I_2$ contains only
  correct nodes.
\end{proof}


\section{Proofs in \S\ref{sec:stellar-broadcast}}
\label{ap:stellar-broadcast}

\begin{proof}[Proof of Lemma~\ref{lem:empty-I-not-v-blocking}]
  Since $I$ is a quorum in $\Ss$ and by the definition of quorum, for every node $v\in I$
  there exists one slice of $v$ that lies within $I$, and the required holds.
\end{proof}

\begin{lemma}[Analogous to Lemma~36 in \cite{GG18b} for intact sets]
  \label{lem:federated-intact-set-ready}
  Let $\Ss$ be an FBQS and $t$ be a tag, and consider an execution of the instance for $t$
  of FV over $\Ss$. Let $I$ be an intact set in $\Ss$. The first node $v\in I$ that sends
  a $\READY(t,a)$ message first needs to receive a $\VOTE(t,a)$ message from every member
  of a quorum $U$ to which $v$ belongs.
\end{lemma}
\begin{proof}
  Let $v$ be any node in $I$. By Lemma~\ref{lem:empty-I-not-v-blocking} no $v$-blocking
  set $B$ exists such that $B\cap I=\emptyset$. Therefore, the first node $v\in I$ that
  sends a $\READY(t,a)$ message does it through
  lines~\ref{lin:ready}--\ref{lin:send-ready} of Algorithm~\ref{alg:broadcast}, which
  means that $v$ received $\VOTE(t,a)$ messages from every member of a quorum $U$ to which
  $v$ belongs.
\end{proof}

\begin{lemma}
  \label{lem:federated-not-v-blocking-availability-intact-set}
  Let $\Ss$ be an FBQS and assume some set of faulty nodes. Let $I$ be an intact set in
  $\Ss$, and consider a set $B$ of nodes. If $B$ is not $v$-blocking for any
  $v\in I \setminus B$, then either $B\supseteq I$ or $I \setminus B$ is a quorum in
  $\Ss|_I$.
\end{lemma}
\begin{proof}
  Assume $B$ is not $v$-blocking for any $v\in I \setminus B$. If $B \supseteq I$ then we
  are done. Otherwise, for every node $v$ in $I \setminus B$, there exists a slice
  $q\in\Ss(v)$ such that $q\cap B=\emptyset$. We know that $q\cap I\not=\emptyset$ since
  $v\in q$ by definition of FBQS. We also know that $q\cap I \in \Ss|_I(v)$ by definition
  of $\Ss|_I$, and since $q\cap B=\emptyset$, the intersection $q\cap I$ is a subset of
  $I\setminus B$. Since for each node $v\in I$ there exists a slice $q\in \Ss(v)$ such
  that $q\cap I$ is a subset of $I\setminus B$, the set $I \setminus B$ is a quorum in
  $\Ss|_I$, as required.
\end{proof}

\begin{lemma}[Analogous to Lemma~16 in \cite{GG18b} for intact sets]
  \label{lem:federated-ready-consistent-intact-set}
  Let $\Ss$ be an FBQS and $t$ be a tag, and consider an execution of the instance for $t$
  of FV over $\Ss$. Let $I$ be an intact set in $\Ss$. If two nodes in $I$ send
  respectively messages $\READY(t,a)$ and $\READY(t,a')$, then $a = a'$.
\end{lemma}
\begin{proof}
  Assume that two nodes in $I$ send respectively messages $\READY(t,a)$ and
  $\READY(t,a')$. By Lemma~\ref{lem:federated-intact-set-ready}, some node $v\in I$ has
  received $\VOTE(t,a)$ from a quorum $U$ to which $v$ belongs, and some node $v'\in I$
  has received $\VOTE(t,a')$ from a quorum $U'$ to which $v'$ belongs. By
  Lemma~\ref{lem:quorums-intersect-intact-set}, the intersection $U\cap U'$ contains some
  node in $I$, so that this node has sent $\VOTE(t,a)$ and $\VOTE(t,a')$. But due to the
  use of the guard variable $\voted$ in lines~\ref{lin:bc-init} and
  \ref{lin:if-echoed}--\ref{lin:echoed} of Algorithm~\ref{alg:broadcast}, a node can only
  vote for one value per tag, and thus it cannot vote different values for the same
  tag. Hence, $a=a'$.
\end{proof}

\begin{lemma}[Analogous to Lemma~17 in \cite{GG18b} for intact sets]
  \label{lem:federated-totality-intact-set}
  Let $\Ss$ be an FBQS and assume some set of faulty nodes. Let $I$ be an intact set in
  $\Ss$. Assume that $I = V^+ \uplus V^-$ and for some quorum $U$ we have
  $U \cap I \subseteq V^+$. Then either $V^-=\emptyset$ or there exists some node
  $v\in V^-$ such that $V^+$ is $v$-blocking.
\end{lemma}
\begin{proof}
  Assume that $V^+$ is not $v$-blocking for any $v\in V^-$. By
  Lemma~\ref{lem:federated-not-v-blocking-availability-intact-set}, either $V^-=\emptyset$
  or $V^-$ is a quorum in $\Ss|_I$. In the former case we are done, while in the latter we
  get a contradiction as follows. By Lemma~\ref{lem:quorum-project}, the intersection
  $U\cap I$ is a quorum in $\Ss|_I$. Since every two quorums in $\Ss|_I$ have non-empty
  intersection by Lemma~\ref{lem:quorum-intersection}, we have
  $(U\cap I)\cap V^-\not=\emptyset$. But this is impossible, since $U\cap I \subseteq V^+$
  and $V^+\cap V^-=\emptyset$.
\end{proof}

\begin{proof}[Proof of Theorem~\ref{thm:reliable-byzantine-broadcast-intact-sets}]
  We prove that the instance for tag $t$ of FV over $\Ss$ enjoys each of the properties
  that define the specification of \emph{reliable Byzantine voting for intact sets}.
  \begin{description}
  \item \emph{No duplication:} Straightforward by the use of the guard variable
    $\delivered$ in line~\ref{lin:delivered} of Algorithm~\ref{alg:broadcast}.
  \item \emph{Totality for intact sets:} Assume some node in $I$ delivers a value $a$ for
    tag $t$. By the condition in line~\ref{lin:ready} of Algorithm~\ref{alg:broadcast},
    the node has received $\READY(t,a)$ messages from every member in a quorum $U$. Since
    $U\cap I$ contains only correct nodes, these nodes send $\READY(t,a)$ messages to
    every node. By the condition in line~\ref{lin:v-blocking} of
    Algorithm~\ref{alg:broadcast}, any correct node $v$ sends $\READY(t,a)$ messages if it
    receives $\READY(t,a)$ from every member in a $v$-blocking set. Hence, the
    $\READY(t,a)$ messages from the nodes in $U\cap I$ may convince additional correct
    nodes to send $\READY(t,a)$ messages to every node. Let these additional correct nodes
    send $\READY(t,a)$ messages until a point is reached at which no further nodes in $I$
    can send $\READY(t,a)$ messages. At this point, let $V^+$ be the set of nodes in $I$
    that sent $\READY(t,a)$ messages (where $U\cap I\subseteq V^+$), and let
    $V^- = I \setminus V^+$. By Lemma~\ref{lem:federated-ready-consistent-intact-set} the
    nodes in $V^-$ did not send any $\READY(t,\_)$ messages at all. The set $V^+$ cannot
    be $v$-blocking for any node $v$ in $V^-$, or else more nodes in $I$ could come to
    send $\READY(t,a)$ messages. Then by Lemma~\ref{lem:federated-totality-intact-set} we
    have $V^-=\emptyset$, meaning that every node in $I$ has sent $\READY(t,a)$
    messages. Since $I$ is a quorum, all the nodes in $I$ will eventually deliver a
    Boolean for tag~$t$ due to the condition in line~\ref{lin:ready} of
    Algorithm~\ref{alg:broadcast}.
  \item \emph{Consistency for intertwined nodes:} Assume that two intertwined nodes $v$
    and $v'$ deliver values $a$ and $a'$ for tag $t$ respectively. By the condition in
    line~\ref{lin:deliver}, the nodes received a quorum of $\READY(t,a)$, respectively,
    $\READY(t,a')$ messages. Since the two nodes are intertwined, there is a correct node
    $u$ in the intersection of the two quorums, which sent both $\READY(t,a)$ and
    $\READY(t,a')$. By the use of the guard variable $\readied$ in line~\ref{lin:ready} of
    Algorithm~\ref{alg:broadcast}, node $u$ can only send one and the same ready message
    for tag $t$ to every other node, and thus $a = a'$ as required.
  \item \emph{Validity for intact sets:} Assume every node in an intact set $I$ votes for
    value $a$. Since $I$ is a quorum, every node in $I$ will eventually send $\READY(t,a)$
    by the condition in line~\ref{lin:ready} of Algorithm~\ref{alg:broadcast}. By
    Lemma~\ref{lem:federated-ready-consistent-intact-set}, these messages cannot carry a
    value different from $a$. Then by the condition in line~\ref{lin:deliver} of
    Algorithm~\ref{alg:broadcast} every node in $I$ will eventually deliver the value $a$
    for tag $t$. Due to \emph{Consistency for intact sets}, no node delivers a value
    different from $a$.\qedhere
  \end{description}
\end{proof}

\begin{proof}[Proof ol Lemma~\ref{lem:bounded-totality}]
  We show how to strengthen the proof of \emph{Totality for intact sets} above to prove
  the required. In order to reach the point at which no further nodes in $I$ can send
  $\READY(b,a)$ messages, each node has to send and receive a finite number of
  messages. This number is unknown, but its upper bound is determined by the size of the
  system and the topology of the slices, and thus it is constant and bounded. Since GST
  has expired, the point at which no further nodes in $I$ can send $\READY(b,a)$ messages
  is reached in bounded time. Therefore, every node in $I$ sends $\READY(b,a)$ to every
  other node in $I$ in bounded time. These messages arrive in bounded time too, after
  which every node in $I$ delivers a value, and the required holds.
\end{proof}


\section{Proofs in \S\ref{sec:abstract-federated-consensus}}
\label{ap:abstract-federated-consensus}

In the remainder of the appendix, we say ballot $b_1$ is \emph{below and compatible than}
ballot $b_2$ (written $b_1\lc b_2$) iff $b_1 \leq b_2$ and $b_1\sim b_2$.

\begin{proof}[Proof of Lemma~\ref{lem:prepared-before-commit-intact-set}]
  Assume that a node $v_1 \in I$ commits ballot $b$. By line~\ref{lin:ready} of
  Algorithm~\ref{alg:broadcast}, node $v_1$ received $\READY(b,\true)$ from every member
  of a quorum to which $v_1$ belongs. By Lemma~\ref{lem:federated-intact-set-ready} the
  first node to do so received $\VOTE(b,\true)$ messages from every member of a quorum $U$
  to which $v_1$ belongs. Since $v_1$ is intertwined with every other node in $I$, there
  exists a correct node $v_2$ in the intersection $U \cap I$ that sent
  $\VOTE(b,\true)$. The node $v_2$ can send $\VOTE(b,\true)$ only through
  line~\ref{lin:send-echo} of Algorithm~\ref{alg:broadcast}, which means that $v_2$
  triggers $\brs[b].\vote(\true)$ in line~\ref{lin:send-commit} of
  Algorithm~\ref{alg:abstract-federated-consensus}. By line~\ref{lin:prepared} of the same
  figure, this is only possible after $v_2$ has aborted every $b' \lic b$, and the lemma
  holds.
\end{proof}

\begin{proof}[Proof of Lemma~\ref{lem:ready-commit-prepare-intact-set}]
  Assume towards a contradiction that $v_1$ prepares $b_1$, and that $v_2$ sends
  $\READY(b_2,\true)$. By definition of prepare, node $v_1$ aborted every ballot
  $b \lic b_1$. By line~\ref{lin:ready} of Algorithm~\ref{alg:broadcast}, node $v_1$
  received $\READY(b,\false)$ from every member of a quorum $U_b$ for each ballot
  $b \lic b_1$.  By assumptions, $b_2\lic b_1$, and therefore $v_2$ received
  $\READY(b_2,\false)$ from every member of the quorum $U_{b_2}$. By
  Lemma~\ref{lem:federated-intact-set-ready}, the first node $u_1\in I$ that sent
  $\READY(b_2,\false)$ received $\VOTE(b_2,\false)$ from a quorum $U_1$ to which $u_1$
  belongs. Since $v_2$ sent $\READY(b_2,\true)$ and by
  Lemma~\ref{lem:federated-intact-set-ready}, the first node $u_2\in I$ that sent
  $\READY(b_2,\true)$ received $\VOTE(b_2,\true)$ from a quorum $U_2$ to which $u_2$
  belongs. Since $u_1$ and $u_2$ are intertwined, the intersection $U_1\cap U_2$ contains
  some correct node $v$, which sent both $\VOTE(b_2,\false)$ and $\VOTE(b_2,\true)$
  messages. By the use of the Boolean $\voted$ in line~\ref{lin:echo} of
  Algorithm~\ref{alg:broadcast} this results in a contradiction and we are done.
\end{proof}

\begin{lemma}\label{lem:commit-largest-prepared-intact-set}
  Let $\Ss$ be an FBQS and consider an execution of \AFCP\ over $\Ss$. Let $I$ be an
  intact set in $\Ss$. If a node $v_1\in I$ commits a ballot $b_1$, then the largest
  ballot $b_2$ prepared by any node $v_2\in I$ before $v_1$ commits $b_1$ is such that
  $b_1\sim b_2$. \al{Check removing $I$ and adding that every node is honest.}
\end{lemma}
\begin{proof}
  Assume node $v_1$ commits ballot $b_1$. By the guard in line~\ref{lin:deliver} of
  Algorithm~\ref{alg:broadcast}, node $v_1$ received the message $\READY(b_1,\true)$
  from every member of a quorum to which $v_1$ belongs, which entails that node $v_1$
  received $\READY(b_1,\true)$ from itself. By
  Lemma~\ref{lem:federated-intact-set-ready}, the first node $u\in I$ that send
  $\READY(b_1,\true)$ needs to receive a $\VOTE(b_1,\true)$ message from every member of
  some quorum to which $u$ belongs. Thus, $u$ itself triggered $\brs[b_1].\vote(\true)$,
  which by lines~\ref{lin:send-prepare} and~\ref{lin:prepare-increased} of
  Algorithm~\ref{alg:abstract-federated-consensus} means that $u$ prepared ballot
  $b_1$. Hence, the largest ballot $b_2$ such that there exists a node $v_2\in I$ that
  triggers $\brs[b_2].\vote(\true)$ before $v_1$ commits $b_1$, is bigger or equal than
  $b_1$. If $b_2=b_1$, then $b_2.\val = b_1.\val$ and by
  lines~\ref{lin:prepared}--\ref{lin:send-commit} of
  Algorithm~\ref{alg:abstract-federated-consensus}, node $v_2$ prepares $b_2$ before it
  triggers $\brs[b_2].\vote(\true)$ and the lemma holds.

  If $b_2>b_1$, then we assume towards a contradiction that $b_2.\val\not=b_1.\val$. By
  lines~\ref{lin:prepared}--\ref{lin:send-commit} of
  Algorithm~\ref{alg:abstract-federated-consensus}, node $v_2$ prepared $b_2$. But this
  results in a contradiction by Lemma~\ref{lem:ready-commit-prepare-intact-set}, because
  $v_1$ and $v_2$ are intertwined and $v_1$ sent $\READY(b_1,\true)$, but $b_1\lic
  b_2$.
  Therefore $b_2.\val=b_1.\val$, and by lines~\ref{lin:prepared}--\ref{lin:send-commit} of
  Algorithm~\ref{alg:abstract-federated-consensus}, node $v_2$ prepares $b_2$ before it
  triggers $\brs[b_2].\vote(\true)$.
\end{proof}

We define the \emph{ready-tree for Boolean $a$ and ballot $b$ at node $v$}, which
characterises the messages that need to be exchanged by the FV protocol in order for node
$v$ to send a $\READY(b,a)$ message, under the assumption that all nodes are honest. The
\emph{ready-tree for Boolean $a$ and ballot $b$ at node $v$} is the tree computed
recursively as follows:
\begin{itemize}
\item If $v$ sent the message after receiving $\VOTE(b,a)$ from every member of a quorum
  $U$, then let $U$ be the root of the tree, which has no children.
\item If $v$ sent the message after receiving $\READY(b,a)$ from every member of a
  $v$-blocking set $B$, then let $B$ be the root of the tree, and let its children be the
  ready-trees for Boolean $a$ and ballot $b$ at each of the members of $B$.
\end{itemize}
For short, we may say `ready-tree at node $v$' when the Boolean $a$ and the ballot $b$ are
clear from the context.

A ready-tree is always of finite height, or otherwise some node would have faked ready
messages by lines~\ref{lin:ready}--\ref{lin:send-ready} and
\ref{lin:v-blocking}--\ref{lin:send-v-blocking} of Algorithm~\ref{alg:broadcast}, which
contradicts the assumption that all nodes are honest.

\begin{lemma}\label{lem:cascade}
  Let $\Ss$ be an FBQS and $b$ be a ballot, and consider and execution of the instance for
  ballot $b$ of FV over $\Ss$. Assume all nodes are honest. If a node $v$ sends
  $\READY(b,a)$ then there exists a quorum $U$ such that every member of $U$ sent
  $\VOTE(b,a)$.
\end{lemma}
\begin{proof}
  Assume that a node $v$ sends $\READY(b,a)$, and consider the ready-tree at node $v$. The
  lemma holds since each leaf of the ready-tree is a quorum whose members sent
  $\VOTE(b,a)$.
\end{proof}

\begin{proof}[Proof of Lemma~\ref{lem:prepared-propose-intact-set}]
  Assume all nodes are honest, which entails that each node in the system sends the same
  set of batches to every node. We proceed by induction on the number of batches
  sent by the nodes in the execution so far.

  Since $v_1$ prepares $b_1$, and by lines~\ref{lin:deliver}--\ref{lin:send-deliver} of
  Algorithm~\ref{alg:broadcast}, for each ballot $b_i\lic b_1$ there exists a quorum $U_i$
  such that $v_1$ receives $\READY(b_i,\false)$ from every member of $U_i$. Since the
  number of quorums to which $v_1$ belongs is finite and by BNS, there exists a quorum
  $U_R$ among the $U_i$ such that every member of $U_R$ sent a batch containing a message
  $\READY(b_j,\false)$ for each $b_j$ in some right-open interval $[b,b_1)$ with
  $b < b_1$, where the ballot $b$ is determined by the ballots that $v_1$ had aborted
  before preparing $b_1$. If $v_1$ had aborted every ballot $b_k\lc b_1$ before preparing
  $b_1$, or if $v_1$ had received from each member of $U_R$ a message $\READY(b_k,\false)$
  for each ballot $b_k\lc b_1$ such that $v_1$ did not abort $b_k$ before preparing $b_1$,
  then $v_1$ would have prepared a ballot bigger than $b_1$ by BNS and by
  lines~\ref{lin:deliver}--\ref{lin:send-deliver} of Algorithm~\ref{alg:broadcast}, which
  contradicts the assumptions. Therefore, there exists a ballot $b_0\lc b_1$ and a node
  $u\in U_R$ such that $v_1$ did not abort $b_0$ before preparing $b_1$, and such that
  $v_1$ received from $u$ a batch that contains a message $\READY(b_j,\false)$ for each
  $b_j\in [b,b_1)$, but which does not contain the message $\READY(b_0,\false)$.

  By Lemma~\ref{lem:cascade}, for each $b_j\in [b,b_1)$ there exists a quorum $U_j$ such
  that $u$ received $\VOTE(b_j,\false)$ from each member of $U_j$. Without loss of
  generality, we fix each $U_j$ to be the union of the quorums at the leaves of the
  ready-tree for Boolean $\false$ and ballot $b_j$ at node $u$, which is itself a quorum
  since quorums are closed by union. What follows mimics the argument used in the previous
  paragraph in order to show that some node exists that sent enough vote messages to
  prepare $b_1$, and not so many as to prepare a ballot bigger than $b_1$. Since the
  number of quorums to which $u$ belongs is finite and by BNS, there exists a quorum $U_V$
  among the $U_j$ such that every member of $U_V$ sent a batch containing a message
  $\VOTE(b_\ell,\false)$ for each $b_\ell$ in some right-open interval $[b',b_1)$ with
  $b' < b_1$, where the ballot $b'$ is determined by the messages sent by $u$ before
  sending the batch that contains the ready messages described in the paragraph above. If
  $u$ had sent $\READY(b_k,\false)$ for every ballot $b_k\lc b_1$ before sending the batch
  described in the paragraph above, or if $u$ had received from each member of $U_V$ a
  message $\VOTE(b_k,\false)$ for each ballot $b_k\lc b_1$ such that $u$ did not send
  $\READY(b_k,\false)$ before sending the batch described in the paragraph above, then the
  ballot $b_0$ would not exist by the definition of ready-tree, by BNS, and by
  lines~\ref{lin:deliver}--\ref{lin:send-deliver} of Algorithm~\ref{alg:broadcast}, which
  would contradict the facts established in the paragraph below. Therefore, there exists a
  ballot $b_0'\lc b_1$ and a node $v_2\in U_V$ such that $u$ did not send
  $\READY(b_0',\false)$ before sending the batch described in the paragraph above, and
  such that $u$ received from $v_2$ a batch that contains a message $\VOTE(b_\ell,\false)$
  for each $b_\ell\in [b',b_1)$, but which does not contain the message
  $\VOTE(b_0',\false)$.

  By BNS and by line~\ref{lin:if-echoed} of Algorithm~\ref{alg:broadcast} and
  lines~\ref{lin:send-prepare} and \ref{lin:prepare-increased} of
  Algorithm~\ref{alg:abstract-federated-consensus}, in order for $v_2$ to send a batch
  containing the vote messages described in the previous paragraph, the node necessarily
  tried to prepare some ballot $b_2\geq b_1$ such that
  $b_2.\val\not\in\{b.\val \mid b\in [b',b_1)\}$---which results in $v_2$ triggering
  $\brs[b_i].\vote(b_i,\false)$ for each ballot $b_i\lic b_2$, thus triggering
  $\brs[b_\ell].\vote(b_\ell,\false)$ for each ballot $b_\ell \in [b',b_1)$---and either
  \begin{enumerate}[(i)]
  \item \label{it:b2-cl-b1} $b_1\lc b_2$, or otherwise
  \item \label{it:b2-cil-b1} $b_1\lic b_2$ and $v_2$ triggered
    $\brs[b_0'].\vote(b_0',\true)$ before triggering the $\brs[b_i].\vote(b_i,\false)$ for
    each $b_i\lic b_2$,
  \end{enumerate}
  which prevents that $v_2$ triggers $\brs[b_0'].\vote(b_0',\false)$.

  We show that (\ref{it:b2-cl-b1}) above always holds. Assume (\ref{it:b2-cil-b1}) towards
  a contradiction. Let us focus on the ballots ${b_3=\bllt{1}{b_1.\val}\lc b_1}$ and
  ${b_4=\bllt{1}{b_2.\val}\lc b_2}$. By BNS node $v_2$ triggered each
  $\brs[b_j].\vote(b_j,\true)$ before triggering any $\brs[b_i].\vote(b_i,\false)$ with
  $i\geq j$, and thus it triggered $\brs[b_3].\vote(b_3,\true)$ before triggering the
  $\brs[b_i].\vote(b_i,\false)$ with $b_i\geq b_3$ and $b_i\lic b_4$. Since $v_2$ triggers
  $\brs[b_3].\vote(b_3,\true)$, and by lines~\ref{lin:prepared}--\ref{lin:send-commit} of
  Algorithm~\ref{alg:abstract-federated-consensus}, the largest ballot $\b$ prepared by
  $v_2$ before triggering $\brs[b_3].\vote(b_3,\false)$ is such that $\b\leq b_3$. Since
  every node proposes some value when the execution of the protocol starts by the
  specification of consensus, the value $\b.\val$ is bigger or equal than the value
  proposed by $v_2$ by lines~\ref{lin:init-b}, \ref{lin:h-reached-b} and
  \ref{lin:assign-b} of Algorithm~\ref{alg:abstract-federated-consensus}. Since $v_2$
  triggered $\brs[b_i].\vote(b_i,\false)$ for each $b_i\lic b_4$, and by BNS and by
  line~\ref{lin:send-prepare} of Algorithm~\ref{alg:abstract-federated-consensus}, node
  $v_2$ proposed a value bigger or equal than $b_4.\val$. But since $\b\leq b_3 < b_4$,
  this contradicts the fact we established earlier that $\b.\val$ is bigger or equal than
  the value proposed by $v_2$. Therefore, we know that $b_2\cl b_1$.

  Now we distinguish the following cases. If $b_2.\round = 1$ then, by
  line~\ref{lin:send-prepare} of Algorithm~\ref{alg:abstract-federated-consensus}, node
  $v_2$ proposed $b_2.\val=b_1.\val$ and the lemma holds. If $b_2.\round > 1$ then, by BNS
  and by lines~\ref{lin:timeout}--\ref{lin:prepare-increased} of
  Algorithm~\ref{alg:abstract-federated-consensus}, node $v_2$ triggered a timeout event
  when its current round was $b_2.\round - 1$. If $v_2$ never prepared any ballot when
  this timeout expired, then by line~\ref{lin:increase-candidate} of
  Algorithm~\ref{alg:abstract-federated-consensus} the node proposed $b_2.\val$ and the
  lemma holds. Otherwise, by line~\ref{lin:increase-prepared} of
  Algorithm~\ref{alg:abstract-federated-consensus} the value $b_2.\val$ is equal to
  $h.\val$ where $h$ is the largest ballot prepared by $v_2$ when that timeout expired,
  and the required holds by applying the induction hypothesis.
\end{proof}

Recall from \S\ref{sec:stellar-broadcast} the \emph{window for intact set $I$ of round
  $n$}, which is the interval of duration $\delta_I$ in which every node in $I$ that has
not decided any value yet prepares a ballot of round $n$.


Now we introduce some additionally terminology for the proofs in the remainder of the
appendix. Assume all faulty nodes have stopped, and consider the window for intact set $I$
of round $n$, where we let $v_n$ be the first node in $I$ that prepares a ballot $b_n$
with round $n$. By line~\ref{lin:prepared} of
Algorithm~\ref{alg:abstract-federated-consensus} and
lines~\ref{lin:deliver}--\ref{lin:send-deliver} of Algorithm~\ref{alg:broadcast}, for each
ballot $b_i\lic b_n$ there exists a quorum $U_i$ such that node $v_n$ received
$\READY(b_i,\false)$ from each member of $U_i$. We call the \emph{prepare-footprint for
  intact set $I$ of round $n$}, written $P_I^n$, to the set computed as follows:
\begin{itemize}
\item Let $v_n\in I$ be the first node in $I$ that prepares a ballot $b_n$ with round $n$.
\item Let $h$ be the highest ballot prepared by $v_n$ before preparing $b_n$.
\item Compute the set of ballots that $v_n$ needs to abort in order to prepare $b_n$ and
  that $v_n$ did not abort previously. This is, take the ballots $b_j\lic b_n$ such that
  either $b_j > h$, or $b_j\lc h$ and $v_n$ never aborted $b_j$ yet.
\item For each $b_j$ in the set computed in the previous step, and for each member
  $u\in U_j$, compute the ready-tree of Boolean $\false$ and ballot $b_j$ at node $u$.
\item Let $P_I^n$ be the union of the leaves of all the ready-trees computed in the previous
  step.
\end{itemize}

Since the leaves of any ready-tree are quorums, and since quorums are closed under union,
the prepare-footprint $P_I^n$ is a quorum. By the definition of ready-tree, node $v_n$
prepares $b_n$ because of the vote messages sent by each member of $P_I^n$.

For each window for intact set $I$ and round $n$, we consider the \emph{abort-interval for
  intact set $I$ of round $n$}, which is the interval in which the nodes in $P_I^n$ send
the vote messages needed for $v_n$ to prepare $b_n$. We write $\delta_I^{An}$ for the
\emph{duration of the abort-interval for intact set $I$ of round $n$}.

Since the prepare-footprint $P_I^n$ of any round $n$ in any execution of \AFCP\ is finite,
and since messages arrive in bounded time after GST expires, after the last node in the
prepare footprint $P_I^n$ sends its messages, the first node $v_n\in I$ that prepares a
ballot $b_n$ with round $n$ does so in bounded time. We write $\delta_I^{Pn}$ for the
finite delay that $v_n$ takes in preparing $b_n$ after every node in $P_I^n$ has sent
their messages, once GST has expired. Since the number of nodes in the universe is finite,
all the $\delta_I^{Pn}$ with any round $n$ are bounded by some finite delay $\delta_I^P$.

From now on we may omit the `for intact set $I$' qualifier from the window, the
abort-interval and the prepare-footprint of round $n$ when the intact set is clear from
the context.

\begin{proof}[Proof of Lemma~\ref{lem:window-no-overlap}]
  If every node in $I$ decides a value before reaching some round $n$ then we are
  done. Otherwise, assume that GST has expired and let $\delta$ be the network delay after
  GST. Without loss of generality, assume that all faulty nodes have already stopped. Let
  $n_0$ be a round such that the abort-interval of round $n_0$ happens entirely after
  GST. (We accompany the proof with Figure~\ref{fig:intervals} to illustrate the intervals
  that we describe in the remainder.) After the abort-interval of $n_0$, the first node
  $v_{n_0}\in I$ to ever prepare a ballot $b_{n_0}$ with round $n_0$ does it less than
  $\delta_I^P$ time after the abort-interval of $n_0$ terminates. By \emph{Totality for
    intact nodes}, each node in $I$ that has not decided any value yet prepares some ballot
  with round $n_0$ during the window of round $n_0$, and thus by lines
  lines~\ref{lin:quorum-round}--\ref{lin:start-timer} of
  Algorithm~\ref{alg:abstract-federated-consensus} each node in $I$ that has not decided
  any value yet sets its timeout to $F(n_0)$ in the period of time between the beginning
  of the abort-interval of $n_0$ and the end of the window of $n_0$. (Remember the
  primitive $\starttimer$ has the effect of restarting the timeout if a node's local timer
  was already running.) This period of time has a duration bounded by
  $\delta_I^{An_0} + \delta_I^P + \delta_I$.

  Now consider any round $n\geq n_0$ and assume that
  $F(n)>\delta_I^{An} + \delta_I^P + \delta_I$. As shown in Figure~\ref{fig:intervals},
  the abort-interval of round $n+1$ happens entirely after the window of round $n$
  ends. Since the window of round $n+1$ can only happen once the nodes in the
  abort-interval of round $n+1$ have sent their messages, the window of round $n+1$ does
  not overlap the window of round $n$. And since the delay function $F$ doubles its value
  with each increment of the round, the same happens with any subsequent window
  thereafter. By the definition of window of round $n$, no node in $I$ ever prepares a
  ballot with round equal or bigger than $n$ before the window starts. Since faulty nodes
  are stopped, and by lines~\ref{lin:timeout}--\ref{lin:prepare-increased} of
  Algorithm~\ref{alg:abstract-federated-consensus}, no messages supporting to prepare a
  ballot of round $n+1$ are ever sent before the abort-interval of round $n+1$
  starts. Therefore, each window of round $n + 1$ with $n\geq n_0$ happens after the
  immediately preceding window of round $n$, and no two of such windows overlap with each
  other, and the lemma holds.\qedhere
\end{proof}

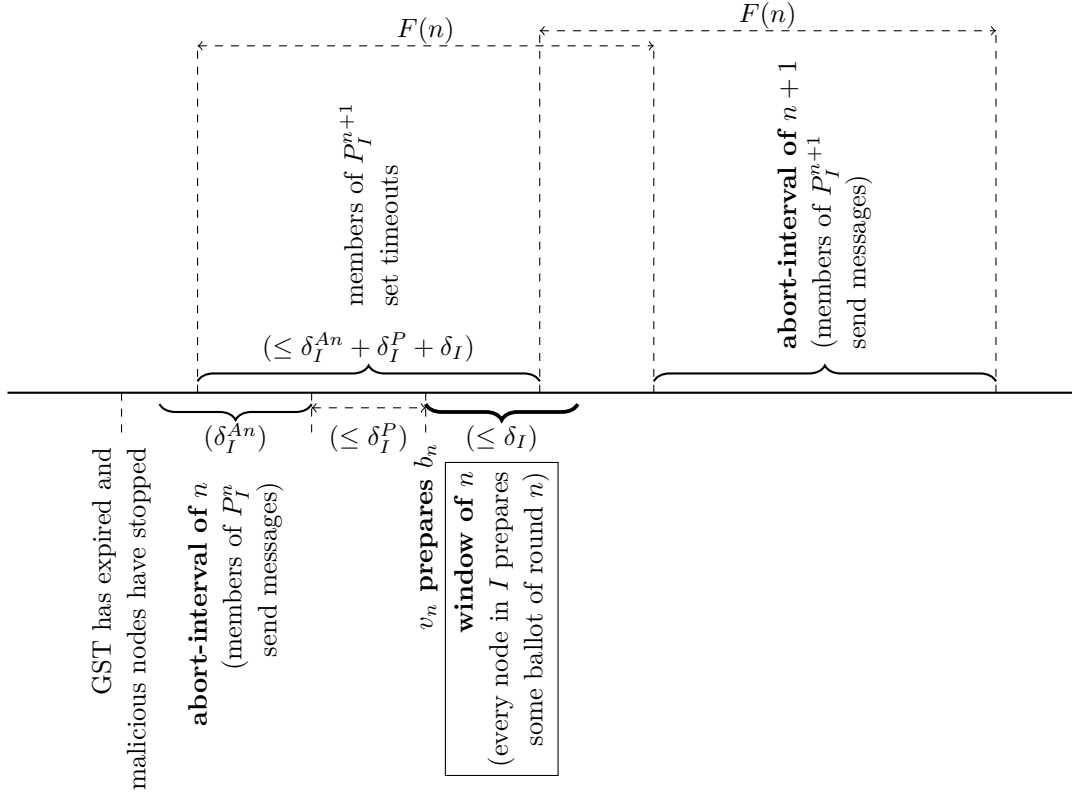
\begin{figure}[t]
  \centering

  \begin{tikzpicture}[scale=1]
    \draw [thick] (-2,0) -- (12,0);

    \draw [thin,dashed] (-.5,0) -- (-.5,-.6);
    \node [align=right,rotate=90] at (-.5,-3) {GST has expired and\\
      malicious nodes have stopped};

    \draw [thick,decorate,decoration={brace,amplitude=6pt,raise=0pt,mirror}]
    (0,-0.15) -- (2,-0.15);
    \node [align=center] at (1,-.6) {$(\delta_I^{An})$};
    \node[align=right,rotate=90] at (1,-2.8) {\textbf{abort-interval of $n$}\\
      (members of $P_I^n$\\send messages)};

    \draw [thin,dashed] (2,0) -- (2,-.6);

    \draw [thin,dashed,<->] (2,-.2) -- (3.5,-.2);
    \node [align=center] at (2.75,-.6) {$(\leq \delta_I^P)$};

    \draw [thin,dashed] (3.5,0) -- (3.5,-.6);
    \node [align=right,rotate=90] at (3.5,-1.9) {\textbf{$v_n$ prepares $b_n$}};

    \draw [line width=0.5mm,decorate,decoration={brace,amplitude=6pt,raise=0pt,mirror}]
    (3.5,-0.15) -- (5.5,-0.15);
    \node [align=center] at (4.5,-.6) {$(\leq \delta_I)$};
    \node[rectangle,draw = black,align=right,rotate=90] at (4.5,-3)
    {\textbf{window of $n$}\\
      (every node in $I$ prepares\\
      some ballot of round $n$)};

    \draw [thin,dashed] (.5,0) -- (.5,4.6);
    \draw [thick,decorate,decoration={brace,amplitude=6pt,raise=0pt}]
    (.5,.15) -- (5,.15);
    \node [align=center] at (2.75,.6) {$(\leq \delta_I^{An}+\delta_I^P+\delta_I)$};
    \node[align=left,rotate=90] at (2.75,2.5) {members of $P_I^{n+1}$\\set timeouts};
    \draw [thin,dashed] (5,0) -- (5,4.8);

    \draw [thin,dashed,<->] (.5,4.6) -- (6.5,4.6);
    \node [align=center] at (3.5,4.8) {$F(n)$};

    \draw [thin,dashed,<->] (5,4.8) -- (11,4.8);
    \node [align=center] at (8,5) {$F(n)$};

    \draw [thin,dashed] (6.5,0) -- (6.5,4.6);
    \draw [thick,decorate,decoration={brace,amplitude=6pt,raise=0pt}]
    (6.5,0.15) -- (11,0.15);
    \node[align=left,rotate=90] at (8.75,2.5) {\textbf{abort-interval of $n+1$}\\
      (members of $P_I^{n+1}$\\send messages)};
    \draw [thin,dashed] (11,0) -- (11,4.8);
  \end{tikzpicture}

  \caption{Window (boxed) and abort-interval of round $n$, assuming that
    $F(n) > \delta_I^{An}+\delta_I^P+\delta_I$.}
  \label{fig:intervals}
\end{figure}

\begin{corollary}\label{cor:b-max}
  Let $m\geq n$ and let $b_{\max}$ be the maximum ballot prepared by any node in $I$
  before the abort-interval of round $m+1$ starts. Every node in $I$ prepares $b_{\max}$
  before the abort-interval of round $m+1$ starts.
\end{corollary}
\begin{proof}
  Since the set-timeout interval is of length at most $\delta$, ballot $b_{\max}$ is
  prepared by some node in $I$ within $\delta+\delta_I^P$ time after the set-timeout
  interval of round $m$ starts. By Lemma~\ref{lem:bounded-totality}, every node in $I$
  prepares $b_{\max}$ within $\delta+\delta_I^P+\delta_I$ time after the set-timeout
  interval of round $m$ starts. The required holds since
  $F(m)>\delta+\delta_I^P+\delta_I$.
\end{proof}

\begin{proof}[Proof of Theorem~\ref{thm:non-blocking-byzantine-consensu-intact-sets}]
  We prove that \AFCP\ over $\Ss$ enjoys each of the properties that define the
  specification of \emph{non-blocking Byzantine consensus for intact sets}.
  \begin{description}
  \item \emph{Integrity:} Straightforward by definition since each node $v$ stops
    execution once $v$ decides some value.
  \item \emph{Agreement for intact sets:} Assume towards a contradiction that two nodes
    $v_1$ and $v_2$ in $I$ decide respectively through ballots $b_1$ and $b_2$ such that
    $b_1.\val \not= b_2.\val$.  Without loss of generality, we assume $b_1< b_2$. By
    line~\ref{lin:decided} of Algorithm~\ref{alg:abstract-federated-consensus}, nodes
    $v_1$ and $v_2$ respectively committed ballots $b_1$ and $b_2$. Since $v_1$ has
    committed ballot $b_1$ and by lines~\ref{lin:deliver}--\ref{lin:send-deliver} of
    Algorithm~\ref{alg:broadcast}, node $v_1$ received $\READY(b_1,\true)$ from a quorum
    to which $v_1$ belongs. And since $v_2$ has committed ballot $b_2$ and by
    Lemma~\ref{lem:prepared-before-commit-intact-set}, we know that some correct node
    $v'\in I$ has prepared $b_2$. But this results in a contradiction by
    Lemma~\ref{lem:ready-commit-prepare-intact-set} and we are done.
  \item \emph{Weak validity for intact sets:} Assume all nodes are honest. To prove the
    first part of the property, let every node propose value $x$. Now assume towards a
    contradiction that a node $v_1\in I$ decides value $y\not=x$. By
    Lemma~\ref{lem:commit-largest-prepared-intact-set}, the largest ballot $b_2$ prepared
    by any node $v_2\in I$ before $v_1$ decides $y$ is such that $b_2.\val=y$. By
    Lemma~\ref{lem:prepared-propose-intact-set}, there exists a node that proposed value
    $y$. But this contradicts our assumption that every node proposed value $x$.

    To prove the second part of the property, let a node $v_1\in I$ decide value $x$. By
    Lemma~\ref{lem:commit-largest-prepared-intact-set}, the largest ballot $b$ prepared by
    any node $v_2\in I$ before $v_1$ decides $x$ is such that $b.x=x$. By
    Lemma~\ref{lem:prepared-propose-intact-set}, there exists a node that proposed value
    $x$, and we are done. \al{Check the proof of validity with the commented lemma above.}
  \item \emph{Non-blocking for intact sets:} Assume all faulty nodes eventually stop. If
    some node in $I$ decides some value, then by Lemma~\ref{lem:bounded-totality} and by
    \emph{Agreement for intact sets} every other node in $I$ will decide the same value
    within bounded time. Without loss of generality, assume that no node in $I$ has
    decided any value and that GST has expires and every malicious node has stopped. By
    Lemma~\ref{lem:prepared-before-commit-intact-set}, no node can ever send a
    $\READY(b,\true)$ message for a ballot $b$ below and incompatible than any ballot that
    is already prepared, and thus no node can block itself by signalling that its
    willingness to commit a ballot that can no longer be committed. By
    Lemma~\ref{lem:window-no-overlap} there exists a round $n$ such that any two windows
    of rounds bigger or equal than $n$ never overlap. Without loss of generality, assume
    that no node in $i$ has decided any valued before the window of round $n$. By
    Corollary~\ref{cor:b-max}, for every window of round $m$ bigger or equal than $n$,
    every node in $I$ that has not decided any value yet prepares the same ballot
    $b_{\max}$ before the abort-interval of round $m+1$ starts. If every node in $I$
    updates its candidate ballot to $b_{\max}$ before the abort-interval of round $m+1$
    starts, then every node in $I$ will try to commit $b_{\max}$ by
    lines~\ref{lin:prepared}--\ref{lin:send-commit} of
    Algorithm~\ref{alg:abstract-federated-consensus}, and they all will decide value
    $b_{\max}.\val$ in bounded time by lines~\ref{lin:echo}--\ref{lin:send-echo},
    \ref{lin:ready}--\ref{lin:send-ready} and \ref{lin:deliver}--\ref{lin:send-deliver} of
    Algorithm~\ref{alg:broadcast}, and lines~\ref{lin:decided}--\ref{lin:send-decided} of
    Algorithm~\ref{alg:abstract-federated-consensus}. Otherwise, every node in $I$ will
    update its candidate ballot to $\bllt{m+1}{b_{\max}.\val}$ in the abort-interval of
    round $m+1$, and all will try to commit $\bllt{m+1}{b_{\max}.\val}$ and decide value
    $b_{\max}.\val$ in bounded time for reasons similar to the ones above.\qedhere
  \end{description}
\end{proof}

Now we introduce a notation for \emph{sequence comprehension} that we will use intensively
in the remainder of the appendices. We write $[{\sf el}_1,\ldots,{\sf el}_m]$ for a
sequence of elements, where each element ${\sf el}_i$ is a ballot, an event or a
message. We write $\emptrace$ for the \emph{empty sequence} and $\tau_1\concat\tau_2$ for
the \emph{concatenation of sequences}. The notation $[{\sf el}(b),\;P(b)]$ stands for the
sequence $[{\sf el}(b_1),\ldots,{\sf el}(b_m)]$ where each element ${\sf el}(b_i)$ depends
on a ballot $b_i$ that meets predicate $P$, and where the elements are ordered in
ascending ballot order, this is, ${\sf el}(b_i)$ occurs before ${\sf el}(b_j)$ in
$[{\sf el}(b),\;P(b)]$ iff $i<j$.

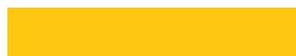
  \begin{landscape}
  \begin{figure}
    \centering
    \begingroup\SmallTabColSep{\small\begin{tabular}[t]{|l|l|l|l|}
    Node $v_1$ &  Node $v_2$  & {\color{red} Node $v_3$} & Node $v_4$ \\
    \hline\hline
    \begin{tabular}[t]{l}
      $\propose(\C)$\\
      $\voteBatch([b ,\, b\lic \bllt{1}{\C}],\false)$\\
      \hdashline
      $[\VOTE(b,\false) ,\; b\lic \bllt{1}{\C}]$\\
    \end{tabular}&
    \begin{tabular}[t]{l}
      $\propose(\C)$\\
      $\voteBatch([b ,\, b\lic \bllt{1}{\C}],\false)$\\
      \hdashline
      $[\VOTE(b,\false) ,\; b\lic \bllt{1}{\C}]$\\
    \end{tabular}&
    {\color{red}\begin{tabular}[t]{l}
                  $[\VOTE(b,\false) ,\; b\lic \bllt{1}{\B}]$
                \end{tabular}}&
    \begin{tabular}[t]{l}
      $\propose(\A)$\\
      $\voteBatch([b ,\, b\lic \bllt{1}{\A}],\false)$\\
      \hdashline
      $[\VOTE(\nbllt,\false)]$
    \end{tabular}\\
    \hline
    \begin{tabular}[t]{l}
      $\starttimer(F(1))$\\
      \hdashline
      $[\READY(b,\false),\; b\lic \bllt{1}{\B}]$
    \end{tabular}&
    \begin{tabular}[t]{l}
      $\starttimer(F(1))$\\
      \hdashline
      $[\READY(b,\false),\; b\lic \bllt{1}{\B}]$
    \end{tabular}&&
    \begin{tabular}[t]{l}
      $\starttimer(F(1))$\\
      \hdashline
      $[\READY(\nbllt,\false)]$
    \end{tabular}\\
    \hline
    \begin{tabular}[t]{l}
      $\deliverBatch([b ,\, b \lic \bllt{1}{\A}], \false)$
    \end{tabular}&
    \begin{tabular}[t]{l}
      $\deliverBatch([b ,\, b \lic \bllt{1}{\A}], \false)$
    \end{tabular}&&
    \begin{tabular}[t]{l}
      $\deliverBatch([b ,\, b \lic \bllt{1}{\A}], \false)$\\
      $\voteBatch([\bllt{1}{\A}], \true)$\\
      \hdashline
      $[\VOTE(\bllt{1}{\A},\true)]$\\
      $[\READY(\bllt{1}{\A},\false)]$
    \end{tabular}\\
    \hline
    \begin{tabular}[t]{l}
      $\deliverBatch([\bllt{1}{\A}], \false)$
    \end{tabular}&
    \begin{tabular}[t]{l}
      $\deliverBatch([\bllt{1}{\A}], \false)$
    \end{tabular}&&
    \begin{tabular}[t]{l}
      $\deliverBatch([\bllt{1}{\A}], \false)$\\
      $\voteBatch([\bllt{1}{\B}], \true)$\\
      \hdashline
      $[\VOTE(\bllt{1}{\B},\true)]$
    \end{tabular}\\
    \hline \qquad\vdots&\qquad\vdots&\qquad\vdots&\qquad\vdots\\
    \hline
    \begin{tabular}[t]{l}
      $\timeout$\\
      $\voteBatch([b\lic \bllt{2}{\B}], \false)$\\
      \hdashline
      $[\VOTE(b,\false),\; \bllt{1}{\C}\leq b\lic\bllt{2}{\B}]$
    \end{tabular}&
    \begin{tabular}[t]{l}
      $\timeout$\\
      $\voteBatch([b\lic \bllt{2}{\B}], \false)$\\
      \hdashline
      $[\VOTE(b,\false),\; \bllt{1}{\C}\leq b\lic\bllt{2}{\B}]$
    \end{tabular}&&
    \begin{tabular}[t]{l}
      $\timeout$\\
      $\voteBatch([b\lic \bllt{2}{\B}], \false)$\\
      \hdashline
      $[\VOTE(b,\false),\; \bllt{1}{\C}\leq b\lic\bllt{2}{\B}]$
    \end{tabular}\\
    \hline
    \begin{tabular}[t]{l}
      $\starttimer(F(2))$\\
      \hdashline
      $[\READY(b,\false),\; \bllt{1}{\C}\leq b\lic\bllt{2}{\B}]$
    \end{tabular}&
    \begin{tabular}[t]{l}
      $\starttimer(F(2))$\\
      \hdashline
      $[\READY(b,\false),\; \bllt{1}{\C}\leq b\lic\bllt{2}{\B}]$
    \end{tabular}&&
    \begin{tabular}[t]{l}
      $\starttimer(F(2))$\\
      \hdashline
      $[\READY(b,\false),\; \bllt{1}{\C}\leq b\lic\bllt{2}{\B}]$
    \end{tabular}\\
    \hline
    \begin{tabular}[t]{l}
      $\deliverBatch([b ,\, \bllt{1}{\C}\leq b\lic\bllt{2}{\B}], \false)$\\
      \hdashline
      $[\VOTE(\bllt{2}{\B},\true)]$
    \end{tabular}&
    \begin{tabular}[t]{l}
      $\deliverBatch([b ,\, \bllt{1}{\C}\leq b\lic\bllt{2}{\B}], \false)$\\
      \hdashline
      $[\VOTE(\bllt{2}{\B},\true)]$
    \end{tabular}&&
    \begin{tabular}[t]{l}
      $\deliverBatch([b ,\, \bllt{1}{\C}\leq b\lic\bllt{2}{\B}], \false)$\\
      \hdashline
      $[\VOTE(\bllt{2}{\B},\true)]$
    \end{tabular}\\
    \hline
    \begin{tabular}[t]{l}
      $[\READY(\bllt{2}{\B},\true)]$
    \end{tabular}&
    \begin{tabular}[t]{l}
      $[\READY(\bllt{2}{\B},\true)]$
    \end{tabular}&&
    \begin{tabular}[t]{l}
      $[\READY(\bllt{2}{\B},\true)]$
    \end{tabular}\\
    \hline
    \begin{tabular}[t]{l}
      $\deliverBatch([\bllt{2}{\B}], \true)$\\
      $\decide(\B)$
    \end{tabular}&
    \begin{tabular}[t]{l}
      $\deliverBatch([\bllt{2}{\B}], \true)$\\
      $\decide(\B)$
    \end{tabular}&&
    \begin{tabular}[t]{l}
      $\deliverBatch([\bllt{2}{\B}], \true)$\\
      $\decide(\B)$
    \end{tabular}\\
    \hline
    \end{tabular}}\endgroup
    \caption{Execution of \AFCP.}
    \label{fig:trace-abstract}
  \end{figure}
\end{landscape} %

\begin{example}\label{ex:abstract}
  Consider the FBQS from Example~\ref{ex:federated-voting} where the universe contains
  four nodes $v_1$ to $v_4$ and every set of three or more nodes is a quorum, and every
  set of two or more nodes is $v$-blocking for any $v\in \V$. Consider an execution of
  \AFCP\ where the node $v_3$ is faulty. The FBQS has the intact set $I=\{v_1,v_2,v_4\}$.

  We assume that the set of values coincides with $\Nat^+$, which we write in boldface. In
  the execution, nodes $v_1$ and $v_2$ propose value $\C$, and node $v_4$ propose value
  $\A$. Faulty node $v_3$ sends a batch containing the messages $\VOTE(\nbllt,\false)$ and
  $\VOTE(\bllt{1}{\A},\false)$ to every correct node, thus helping them to prepare ballot
  $\bllt{1}{\B}$. Since $\bllt{1}{\B}$ exceeds $v_4$'s candidate ballot $\bllt{1}{\A}$,
  node $v_4$ will try to commit both $\bllt{1}{\A}$ and $\bllt{1}{\B}$. However, neither
  of $v_1$ or $v_2$ will try to commit any ballot since $\bllt{1}{\B}$ is smaller than
  their candidate ballot $\bllt{1}{\C}$, and therefore no quorum exists that tries to
  commit a ballot. Consequently, the timeout at round $1$ of every correct node will
  expire, and since all of them managed to prepare $\bllt{1}{\B}$, they all will try to
  prepare the increased ballot $\bllt{2}{\B}$, and will ultimately commit that ballot and
  decide value $\B$. Notice that value $\B$ was not proposed by any correct node, but
  nevertheless all of them agree on the same decision. To the nodes in $I$, node $v_3$
  being faulty is indistinguishable from the situation where node $v_3$ is correct but
  slow, and it proposes $\B$. Therefore the nodes in $I$ cannot detect whether the decided
  value was proposed by some node in $I$ or not.

Figure~\ref{fig:trace-abstract} depicts the trace of the execution of \AFCP\ described
above. In each cell, we separate by a dashed line the events (above the line) that are
triggered atomically, if any, from the batches of messages (below the line) that are sent
by the node, if any. By BNS, the sending of every batch happens atomically with the events
above the dashed line. At each cell, a node has received every batch in the rows above
it. (For convenience, above the dashed line, we depict `batched' events $\voteBatch$ and
$\deliverBatch$, which are defined in \S\ref{sec:concrete-federated-consensus}. Under the
dashed line, we save the `batched' send and receive primitives, and we depict one batch of
messages per line.)

In the first row of Figure~\ref{fig:trace-abstract}, the correct nodes $v_1$, $v_2$ and
$v_4$ try to prepare the ballots that they propose
(lines~\ref{lin:propose}--\ref{lin:send-prepare} of
Algorithm~\ref{alg:abstract-federated-consensus} and
lines~\ref{lin:echo}--\ref{lin:send-echo} of Algorithm~\ref{alg:broadcast}), which results
in each of the $v_1$, $v_2$ and $v_4$ sending a $\VOTE(b,\false)$ message for each
$b\lic \bllt{1}{x}$, where $x$ is respectively $\C$, $\C$ and $\A$. (Faulty node $v_3$
sends a $\VOTE(b,\false)$ message for each $b\lic \bllt{1}{\B}$.) Notice the use of the
sequence comprehension notation to denote sequences of events triggered in a cell, as well
as sequences of messages in a batch. To wit, node $v_1$ triggers $\propose(\C)$ followed
by the batched event $\voteBatch([b,\, b\lic \bllt{1}{\C}],\false)$, which stands for
\begin{align*}
  &[\fvs[\nbllt].\vote(\nbllt,\false),\, \fvs[\bllt{1}{\A}].\vote(\bllt{1}{\A},\false),\\
  &\fvs[\bllt{1}{\B}].\vote(\bllt{1}{\B},\false)],
\end{align*}
and it sends a batch with the sequence of messages
$[\VOTE(b,\false),\,b\lic\bllt{1}{\C}]$, which stands for
\begin{displaymath}
  [\VOTE(\nbllt,\false),\,\VOTE(\bllt{1}{\A},\false),\,\VOTE(\bllt{1}{\B},\false)].
\end{displaymath}

In the second row of Figure~\ref{fig:trace-abstract}, nodes $v_1$, $v_2$, and $v_4$ start
the timer with delay $F(1)$, since there exist ballot $\bllt{1}{\A}$ and open interval
$[\nbllt,\bllt{1}{\A})$ such that the quorum $\{v_1,v_2,v_4\}$ receives from itself a
message $\VOTE(\nbllt,\false)$, and $[\nbllt,\bllt{1}{\A})$ is the singleton containing
the null ballot $\nbllt$ (lines~\ref{lin:quorum-round}--\ref{lin:start-timer} of
Algorithm~\ref{alg:abstract-federated-consensus}). This means that all correct nodes
receive from themselves vote messages that support preparing ballots with rounds bigger or
equal than $1$. In addition to this, nodes $v_1$ and $v_2$ send the batch
$[\READY(b,\false),\,b\lic \bllt{1}{\B}]$, since they receive a message $\VOTE(b,\false)$
for each $b\lic \bllt{1}{\B}$ from the quorum $\{v_1,v_2,v_3\}$, to which they belong
(lines~\ref{lin:ready}--\ref{lin:send-ready} of Algorithm~\ref{alg:broadcast}). And
similarly, node $v_4$ sends a $\READY(\nbllt,\false)$, since it receives the message
$\VOTE(\nbllt,\false)$ from all nodes, which constitute a quorum to which $v_4$
belongs. Notice that node $v_4$ cannot send $\READY(\bllt{1}{\A},\false)$ because no
quorum to which $v_4$ belongs exists that sends $\VOTE(\bllt{1}{\A},\false)$.

In the third row of of Figure~\ref{fig:trace-abstract}, nodes $v_1$, $v_2$ and $v_4$
deliver $\false$ for ballot $\bllt{1}{\A}$, since they receive the message
$\READY(\nbllt,\false)$ from the quorum $\{v_1,v_2,v_4\}$ to which they all belong
(lines~\ref{lin:deliver}--\ref{lin:send-deliver} of Algorithm~\ref{alg:broadcast}), which
results in each of those nodes preparing ballot $\bllt{1}{\A}$ and triggering
lines~\ref{lin:prepared}--\ref{lin:send-commit} of
Algorithm~\ref{alg:abstract-federated-consensus}. Since the prepared ballot $\bllt{1}{\A}$
reaches $v_4$'s candidate ballot, then $v_4$ triggers the batched event
$\voteBatch([\bllt{1}{\A}],\true)$ and prepares a batch with the message
$\VOTE(\bllt{1}{\A},\true)$ that it will send later
(lines~\ref{lin:prepared}--\ref{lin:send-commit} of
Algorithm~\ref{alg:abstract-federated-consensus} and
lines~\ref{lin:echo}--\ref{lin:send-echo} of Algorithm~\ref{alg:broadcast}). In addition
to this, node $v_4$ also prepares a batch with the message $\READY(\bllt{1}{\A},\false)$
that it will also send later, since it receives $\READY(\bllt{1}{\A},\false)$ from the
$v_4$-blocking set $\{v_1,v_2\}$ (lines~\ref{lin:v-blocking}--\ref{lin:send-v-blocking} of
Algorithm~\ref{alg:broadcast}). Recall that the rule in
lines~\ref{lin:v-blocking}--\ref{lin:send-v-blocking} of Algorithm~\ref{alg:broadcast}
allows a node to send a ready message with some Boolean even if the node previously voted
a different Boolean for the same ballot. Finally, node $v_4$ sends the two batches
prepared before atomically.

In the fourth row of Figure~\ref{fig:trace-abstract}, nodes $v_1$, $v_2$ and $v_4$ deliver
$\false$ for ballot $\bllt{1}{\B}$, since they receive a message $\READY(b,\false)$ for
each $b\lic \bllt{1}{\B}$ from the quorum $\{v_1,v_2,v_4\}$ to which they all belong
(lines~\ref{lin:deliver}--\ref{lin:send-deliver} of Algorithm~\ref{alg:broadcast}), which
results in each of those nodes preparing ballot $\bllt{1}{\B}$ and triggering
lines~\ref{lin:prepared}--\ref{lin:send-commit} of
Algorithm~\ref{alg:abstract-federated-consensus}. Since the prepared ballot $\bllt{1}{\B}$
exceeds $v_4$'s candidate ballot, then $v_4$ updates its candidate ballot to
$\bllt{1}{\B}$ and triggers $\voteBatch([\bllt{1}{\B}], \true)$, which results in $v_4$
sending $\VOTE(\bllt{1}{\B},\true)$ (lines~\ref{lin:prepared}--\ref{lin:send-commit} of
Algorithm~\ref{alg:abstract-federated-consensus} and
lines~\ref{lin:echo}--\ref{lin:send-echo} of Algorithm~\ref{alg:broadcast}).

At this point no node can decide any value, because there exists not any ballot such that
a quorum of nodes votes $\true$ for it, and the timeouts of all correct nodes will expire
after $F(1)$ time.

In the sixth row of Figure~\ref{fig:trace-abstract}, nodes $v_1$, $v_2$ and $v_4$ trigger
$\timeout$, and since they all prepared ballot $\bllt{1}{\B}$, they update their candidate
ballot to $\bllt{2}{\B}$ and trigger the batched event
$\voteBatch.([b,\, b\lic \bllt{2}{\B}],\false)$
(lines~\ref{lin:timeout}--\ref{lin:increase-prepared} of
Algorithm~\ref{alg:abstract-federated-consensus}). Nodes $v_1$, $v_2$ and $v_4$ send the
batch $[\VOTE(\bllt{2}{\B},\false),\,\bllt{1}{\C}\leq b \lic \bllt{2}{\B}]$, which
contains infinitely many messages that are sent at once by BNS.

In the seventh row of Figure~\ref{fig:trace-abstract}, nodes $v_1$, $v_2$ and $v_4$ start
the timer with delay $F(2)$, since there exist ballot $\bllt{2}{\B}$ and open interval
$[\bllt{1}{\B},\bllt{2}{\B})$ such that the quorum $\{v_1,v_2,v_4\}$ receives from itself
the infinitely many messages $\VOTE(b,\false)$ with $b\in [\bllt{1}{\B},\bllt{2}{\B})$
(lines~\ref{lin:quorum-round}--\ref{lin:start-timer} of
Algorithm~\ref{alg:abstract-federated-consensus}), which are received at once by BNS. This
means that all correct nodes receive from themselves vote messages that support preparing
ballots with rounds bigger or equal than $2$. Then, nodes $v_1$, $v_2$ and $v_4$ send the
batch $[\READY(b,\false),\,\bllt{1}{\C}\leq b\lic \bllt{2}{\B}]$, since they receive a
message $\VOTE(b,\false)$ for each $b$ such that $\bllt{1}{\C}\leq b\lic \bllt{2}{\B}$
from the quorum $\{v_1,v_2,v_3\}$ to which they belong
(lines~\ref{lin:ready}--\ref{lin:send-ready} of Algorithm~\ref{alg:broadcast}). The batch
contains infinitely many messages, which are sent at once by BNS.

In the eight row of Figure~\ref{fig:trace-abstract}, nodes $v_1$, $v_2$ and $v_4$ trigger
$\voteBatch(b,\, \bllt{1}{\C}\leq b\lic\bllt{2}{\B}],\false)$, which
stands for a vote $\false$ for each~$b$ below and incompatible than $\bllt{2}{\B}$ for
which the node didn't vote any Boolean yet, since they receive a message
$\READY(b,\false)$ for each of such $b$'s from the quorum $\{v_1,v_2,v_4\}$ to which they
all belong (lines~\ref{lin:deliver}--\ref{lin:send-deliver} of
Algorithm~\ref{alg:broadcast}). Since the prepared ballot $\bllt{2}{\B}$ reaches the
candidate ballot of all correct nodes, they trigger the event
$\voteBatch([\bllt{2}{\B}], \true)$ and send a $\VOTE(\bllt{2}{\B},\true)$
(lines~\ref{lin:prepared}--\ref{lin:send-commit} of
Algorithm~\ref{alg:abstract-federated-consensus} and
lines~\ref{lin:echo}--\ref{lin:send-echo} of Algorithm~\ref{alg:broadcast}).

In the ninth row of Figure~\ref{fig:trace-abstract}, nodes $v_1$, $v_2$ and $v_4$ send the
batch $[\READY(\bllt{2}{\B},\true)]$, since they all received $\VOTE(\bllt{2}{\B},\true)$
from the quorum $\{v_1,v_2,v_4\}$ to which all belong
(lines~\ref{lin:ready}--\ref{lin:send-ready} of Algorithm~\ref{alg:broadcast}).

Finally, in the tenth row of Figure~\ref{fig:trace-abstract}, nodes $v_1$, $v_2$ and $v_4$
trigger $\deliverBatch([\bllt{2}{\B}], \true)$, since they all received
$\READY(\bllt{2}{\B},\true)$ from the quorum $\{v_1,v_2,v_4\}$ to which all belong
(lines~\ref{lin:deliver}--\ref{lin:send-deliver} of Algorithm~\ref{alg:broadcast}), and
they all decide value $\B$ and end the execution.
\end{example}


\section{Proofs in \S\ref{sec:concrete-federated-consensus}}
\label{ap:correctness-concrete}

\begin{definition} \label{def:maptrace} %
  We define $\fh(\tau) = max\{b \mid \forall b' \lic b.~v.b'.\deliver(\false) \in \tau\}$
  and let $\transmit \in \{\send,\receive\}$ and $\M\in\{\VOTE,\READY\}$.
  \begin{displaymath}
    \begin{array}{rcl}
      \maptrace(\emptrace) &=& \emptrace\\

      \maptrace(\tau\concat [v.\prepare(b)])
                           &=& \maptrace(\tau) \concat
                               v.\voteBatch([b',\; b' \lic b], \false)\\

      \maptrace(\tau\concat [v.\commit(b)])
                           &=& \maptrace(\tau) \concat
                               v.\voteBatch([b',\;
                               \fh(\maptrace(\tau)) < b' \leq b], \true)\\

      \maptrace(\tau\concat [v.\prepared(b)])
                           &=&\maptrace(\tau) \concat
                               v.\deliverBatch([b' ,\; b' \lic b \\
                           &&\qquad {}
                               \land \forall v.\deliverBatch(bs) \in \maptrace(\tau).\, (b', \false) \not\in bs ], \false)\\

      \maptrace(\tau\concat [v.\committed(b)])
                           &=& \maptrace(\tau) \concat \deliverBatch([b], \true) \\

      \maptrace(\tau\concat [v.\transmit(\VOTE(\PREP\ b),u)])
                           &=& \maptrace(\tau) \concat
                               v.\transmitBatch([\M(b', \false),\; b' \lic b \\
                           &&\qquad {} \land {\forall a \in \Bool.\,
                              \forall v.\transmitBatch(ms,u) \in \maptrace(\tau).\,
                              \M(b', a) \not\in ms}], u)\\

      \maptrace(\tau\concat [v.\transmit(\READY(\PREP\ b),u)])
                           &=& \maptrace(\tau) \concat
                               v.\transmitBatch([\M(b', \false),\; b' \lic b \\
                           &&\qquad {} \land \forall v.\transmitBatch(ms,u) \in
\maptrace(\tau).\, \M(b', \false) \not\in ms], u)\\

      \maptrace(\tau\concat [v.\transmit(\VOTE(\CMT\ b),u)])
                           &=& \maptrace(\tau) \concat
                               v.\transmitBatch([\VOTE(b', \true),\; \fh(\maptrace(\tau)) < b' \leq b ], u)\\

      \maptrace(\tau\concat [v.\transmit(\READY(\CMT\ b),u)])
                           &=& \maptrace(\tau) \concat v.\transmitBatch([\READY(b, \true)], u)\\

      \maptrace(\tau\concat [e])
                           &=& \maptrace(\tau) \concat [e]\qquad\text{otherwise}
    \end{array}
  \end{displaymath}
\end{definition}



\begin{lemma}\label{lem:coincide-consensus-state} %
  Let $\Ss$ be an FBQS with some intact set $I$, $v$ be a node with $v \in I$, and $\tau$
  be a trace entailed by an execution of \CFCP. If $\maptrace(\tau)$ is a trace entailed
  by an execution of \AFCP, then $v.\n$, $v.\h$, and $v.\b$ coincide in both executions.
\end{lemma}
\begin{proof}
  We prove the statement by induction on $\tau$. For the base case, it suffices to
  observe, that $\b$, $\h$, and $\n$ coincide when initialised in
  line~\ref{lin:pbd-candidate-highest-init} and \ref{lin:pbd-current-round-init} of
  Algorithm~\ref{alg:concrete-federated-consensus} and
  line~\ref{lin:candidate-highest-init} and \ref{lin:current-round-init} of
  Algorithm~\ref{alg:abstract-federated-consensus}.  For the step case
  $\tau = \tau' \concat e$ we consider only the interesting cases, where $\b$, $\h$, or
  $\n$ are modified in line~\ref{lin:pbd-init-b}, line~\ref{lin:pbd-assign-b},
  line~\ref{lin:pbd-update-current-round}, line~\ref{lin:pbd-increase-candidate}, and
  line~\ref{lin:pbd-increase-round} of
  Algorithm~\ref{alg:concrete-federated-consensus}. For the other events in the concrete
  trace~$\tau$, the fields are not modified and the statement holds.
  \begin{description}
  \item[Case $e = v.\propose(x)$:] %
    By definition $\maptrace(\tau)$ contains $v.\propose(x)$, and by
    line~\ref{lin:pbd-init-b} of Algorithm~\ref{alg:concrete-federated-consensus} and by
    line~\ref{lin:init-b} of Algorithm~\ref{alg:abstract-federated-consensus}, $\b$ coincides.
  \item[Case $e = \prepared(b)$:] %
    By definition $\maptrace(\tau)$ contains
    $v.\deliverBatch([b', b' \lic b], \false)$. By induction hypothesis $\h$
    coincide, and therefore $\h < b$. Then by line~\ref{lin:assign-h} of
    Algorithm~\ref{alg:concrete-federated-consensus} and by line~\ref{lin:pbd-assign-h} of
    Algorithm~\ref{alg:abstract-federated-consensus}, $\h$ coincides. Again, by induction
    hypothesis $\b$ coincides, and therefore $\b \leq \h$ coincides. If $\b \leq \h$ holds
    then by line~\ref{lin:assign-b} of Algorithm~\ref{alg:concrete-federated-consensus}
    and by line~\ref{lin:pbd-assign-b} of
    Algorithm~\ref{alg:abstract-federated-consensus}, $\b$ coincides.
  \item[Case $e = \starttimer(n)$:] %
    By line~\ref{lin:pbd-quorum-round} of Algorithm~\ref{alg:concrete-federated-consensus}
    $\tau'$ contains $v.\receive(\M_u(\STA_u\ b_u), u)$ from $u$ with
    $\STA_u\in \{\CMT,\PREP\}$ for a $U\in\mathcal{Q}$ such that $v\in U$ and for each
    $u \in U$ exists $\M_u\in\{\VOTE,\READY\}$ and $b_u \in \Ballot$ such that
    $\n < b_u.\round$. %
    \begin{description}
    \item[Sub-case $\M_u(\PREP\ b_u)$.] By definition $\maptrace(\tau')$ contains a batch
      with $\M_u(b_u', \false)$ for every $b_u' \lic b_u$ and every $\M_u(\PREP\ b_u)$.
    \item[Sub-case $\M_u(\CMT\ b_u)$.] By definition $\maptrace(\tau')$ contains a batch
      with $\M_u(b_u, \true)$ for every $\M_u(\CMT\ b_u)$.
    \end{description}
    By induction hypothesis, $\n$ and therefore $\n < b_u.\round$ coincides. By
    line~\ref{lin:pbd-update-current-round} of Algorithm~\ref{alg:concrete-federated-consensus}
    and by line~\ref{lin:update-current-round} of
    Algorithm~\ref{alg:abstract-federated-consensus}, $\n$ coincides.
  \item[Case $e = \timeout$:] %
    By definition $\maptrace(\tau)$ contains $v.\timeout$, and by induction hypothesis
    $\b$, $\h$, and $\n$ coincide. Then by line~\ref{lin:pbd-increase-candidate} and
    \ref{lin:pbd-increase-round} of Algorithm~\ref{alg:concrete-federated-consensus} and
    line~\ref{lin:increase-candidate} and \ref{lin:increase-prepared} of
    Algorithm~\ref{alg:abstract-federated-consensus}, $\b$, $\h$, and $\n$ coincide.
  \end{description}
\end{proof}

\begin{lemma} \label{lem:upper-prep} %
  Let $\Ss$ be an FBQS with some intact set $I$, $v$ be a node with $v \in I$, and $\tau$
  be a trace entailed by an execution of \CFCP. Then for every ballot $b \in v.\Cd$
  (respectively, $b \in v.\Cr$) holds $b \leq v.\hd$ (respectively, $b \leq v.\hd$).
\end{lemma}
\begin{proof}
  Assume towards a contradiction, that there is a ballot~$b \in \Cd$ (respectively,
  $b \in v.\Cr$) such that $b > \hd$ (respectively, $b > v.\hd$). This is only possible,
  if $v$ sent $\READY(\PREP\ b')$ and $\READY(\CMT\ b)$ to itself where $b' < b$
  (lines~\ref{lin:hd} and \ref{lin:b-send-prepared}, and lines~\ref{lin:committed} and
  \ref{lin:delivered} of Algorithm~\ref{alg:bunch-voting}). But then $v$ sent contradicting
  messages, which contradicts that $v \in I$.
\end{proof}

\begin{lemma} \label{lem:max-false} %
  Let $\Ss$ be an FBQS with some intact set $I$, $v$ be a node with $v \in I$, and $\tau$
  be a trace entailed by an execution of \CFCP. If $\maptrace(\tau)$ is a trace entailed
  by an execution of \AFCP, for every $b > v.\hd$ holds $v.\brs[b].\delivered$ is false.
\end{lemma}
\begin{proof}
  Assume towards a contradiction, that $v.\brs[b].\delivered$ is true.  By
  lines~\ref{lin:deliver}--\ref{lin:send-deliver} of Algorithm~\ref{alg:broadcast} this is
  only possible if $\maptrace(\tau)$ contains an event $v.\sendBatch(ms, u)$ with
  $\READY(b, a) \in ms$ for $a \in \{ \true, \false \}$ from every $u$ in a quorum $U$.
  Assume $\READY(b, \true) \in ms$. Then by definition $\maptrace(\tau)$ contains
  $v.\send(\READY(\CMT\ b), u)$ and by lines~\ref{lin:committed} and \ref{lin:delivered}
  of Algorithm~\ref{alg:bunch-voting}, $b \in v.\Cd$. But then $b \leq v.\hd$ by
  Lemma~\ref{lem:upper-prep}. As $b > v.\hd$, $\maptrace(\tau)$ contains an event
  $v.\sendBatch(ms, u)$ with $\READY(b, \false) \in ms$ and by
  lines~\ref{lin:deliver}--\ref{lin:send-deliver} of Algorithm~\ref{alg:broadcast} this is
  only possible if $\maptrace(\tau)$ contains an event $v.\sendBatch(ms, u)$ where
  $\READY(b, \false) \in ms$ from every node~$u$ in a quorum~$U$ where $v \in U$. Again,
  by definition of $\maptrace$ and BNS this entails that $\tau$ contains
  $v.\receive(\READY(\PREP\ b_u), u)$ for $b' \lic b_u$ for every $b' \lic b$. But then,
  by lines~\ref{lin:b-prepared} and \ref{lin:hd} of Algorithm~\ref{alg:bunch-voting},
  $v.\hd$ is assigned to $b$ and this contradicts $b > v.\hd$.
\end{proof}

\begin{lemma} \label{lem:ready-false} %
  Let $\Ss$ be an FBQS with some intact set $I$, $v$ be a node with $v \in I$, and $\tau$
  be a trace entailed by an execution of \CFCP. If $\maptrace(\tau)$ is a trace entailed
  by an execution of \AFCP, for every $b > v.\hr$ holds $v.\brs[b].\ready$ is false.
\end{lemma}
\begin{proof}
  Assume towards a contradiction, that $v.\brs[b].\ready$ is true.  By
  lines~\ref{lin:ready}--\ref{lin:send-ready} and
  lines~\ref{lin:v-blocking}--\ref{lin:send-v-blocking} of Algorithm~\ref{alg:broadcast}
  this is only possible if $\maptrace(\tau)$ contains an event $v.\sendBatch(ms, u)$ with
  $\READY(b, a) \in ms$ for $a \in \{ \true, \false \}$ for every $u$ in either a
  quorum~$U$ or a $v$-blocking set~$B$.  Assume $\READY(b, \true) \in ms$. Then by
  definition $\maptrace(\tau)$ contains $v.\send(\READY(\CMT\ b), u)$ for every $u$ and by
  lines~\ref{lin:ready-prepare} and \ref{lin:update-hr}, or
  lines~\ref{lin:ready-prepare-v-blocking} and \ref{lin:update-hr-v-blocking} of
  Algorithm~\ref{alg:bunch-voting}, $b \in v.\Cr$. But then $b \leq v.\hr$ by
  Lemma~\ref{lem:upper-prep}. As $b > v.\hr$, $\maptrace(\tau)$ contains an event
  $v.\sendBatch(ms, u)$ with $\READY(b, \false) \in ms$ for every $u$ in either a
  quorum~$U$ or a $v$-blocking set~$B$.  Assume $\READY(b, \true) \in ms$. %
  Then again, by definition of $\maptrace$ and BNS this entails that $\tau$ contains
  $v.\receive(\READY(\PREP\ b_u), u)$ for $b' \lic b_u$ for every $b' \lic b$ for every
  $u$ in either a quorum~$U$ or a $v$-blocking set~$B$. But then, by
  lines~\ref{lin:ready-prepare} and \ref{lin:update-hr}, or
  lines~\ref{lin:ready-prepare-v-blocking} and \ref{lin:update-hr-v-blocking} of
  Algorithm~\ref{alg:concrete-federated-consensus} of Algorithm~\ref{alg:bunch-voting},
  $v.\hr$ is assigned to $b$ and this contradicts $b > v.\hr$.
\end{proof}

\begin{lemma}\label{lem:not-cd-delivered} %
  Let $\Ss$ be an FBQS with some intact set $I$, $v$ be a node with $v \in I$, and $\tau$
  be a trace entailed by an execution of \CFCP. If $\maptrace(\tau)$ is a trace entailed
  by an execution of \AFCP\ and $b \not\in \Cd$ then $b.\delivered$ is false.
\end{lemma}
\begin{proof}
  Assumes towards a contradiction that $b.\delivered$ is true. By
  lines~\ref{lin:deliver}--\ref{lin:send-deliver} of Algorithm~\ref{alg:broadcast} and
  BNS, this is only possible if $\maptrace(\tau)$ contains an event
  $v.\receiveBatch(ms, u)$ with $\READY(b, a) \in ms$ for $a \in \{ \true, \false \}$ from
  a quorum~$U$ such that $v\in U$. If $a$ is $\true$, then by definition of $\maptrace$,
  $\tau$ contains $v.\receive(\READY(\CMT\ b), u)$ from a quorum~$U$ such that $v\in
  U$. By lines~\ref{lin:committed} and \ref{lin:delivered} in
  Algorithm~\ref{alg:concrete-federated-consensus}, $b \in \Cd$ and this contradicts
  $b \not\in \Cd$. If $a$ is $\false$, then $v.\receive(\READY(\PREP\ b_u), u)$ from a
  quorum~$U$ such that $v\in U$ and $b' \lic b_u$ for every $b' \lic b$. Then by
  lines~\ref{lin:b-prepared} and \ref{lin:hd} of Algorithm~\ref{alg:bunch-voting}, $\hd$
  is assigned to $b$ and $b.\delivered$ is true contradicts Lemma~\ref{lem:max-false}.
\end{proof}

\begin{lemma}\label{lem:simulate-prefix}
  Let $\Ss$ be an FBQS with some intact set $I$ and $\tau$ be a trace entailed by an
  execution of \CFCP. For every finite prefix $\tau'$ of the projected trace $\tau|_I$,
  the simulated $\rho' = \maptrace(\tau')$ is the prefix of a trace entailed by an
  execution of \AFCP.
\end{lemma}
\begin{proof}
  We proceed by induction on the length of $\tau'$. The case $\tau'=\emptrace$
  is trivial since $\maptrace(\emptrace) = \emptrace$ is the prefix of any trace. We let
  $\tau' = \tau'_1 \concat [e]$ and consider the following cases:
  \begin{description}
  \item[Case $e = v.\prepare(b)$:] %
    For any execution of the \CFCP\ with trace $\tau_1'$, the prefix $\tau'_1$ contains
    either the event $v.\propose(b.x)$ by lines~\ref{lin:pbd-propose} and
    \ref{lin:pbd-send-prepare} of Algorithm~\ref{alg:concrete-federated-consensus}, or the
    event $v.\timeout$ by lines \ref{lin:pbd-timeout} and \ref{lin:pbd-prepare-increased}
    of Algorithm~\ref{alg:concrete-federated-consensus}. %
    The definition of $\maptrace$ entails that the simulated prefix
    $\rho_1' = \maptrace(\tau_1')$ contains either $v.\propose(b.x)$ or $v.\timeout$. By
    the induction hypothesis, the simulated prefix $\rho'_1$ is entailed by an execution
    of \AFCP. Let the subtrace that simulates event $e$ be
    $\rho_e' = v.\voteBatch([b', b' \lic b], \false)$. %
    We show that $\rho_1'\concat \rho_e'$ is the prefix of a trace entailed by an
    execution of \AFCP. %
    \begin{description}
    \item[Sub-case $v$ proposes $b.x$.] %
      By lines~\ref{lin:propose}--\ref{lin:send-prepare} of
      Algorithm~\ref{alg:abstract-federated-consensus}, $v$ triggers $v.b'.\vote(\false)$
      for every $b' \lic \bllt{1}{b.x}$ is in the execution of \AFCP.
    \item[Sub-case $v$ triggers $\timeout$] %
      By line~\ref{lin:pbd-prepare-increased} of
      Algorithm~\ref{alg:concrete-federated-consensus} ballot $b$ equals $\b$ and by
      Lemma~\ref{lem:coincide-consensus-state} $\b$ coincides. By
      lines~\ref{lin:timeout}--\ref{lin:prepare-increased} of
      Algorithm~\ref{alg:abstract-federated-consensus}, $v.b'.\vote(\false)$ for every
      $b' \lic b$ is in the execution of \AFCP.
    \end{description}
    As $v$ triggered $\vote(\false)$ for every $b' \lic b$ in both cases. When batched,
    this results in the event $\voteBatch([b', b' \lic b], \false)$, and
    $\rho'_1\concat \rho'_e$ is the prefix of a trace entailed by an execution of \AFCP.
  \item[Case $v.\commit(b)$.] %
    By lines~\ref{lin:pbd-prepared} and \ref{lin:pbd-send-commit} of
    Algorithm~\ref{alg:concrete-federated-consensus}, for any execution of \CFCP\ with
    trace~$\tau'$, the prefix $\tau'_1$ contains the event $v.\prepared(b)$.
    The definition of $\maptrace$ entails that the simulated prefix
    $\rho'_1 = \maptrace(\tau'_1)$ contains the event
    $v.\deliverBatch([b', b'\lic b], \false)$. By the induction hypothesis, the simulated
    prefix $\rho'_1$ is entailed by an execution of \AFCP.  Let the subtrace that
    simulates event $e$ be
    $\rho'_e = v.\voteBatch([b'', \fh(\tau_1') < b'' \leq b], true)$. We show that
    $\rho'_1\concat \rho_e'$ is the prefix of a trace entailed by an execution of \AFCP.
    Fix a ballot $b''$ where $\fh(\tau_1') < b'' \leq b$. By definition $\fh(\tau_1')$
    equals $\h$ and for every $b''$ holds $\h < b''$. Since $\rho'_1$ contains the event
    $v.\deliverBatch([b_i, b'\lic b], \false)$, $v$ triggered $b'.\deliver(\false)$ for
    each $b'\lic b$, and $\b$ and $\h$ coincide by
    Lemma~\ref{lem:coincide-consensus-state}, the guard at line~\ref{lin:prepared} of
    Algorithm~\ref{alg:abstract-federated-consensus} holds after any of such executions of
    \AFCP. %
    We can reason in the same fashion for every $b''$ in $\fh(\tau_1') < b'' \leq b$. By
    processing $b''$ in increasing order of ballots, $\b$ increases monotonically and
    triggers $v.\vote(b'', \true)$ for every ballot $b''$. As $v$ triggered
    $\vote(b'', \true)$ for every $\fh(\tau_1') < b'' \leq b$. When batched, this results
    in the event $v.\voteBatch([b'', \fh(\tau_1') < b'' \leq b], true)$, and therefore
    $\rho'_1\concat \rho'_e$ is the prefix of a trace entailed by an execution of \AFCP.
  \item[Case $e = v.\prepared(b)$.]
    By lines~\ref{lin:b-prepared} and \ref{lin:b-send-prepared} of
    Algorithm~\ref{alg:bunch-voting}, for any execution of \CFCP\ with trace~$\tau'$ there
    exists a maximum $b$ such $b > \hd$ and a quorum $U$ that contains node~$v$ and for
    each $u\in U$ node~$v$ received $\READY(\PREP\ b_u)$ where $b' \lic b_u$ for every
    $b' \lic b$.  Therefore the prefix $\tau'_1$ contains for every $u\in U$ the event
    $v.\receive(\READY(\PREP\ b_u),u)$.
    The definition of $\maptrace$ entails that the simulated prefix
    $\rho'_1 = \maptrace(\tau'_1)$ contains the event
    $v.\receiveBatch([\READY(b'_u, \false), b'_u\lic b_u], u)$ for each
    $v.\receive(\READY(\PREP\ b_u), u)$ that occurs in $\tau_1'$. By the induction
    hypothesis, the simulated prefix $\rho'_1$ is entailed by an execution of \AFCP.  Let
    the subtrace that simulates event $e$ be
    $\rho'_e = v.\deliverBatch([b'',\; b'' \lic b \land \forall v.\deliverBatch(bs) \in
    \maptrace(\tau).\, (b', \false) \not\in bs], \false)$. We show that
    $\rho'_1\concat \rho'_e$ is the prefix of a trace entailed by an execution of \AFCP.
    Fix a ballot $b''$ where $b'' \lic b$ and there is no batch
    $v.deliverBatch(bs, \false)$ with $b'' \in bs$ in $\rho_1'$. For each node $u\in U$,
    we know that $\rho'_1$ contains an event
    $v.\receiveBatch([\READY(b'_u, \false), b'_u\lic b_u], u)$.  As for every $b' \lic b$
    we know $b' \lic b_u$, we have $b'' \lic b_u$.  Thus and by BNS, we know that $v$
    received $\READY(b'',\false)$ from $u$.
    By Lemma~\ref{lem:max-false} and by $b > \hd$, we know that $b'.\delivered$ is false.
    Therefore, by lines~\ref{lin:deliver}--\ref{lin:send-deliver} of
    Algorithm~\ref{alg:broadcast}, triggers $v.b'.\deliver(b',\false)$.
    We can reason in the same fashion for every ballot $b'$ and batch the delivers in the
    event
    $v.\deliverBatch([b'',\; b'' \lic b \land \forall v.\deliverBatch(bs) \in
    \maptrace(\tau).\, (b', \false) \not\in bs], \false)$, and therefore
    $\rho'_1 \concat \rho'_e$ is the prefix of a trace entailed by an execution of \AFCP.

  \item[Case $e = v.\committed(b)$:] %
    By lines~\ref{lin:committed} and \ref{lin:send-committed} of
    Algorithm~\ref{alg:bunch-voting}, for any execution of \CFCP\ with trace $\tau_1'$, there
    exists a quorum~$U$ that contains node $v$ which is such that $v$ receives
    $\READY(\CMT\ b)$ from every $u$ in $U$ and $b \not\in \Cd$.
    The definition of $\maptrace$ entails that the simulated prefix
    $\rho'_1 = \maptrace(\tau'_1)$ contains an event
    $v.\receiveBatch([\READY(b, \true)], u)$ for each $v.\receive(\READY(\CMT\ b), u)$ that
    occurs in $\tau_1'$. By the induction hypothesis, the simulated prefix $\rho'_1$ is
    entailed by an execution of \AFCP. Let the subtrace that simulates the event $e$ be
    $\rho_e' = v.\deliverBatch([b], \true)$. We show that $\rho'_1 \concat \rho_e'$ is the
    prefix of trace entailed by an execution of \AFCP.
    As $v$ received $\READY(b, \true)$ from a quorum~$U$ where $v \in U$. As
    $b \not\in \delivered$ by Lemma~\ref{lem:not-cd-delivered} $\deliver$ is false, and by
    lines~\ref{lin:ready} and \ref{lin:send-ready} of Algorithm~\ref{alg:broadcast}
    triggers $v.\deliver(b, \true)$. When batched, this results in the event
    $v.\deliverBatch([b], \true)$, and therefore $\rho_1' \concat \rho_e'$ is the prefix
    of a trace entailed by an execution of \AFCP.

  \item[Case $e = v.\send(\VOTE(\PREP\ b), u)$.]  %
    By lines~\ref{lin:prepare} and \ref{lin:send-vote-prepared} of
    Algorithm~\ref{alg:bunch-voting}, for any execution of \CFCP\ with
    trace $\tau'$, the prefix $\tau_1'$ contains the event $v.\prepare(b)$.
    The definition of $\maptrace$ entails that the simulated prefix
    $\rho_1'=\maptrace(\tau_1')$ contains the event %
    $v.\voteBatch([b', b' \lic b], \false)$. %
    By the induction hypothesis, the simulated prefix~$\rho_1'$ is a trace entailed by an
    execution of \AFCP. Let the subtrace that simulates event $e$ be
    $\rho_e' = v.\sendBatch([\VOTE(b', \false), {b' \lic b} \land \forall a \in \Bool.\,
    \forall v.\sendBatch(ms,u) \in \maptrace(\tau).\, \M(b', a) \not\in ms], u)$.  We show
    that $\rho_1' \concat \rho_e'$ is the prefix of a trace entailed by an execution of
    \AFCP.
    Fix a ballot $b' \lic b$ such that $\rho_1'$ does not contain the an event
    $v.\sendBatch(ms, u)$ with $\VOTE(b', \false) \in ms$. Then by lines~\ref{lin:bc-init}
    and \ref{lin:echoed} of Algorithm~\ref{alg:broadcast} we know that the Boolean
    $\voted$ is $\false$. Hence, the condition in line~\ref{lin:if-echoed} of the same
    figure is satisfied, and since $v.\voteBatch([b', b' \lic b], \false)$, $v$ triggered
    $b'.\vote(\false)$, appending $v.\sendBatch(ms, u)$ with $\VOTE(b', \false) \in ms$
    results in a trace entailed by an execution of \AFCP\ by line~\ref{lin:send-echo} of
    the same figure. We can reason in the same fashion for every ballot $b' \lic b$ and
    conclude together with BNS that $\rho_1' \concat \rho_e'$ is the prefix of a trace
    entailed by an execution of \AFCP.

  \item[Case $e = v.\receive(\VOTE(\PREP\ b), u)$.] %
    By assumption the network does not create or drop messages, hence $v$ receives
    $\VOTE(\PREP\ b)$ only after $u$ previously sent the same message and the prefix
    $\tau_1'$ contains the event $u.\send(\VOTE(\PREP\ b), v)$.
    The definition of $\maptrace$ entails that the simulated prefix
    $\rho_1' = \maptrace(\tau_1')$ contains an event with
    $u.\sendBatch([\VOTE(b', a), b' \lic b], v)$ for $a \in \Bool$. By the induction
    hypothesis, the simulated prefix $\rho'_1$ is entailed by an execution of \AFCP.  Let
    the subtrace that simulates event $e$ be
    \begin{align*}
      \rho_e'={} &v.\receiveBatch([\VOTE(b', \false),\, {b' \lic b}\\
                &{}\land \forall a \in \Bool.\,
                 \forall v.\receiveBatch(ms,u) \in \maptrace(\tau).\,
                 \VOTE(b', a) \not\in ms], u)
    \end{align*}
    We show that $\rho_1' \concat \rho_e'$ is the prefix of a trace entailed by an
    execution \AFCP.
    By the ascending-ballot-order convention, it is enough to show that each $b' \lic b$,
    $v$ receives a batch with $\VOTE(b', a)$ for $a \in \Bool$ exactly once
    in $\rho'$.
    For a fixed $b'$, an event with $v.\receiveBatch(ms, u)$ with
    $\VOTE(b', \false) \in ms$ is in $\rho_e'$ only if $v.b'.\receive(\VOTE(b', a), u)$ is
    not in $\rho_1'$. On the other hand, $u$ sent a batch event with
    $u.b'.\send(\VOTE(b', a), v)$ for each $b' \lic b$. Hence, $\rho_1' \concat \rho_e'$
    is the prefix of a trace entailed by an execution of \AFCP.

  \item[Case $e = v.\send(\READY(\PREP\ b), u)$.] %
    For any execution of the \CFCP\ with trace $\tau_1'$, the node $v$ sends
    $\READY(\PREP\ b)$ either after hearing from a quorum in
    line~\ref{lin:send-ready-prepare} of Algorithm~\ref{alg:bunch-voting}, or after
    hearing from a $v$-blocking set in line~\ref{lin:send-ready-prepare-v-blocking} of the
    same figure. We consider both cases:
    \begin{description}
    \item[Sub-case $v$ sends $\READY(\PREP\ b)$ after hearing from a quorum.] %
      By lines~\ref{lin:ready-prepare}--\ref{lin:send-ready-prepare} of
      Algorithm~\ref{alg:bunch-voting}, exists a maximum ballot $b$ such that $\hr < b$ and
      there exists a quorum $U$ such that $v \in U$ and for every node $u \in U$ the node
      $v$ received $\VOTE(\PREP\ b_u)$ where $b' \lic b_u$ for every $b' \lic b$. %
      The definition of $\maptrace$ entails that the simulated prefix
      $\rho_1' = \maptrace(\tau_1')$ contains an event with %
      $v.\receiveBatch([\VOTE(b'_u, a), b'_u \lic b_u \land \forall a \in \Bool.\, \forall
      v.\receiveBatch(ms,u) \in \maptrace(\tau).\, \VOTE(b'_u, a) \not\in ms], u)$ for each
      node $u \in U$ and each event $v.\receive(\VOTE(\PREP\ b_u), u)$. %
      By the induction hypothesis, the simulated prefix $\rho_1'$ is a trace entailed by
      an execution of \AFCP.  Let the subtrace that simulates event~$e$ be
      $\rho_e' = v.\sendBatch([\READY(b', \false), b' \lic b \land \forall
      v.\sendBatch(ms,u) \in \maptrace(\tau).\, \READY(b', \false) \not\in ms], u)$.  We
      show that $\rho_1' \concat \rho_e'$ is the prefix of a trace entailed by an
      execution of \AFCP.
      If $b.\round = 1$ then by $b$ maximal and $b \lic b_u$, $v$ received a batch with
      $b'.\VOTE(b', \false)$ for $b' \lic b$ from every $u \in U$ such that $v \in U$.
      Then, by lines~\ref{lin:ready}--\ref{lin:send-ready} in
      Algorithm~\ref{alg:broadcast}, by BNS a batch with $\READY(b_j, \false)$
      is in $\rho_1'$.
      If $b.\round > 1$, and as $v$ is correct, by
      lines~\ref{lin:pbd-timeout}--\ref{lin:pbd-prepare-increased} and
      lines~\ref{lin:pbd-prepared}--\ref{lin:pbd-start-timer} of
      Algorithm~\ref{alg:concrete-federated-consensus} $v$ prepared the ballot
      $b_p^v = \bllt{b.\round - 1}{b^v.\val}$ in the previous round. By
      lines~\ref{lin:b-prepared}--\ref{lin:b-send-prepared} of
      Algorithm~\ref{alg:bunch-voting}, $v$ sends $\READY(\PREP\ b_p^v)$. Hence by
      definition of $\maptrace$, a batch with $\READY(b_j, a)$ is in
      $\rho_1'$ for every $b_j \lic b_p^v$.
      It remains to show that a batch with
      $v.\sendBatch([\READY(b_j, \false), b^p_v < b_j < b], u)$ is in
      $\rho_1' \concat \rho_e'$.  By assumption, for each node $u$ and ${b'_u \lic b_u}$
      the node $v$ received $\VOTE(b'_u, a)$. It suffices to show that the node $v$
      receives $\VOTE(b_j, \false)$ from every $u \in U$ for every ballot $b_j$.  Then, by
      lines~\ref{lin:ready}--\ref{lin:send-ready} and BNS in Algorithm~\ref{alg:broadcast}
      a batch with $\READY(b_j, \false)$ is in $\rho_1'$. By
      Lemma~\ref{lem:ready-false} and $b' > b > \hr$, $\ready$ is false for $b'$.
    \item[Sub-case $v$ sends $\READY(\PREP\ b)$ after hearing from a $v$-blocking set.] %
      By lines~\ref{lin:ready-prepare-v-blocking}--\ref{lin:send-ready-prepare-v-blocking}
      of Algorithm~\ref{alg:bunch-voting} there exists a maximum ballot $b$ such that $\hr < b$
      and there exists a $v$-blocking set $B$ such that for every $u \in B$ the node $v$
      received $\READY(\PREP\ b_u)$ where $b' \lic b_u$ for every $b' \lic b$. %
      The definition of $\maptrace$ entails that the simulated prefix
      $\rho_1' = \maptrace(\tau_1')$ contains an event
      $v.\receiveBatch([\READY(b'_u, \false), b'_u \lic b_u], u)$ for each node $u \in B$
      and each event $v.\receive(\READY(\PREP\ b_u), u)$.  By the induction hypothesis,
      the simulated prefix $\rho_1'$ is entailed by an execution of \AFCP. Let the
      subtrace that simulates event $e$ be
      $\rho_e' = v.\sendBatch([\READY(b', \false), b' \lic b \land \forall v.\sendBatch(ms,u)
      \in \maptrace(\tau).\, \READY(b', \false) \not\in ms], u)$.
      Fix a ballot $b'' \in B$ such that ${b'' \lic b}$ and for a batch with
      $v.b''.\send(\READY(b'', \false), u) \not\in \maptrace(\tau_1')$. By
      Lemma~\ref{lem:ready-false} and $b'' > b > \hr$, $\ready$ is false for $b''$. We
      have to show that $v$ received $\READY(b'', \false)$ from every $u$ in the
      $v$-blocking set $B$. Then by lines~\ref{lin:v-blocking}--\ref{lin:send-v-blocking}
      in Algorithm~\ref{alg:broadcast} $v$ send $\READY(b'', \false)$ to $u$. As for every
      $b' \lic b$ we know $b' \lic b_u$, we have $b'' \lic b_u$.  Thus, we know that a
      batch with $\READY(b'',\false)$ is in $\rho_1'$.
  \end{description}
  Both cases show that for the subtrace $\rho_e'$ that simulates event $e$, the trace
  $\rho_1' \concat \rho_e'$ is the prefix of a trace entailed by an execution of \AFCP.

  \item[Case $e = v.\receive(\READY(\PREP\ b), u)$.]  %
    Analogue to case $v.\receive(\VOTE(\PREP\ b), u)$.

  \item[Case $e = v.\send(\VOTE(\CMT\ b), u)$.]  %
    By lines~\ref{lin:commit} and \ref{lin:send-vote-commit} of
    Algorithm~\ref{alg:bunch-voting}, for any execution of \CFCP\ with trace $\tau'$ the
    prefix $\tau_1'$ contains the event $v.\commit(b)$.
    The definition of $\maptrace$ entails that the simulated prefix
    $\rho_1' = \maptrace(\tau_1')$ contains the event
    $v.\voteBatch([b', \fh(\rho'_1) < b' \leq b], \true)$. %
    By the induction hypothesis, the simulated prefix $\rho_1'$ is the prefix of a trace
    entailed by an execution of \AFCP.  Let the subtrace that simulates event $e$ be
    $\rho_e' = v.\sendBatch([\VOTE(b', \true), {\fh(\rho_1') < b' \leq b}],
    u)$. We show that $\rho_1' \concat \rho_e'$ is the prefix of a trace entailed by an
    execution of \AFCP.
    Fix a ballot $b'$ such that $\rho_1'$ does not contain a $v.\sendBatch(ms, u)$ with
    $\VOTE(b', \true) \in ms$. By line~\ref{lin:condition-commit} of
    Algorithm~\ref{alg:bunch-voting}, we know that $b' = \he$, and by
    lines~\ref{lin:prepare}--\ref{lin:send-vote-prepared} of the same figure, $v$ did not
    send $\VOTE(\PREP\ b'')$ for any $b'' > \he$. By definition of $\maptrace$, a batch
    event with $v.b'.\send(\VOTE(b', \false), u) \not\in \maptrace(\tau)$ for $b' > b$. As
    $b \not\in \Ce$, again by definition of $\maptrace$, a batch event with
    $v.b'.\send(\VOTE(b', \true), u) \not\in \maptrace(\tau)$. Therefore we know that the
    Boolean $\voted$ is $\false$.  Hence, the condition in line~\ref{lin:if-echoed} of the
    same figure is satisfied.  Since
    $v.\voteBatch([b', \fh(\rho'_1) < b' \leq b], \true)$, $v$ triggered
    $b'.\vote(\true)$, appending an event $v.\sendBatch(ms, u)$ with $\VOTE(b, \true) \in
    ms$ results in the prefix of a trace entailed by an execution of \AFCP\ by
    line~\ref{lin:send-echo}.
    We can reason in the same fashion for every $b'$ in
    $\fh(\maptrace(\tau)) < b' \leq b$, and therefore and by BNS $\rho'_1\concat \rho'_e$
    is a trace entailed by an execution of \AFCP.

  \item[Case $e = v.\receive(\VOTE(\CMT\ b), u)$.]  %
    By assumption the network does not create or drop messages, hence $v$ receives
    $\VOTE(\CMT\ b)$ only after $u$ previously sent the same message and the prefix
    $\tau_1'$ contains the event $u.\send(\VOTE(\CMT\ b), v)$.
    The definition of $\maptrace$ entails that the simulated prefix
    $\rho_1' = \maptrace(\tau_1')$ contains an event with
    $u.\sendBatch([\VOTE(b, \false)], v)$. By induction hypothesis $\rho_1'$ is the prefix
    of a trace entailed an execution of \AFCP.  Let the subtrace that simulates event $e$
    be
    $\rho_e' = v.\receiveBatch([\VOTE(b', \true), b' \in \{ b' \mid \fh(\rho_1') < b' \leq
    b \} ], u)$.  We show that $\rho_1' \concat \rho_e'$ is the prefix of a trace entailed
    by an execution of \AFCP.
    As $u$ sent $\VOTE(\CMT\ b)$ to $v$, we know that $v$ receives $\VOTE(b', \true)$
    exactly once for every $b' \in \{ b' \mid \fh(\rho_1') < b' \leq b \} $ and the batch
    is exactly once in $\rho'$.
    Hence, $\rho_1' \concat \rho_e'$ is the prefix of a trace entailed by an execution of
    \AFCP.

  \item[Case $e = v.\send(\READY(\CMT\ b), u)$.] %
    For any execution of \CFCP\ with trace $\tau_1'$, the node $v$
    sends $\READY(\CMT\ b)$ either after hearing from a quorum in
    line~\ref{lin:send-ready-commit} of Algorithm~\ref{alg:bunch-voting}, or after hearing
    from a $v$-blocking set in line~\ref{lin:send-ready-commit-v-blocking} of the same
    figure. We consider both cases:
    \begin{description}
    \item[Sub-case $v$ sends $\READY(\CMT\ b)$ after hearing from a quorum.] %
      By lines~\ref{lin:ready-commit}--\ref{lin:send-ready-commit} of
      Algorithm~\ref{alg:bunch-voting} there exists a quorum $U$ such that $v \in U$ and for
      every node $u \in U$ the node $v$ received $\VOTE(\CMT\ b)$ and
      $b \not\in \readied$ and $b \geqslant \hr$. %
      The definition of $\maptrace$ entails that the simulated prefix
      $\rho_1' = \maptrace(\tau_1')$ contains an event
      $v.\receiveBatch([\VOTE(b', a), \fh(\rho_1') < b' \leq b], u)$ and for every
      $u \in U$ such that $v \in U$ for every event $v.\receive(\VOTE(\CMT\ b), u)$.  By
      the induction hypothesis, the simulated prefix $\rho_1'$ is a trace entailed by an
      execution of \AFCP.  Let the subtrace that simulates event~$e$ be
      $v.\sendBatch([\READY(b, \true)], u)$.  We show that
      $\rho_1' \concat \rho_e'$ is the prefix of a trace entailed by an execution of
      \AFCP.
      If $v$ received $\VOTE(b, \true)$ from a quorum $U$ such that $v \in U$ and
      $\readied$ in Algorithm~\ref{alg:broadcast} is false, then by
      lines~\ref{lin:ready}--\ref{lin:send-ready} in Algorithm~\ref{alg:broadcast}, a
      batch with $\READY(b, \true)$ is in $\rho_1'$.
      Assume a $v.\receiveBatch(ms, u)$ with $\VOTE(b, \false) \in ms$ is in $\rho_1'$. By
      definition of $\maptrace$ this is only possible, if $v$ received $\VOTE(\PREP\ b_u)$
      for some $b_u > b$. As $v$ processed $v.\receive(\VOTE(\CMT\ b), u)$ and as $v$ is
      correct, $v$ cannot have processed $v.\receive(\READY(\PREP\ b_u) , u)$. Hence $v$
      received $\VOTE(b,\true)$ from $u$, and as $v$ has not received $\VOTE(b, \false)$,
      $\readied$ in Algorithm~\ref{alg:broadcast} is false.
    \item[Sub-case $v$ sends $\READY(\CMT\ b)$ after hearing from a $v$-blocking set.] %
      By lines~\ref{lin:ready-commit-v-blocking}--\ref{lin:send-ready-commit-v-blocking}
      of Algorithm~\ref{alg:bunch-voting} there exists a maximum ballot $b$ and a
      $v$-blocking set $B$ such that $v$ received $\READY(\CMT\ b)$ from every node
      $u \in B$ and $b \not\in \readied$ and $b \geqslant \hr$.
      The definition of $\maptrace$ entails that the simulated prefix
      $\rho_1' = \maptrace(\tau_1')$ contains the event $v.b.\receive(\READY(b, a), u)$
      for $a \in \{ \true, \false \}$ for every $u \in B$. By the induction hypothesis,
      $\rho_1'$ is the prefix of a trace entailed by an execution of \AFCP. Let the
      subtrace that simulates $e$ be $v.\sendBatch([\READY(b, \true)], u)$.  We
      show that $\rho_1' \concat \rho_e'$ is the prefix of a trace entailed by an
      execution of \AFCP.
      We have to show that $v$ received $\READY(b, \true)$ from a $v$-blocking set $B$ and
      $\readied$ in Algorithm~\ref{alg:broadcast} is false. Then by
      lines~\ref{lin:v-blocking}--\ref{lin:send-v-blocking} in
      Algorithm~\ref{alg:broadcast}, $v.\sendBatch(ms, u)$ with $\READY(b, \true)$ is in
      $\rho_1'$.
      Assume $v$ received $\READY(b, \false)$ from $u$. By definition of $\maptrace$ this
      is only possible, if $v$ received $\READY(\PREP\ b_u)$ for some $b_u > b$. As $v$
      processed $v.\receive(\READY(\CMT\ b), u)$ and as $v$ is correct, $v$ cannot have
      processed $v.\receive(\READY(\PREP\ b_u) , u)$. Hence $v$ received $\READY(b,\true)$
      from $u$, and as $v$ has not received $\VOTE(b, \false)$, $\readied$ in
      Algorithm~\ref{alg:broadcast} is false.
    \end{description}
    Both cases show that for the subtrace $\rho_e'$ that simulates event $e$, the trace
    $\rho_1' \concat \rho_e'$ is the prefix of a trace entailed by an execution of \AFCP.

  \item[Case $e = v.\receive(\READY(\CMT\ b), u)$.]  %
    Analogue to case $v.\receive(\VOTE(\CMT\ b), u)$.

  \item[Case $e = v.\propose(x)$.]  %
    Straightforward by definition of $\maptrace$, since $\tau$ contains $v.\propose(x)$
    iff the simulated $\rho =\maptrace(\tau)$ contains $v.\propose(x)$.

  \item[Case $e = v.\decide(x)$.]  %
    By lines~\ref{lin:pbd-decided}--\ref{lin:pbd-send-decided} in
    Algorithm~\ref{alg:concrete-federated-consensus}, for any execution of \CFCP\ with
    trace~$\tau'$ the node $v$ decides value $x$ only after $v$ triggers $\committed(b)$
    for a ballot $b$ with $b.\val = x$.
    The definition of $\maptrace$ entails that the simulated prefix
    $\rho_1' = \maptrace(\tau_1')$ contains the event
    $v.\deliverBatch([b], \true)$. By induction hypothesis $\rho_1'$, the simulated
    prefix $\rho'_1$ is entailed by an execution of \AFCP.  Let the subtrace that
    simulates event $e$ be $\rho_e' = [v.\decide(x)]$. We show that
    $\rho_1' \concat \rho_e'$ is the prefix of a trace entailed by an execution of
    \AFCP. %
    As $v.\deliverBatch([b], \true)$, $v$ triggered $\deliver(\true)$ for ballot~$b$, by
    lines~\ref{lin:decided} and \ref{lin:send-decided} of
    Algorithm~\ref{alg:abstract-federated-consensus}, $v.\decide(x)$ is in the execution
    of \AFCP\ and $\rho_1' \concat \rho_e'$ is the prefix of a trace entailed by an
    execution of \AFCP.

  \item[Case $e = v.\starttimer(n)$:] %
    By lines~\ref{lin:pbd-quorum-round}--\ref{lin:pbd-start-timer} of
    Algorithm~\ref{alg:concrete-federated-consensus} for any execution of \CFCP\ with
    trace~$\tau_1'$, there exists a quorum $U$ which is such that $v$ receives
    $\M(\STA\ b_{u})$ where $\M\in\{\VOTE,\READY\}$ and $\STA\in \{\CMT,\PREP\}$ from
    every $u$ in $U$ and $\n < b_{u}.\round$.
    The definition of $\maptrace$ entails that the simulated prefix
    $\rho'_1 = \maptrace(\tau'_1)$ contains
    the event
    \begin{displaymath}
      v.\receiveBatch([\M(b', \false), b' \lic b_{u}], u)
    \end{displaymath}
    for every $v.\receive(\M(\PREP\ b_{u}), u)$ that occurs in $\tau_1'$, or
    $v.\receiveBatch([\M(b_{u}, \true)], u)$ for every $v.\receive(\M(\CMT\ b_{u}), u)$
    that occurs in $\tau_1'$.
    By induction hypothesis $\rho'_1$ is the prefix of a trace entailed by an execution of
    \AFCP. Let the subtrace that simulates the event $e$ be
    $\rho_e' = [v.\starttimer(n)]$. We show that $\rho'_1 \concat \rho_e'$ is the
    prefix of trace entailed by an execution of \AFCP.
    By Lemma~\ref{lem:coincide-consensus-state} coincides $\n$ and by assumption
    $\round < b_{u}.\n$ holds.
    We have distinguish two cases:
    \begin{description}
    \item[Sub-case $v$ received $\M(\PREP\ b_{u})$ from $u$.] %
      The definition of $\maptrace$ entails that the simulated prefix
      $\rho'_1 = \maptrace(\tau'_1)$ contains an event with
      $v.\receiveBatch([\M_u(b_u', \false), b_u' \lic b_u], u)$ and every
      $\M_u(\PREP\ b_u)$.
    \item[Sub-case $\M_u(\CMT\ b_u)$.] The definition of $\maptrace$ entails that the
      simulated prefix $\rho'_1 = \maptrace(\tau'_1)$ contains a
      $v.\receiveBatch([\M_u(b_u, \true)], u)$ for every $\M_u(\CMT\ b_u)$.
    \end{description}
    Combining the cases leads to the conditions in line~\ref{lin:quorum-round} in
    Algorithm~\ref{alg:abstract-federated-consensus} satisfied.
    Thus, by line~\ref{lin:pbd-start-timer} of the same figure, $v.\starttimer(n)$ is in
    the execution of \AFCP\ and $\rho_1' \concat \rho_e'$ is the prefix of a trace entailed
    by an execution of \AFCP.

  \item[Case $e = v.\timeout$:] %
    Straightforward by definition of $\maptrace$, since $\tau$ contains $v.\timeout$
    iff the simulated $\rho =\maptrace(\tau)$ contains $v.\timeout$.
    \qedhere
  \end{description}
\end{proof}

\begin{proof}[Proof of Theorem~\ref{thm:refinement}]
  Let $\tau$ be the trace entailed by an execution of \CFCP\ over $\Ss$. We prove that
  there exists a trace $\rho$ entailed by an execution of \AFCP\ such that
  $\hist(\tau|_I) = \hist(\rho|_I)$ by \emph{reductio ad absurdum}. Assume towards a
  contradiction that for all traces~$\rho$, if $\hist(\tau|_I) = \hist(\rho|_I)$ then
  $\rho$ is not a trace entailed by an execution of \AFCP. Fix the trace $\rho$ to be
  $\maptrace(\tau|_I)$, which entails that $\hist(\tau|_I) = \hist(\rho|_I)$ by definition
  of $\maptrace$ and $\hist$. Since the number of events in a trace entailed by \AFCP\ is
  bounded by $\omega$, we denote the $i$th event of $\rho$ as $e_i$, with $i$ a natural
  number. Since $\rho$ is not a trace entailed by an execution of \AFCP\ by assumptions,
  there exists $i\geq 0$ such that the prefix $[e_i,\ldots, e_i]$ of $\rho$ is a trace
  entailed by \AFCP, but the prefix $[e_0,\ldots,e_{i+1}]$ of $\rho$ is not a trace
  entailed by \AFCP. Since $\maptrace$ maps one event of a concrete trace into one event
  of an abstract trace, there exists a finite prefix $\tau'$ of $\tau|_I$ such that
  $\maptrace(\tau') = [e_0,\ldots,e_{i+1}]$. But this leads to a contradiction because
  $\maptrace(\tau')$ is a trace of an execution of \AFCP\ by
  Lemma~\ref{lem:simulate-prefix}. Therefore, there exists a trace $\rho$ entailed by an
  execution of \AFCP\ such that $\hist(\tau|_I) = \hist(\rho|_I)$, and we are done. %
\end{proof}

\begin{proof}[Proof of Corollary~\ref{cor:cscp-correct}] %
  Let $\tau$ be the trace entailed by an execution of \CFCP\ over $\Ss$. Assume towards a
  contradiction that the execution does not satisfy some of the properties of
  \emph{Integrity}, \emph{Agreement for intact sets}, \emph{Weak validity for intact
    sets}, or \emph{Non-blocking for intact sets}. By Theorem~\ref{thm:refinement}, there
  exists a trace $\rho$ entailed by an execution of \AFCP\ over $\Ss$ and such that
  $\hist(\tau|_I) = \hist(\rho|_I)$. By definition of history, $\hist(\tau|_I)$ and
  $\hist(\rho|_I)$ coincide in their respective propose and decide events. Since $\rho|_I$
  is entailed by an execution of \AFCP, this execution fails to satisfy some of the
  properties of \emph{Integrity}, \emph{Agreement for intact sets}, \emph{Weak validity
    for intact sets}, or \emph{Non-blocking for intact sets}, which contradicts
  Theorem~\ref{thm:non-blocking-byzantine-consensu-intact-sets}.
\end{proof}

\begin{example}\label{ex:concrete}
\begin{figure}[t]
\begin{center}
  \begingroup\SmallTabColSep{\small\begin{tabular}[t]{|l|l|l|l|}
      Node $v_1$&Node $v_2$&{\color{red} Node $v_3$}&Node $v_4$\\
      \hline\hline
      \begin{tabular}[t]{l}
        $\propose(\C)$\\
        $\brs.\prepare(\bllt{1}{\C})$\\
        \hdashline
        $\VOTE(\PREP\ \bllt{1}{\C})$
      \end{tabular}&

      \begin{tabular}[t]{l}
        $\propose(\C)$\\
        $\brs.\prepare(\bllt{1}{\C})$\\
        \hdashline
        $\VOTE(\PREP\ \bllt{1}{\C})$
      \end{tabular}&

      \begin{tabular}[t]{l}
        {\color{red}$\VOTE(\PREP\ \bllt{1}{\B})$}
      \end{tabular}&

      \begin{tabular}[t]{l}
        $\propose(\A)$\\
        $\brs.\prepare(\bllt{1}{\A})$\\
        \hdashline
        $\VOTE(\PREP\ \bllt{1}{\A})$
      \end{tabular}\\

      \hline

      \begin{tabular}[t]{l}
        $\starttimer(F(1))$\\
        \hdashline
        $\READY(\PREP\ \bllt{1}{\B})$
      \end{tabular}&

      \begin{tabular}[t]{l}
        $\starttimer(F(1))$\\
        \hdashline
        $\READY(\PREP\ \bllt{1}{\B})$
      \end{tabular}&&

      \begin{tabular}[t]{l}
        $\starttimer(F(1))$\\
        \hdashline
        $\READY(\PREP\ \bllt{1}{\A})$
      \end{tabular}\\

      \hline

      \begin{tabular}[t]{l}
        $\brs.\prepared(\bllt{1}{\A})$
      \end{tabular} &

      \begin{tabular}[t]{l}
        $\brs.\prepared(\bllt{1}{\A})$
      \end{tabular} &
      &
        \begin{tabular}[t]{l}
          $\brs.\prepared(\bllt{1}{\A})$\\
          $\brs.\commit(\bllt{1}{\A})$\\
          \hdashline
          $\VOTE(\CMT\ \bllt{1}{\A})$\\
          $\READY(\PREP\ \bllt{1}{\B})$
      \end{tabular}\\

      \hline

      \begin{tabular}[t]{l}
        $\brs.\prepared(\bllt{1}{\B})$
      \end{tabular}&

      \begin{tabular}[t]{l}
        $\brs.\prepared(\bllt{1}{\B})$
      \end{tabular}&&

      \begin{tabular}[t]{l}
        $\brs.\prepared(\bllt{1}{\B})$\\
        $\brs.\commit(\bllt{1}{\B})$\\
        \hdashline
        $\VOTE(\CMT\ \bllt{1}{\B})$
      \end{tabular}\\

      \hline \qquad\vdots&\qquad\vdots&\qquad\vdots&\qquad\vdots\\
      \hline

      \begin{tabular}[t]{l}
        $\timeout$\\
        $\brs.\prepare(\bllt{2}{\B})$\\
        \hdashline
        $\VOTE(\PREP\ \bllt{2}{\B})$
      \end{tabular}&

      \begin{tabular}[t]{l}
        $\timeout$\\
        $\brs.\prepare(\bllt{2}{\B})$\\
        \hdashline
        $\VOTE(\PREP\ \bllt{2}{\B})$
      \end{tabular}&&

      \begin{tabular}[t]{l}
        $\timeout$\\
        $\brs.\prepare(\bllt{2}{\B})$\\
        \hdashline
        $\VOTE(\PREP\ \bllt{2}{\B})$
      \end{tabular}\\

      \hline

      \begin{tabular}[t]{l}
        $\starttimer(F(2))$\\
        \hdashline
        $\READY(\PREP\ \bllt{2}{\B})$
      \end{tabular}&

      \begin{tabular}[t]{l}
        $\starttimer(F(2))$\\
        \hdashline
        $\READY(\PREP\ \bllt{2}{\B})$
      \end{tabular}&&

      \begin{tabular}[t]{l}
        $\starttimer(F(2))$\\
        \hdashline
        $\READY(\PREP\ \bllt{2}{\B})$
      \end{tabular}\\

      \hline

      \begin{tabular}[t]{l}
        $\brs.\prepared(\bllt{2}{\B})$\\
        $\brs.\commit(\bllt{2}{\B})$\\
        \hdashline
        $\VOTE(\CMT\ \bllt{2}{\B})$
      \end{tabular}&

      \begin{tabular}[t]{l}
        $\brs.\prepared(\bllt{2}{\B})$\\
        $\brs.\commit(\bllt{2}{\B})$\\
        \hdashline
        $\VOTE(\CMT\ \bllt{2}{\B})$
      \end{tabular}&&

      \begin{tabular}[t]{l}
        $\brs.\prepared(\bllt{2}{\B})$\\
        $\brs.\commit(\bllt{2}{\B})$\\
        \hdashline
        $\VOTE(\CMT\ \bllt{2}{\B})$
      \end{tabular}\\

      \hline

      \begin{tabular}[t]{l}
        $\READY(\CMT\ \bllt{2}{\B})$
      \end{tabular}&

      \begin{tabular}[t]{l}
        $\READY(\CMT\ \bllt{2}{\B})$
      \end{tabular}&&

      \begin{tabular}[t]{l}
        $\READY(\CMT\ \bllt{2}{\B})$
      \end{tabular}\\

      \hline

      \begin{tabular}[t]{l}
        $\brs.\committed(\bllt{2}{\B})$\\
        $\decide(\B)$
      \end{tabular}&

      \begin{tabular}[t]{l}
        $\brs.\committed(\bllt{2}{\B})$\\
        $\decide(\B)$
      \end{tabular}&&

      \begin{tabular}[t]{l}
        $\brs.\committed(\bllt{2}{\B})$\\
        $\decide(\B)$
      \end{tabular}\\

      \hline

    \end{tabular}}\endgroup
  \end{center}
  \caption{Execution of \CFCP.}
  \label{fig:trace-concrete}
\end{figure}

  Recall Example~\ref{ex:abstract} in Appendix~\ref{ap:abstract-federated-consensus}. Compare the
  execution of \AFCP\ in Figure~\ref{fig:trace-abstract} with infinitely many events and
  messages with the finite execution of \CFCP\ in Figure~\ref{fig:trace-concrete}.
  The nodes propose the same values as in Example~\ref{ex:abstract}. In particular, in the
  first row, the faulty node $v_3$ sends $\VOTE(\PREP\ \bllt{1}{\B})$ to every correct
  node.
  As in \AFCP\ every correct node starts a timer in the second row. As in \AFCP\
  node~$v_4$ has prepared $\bllt{1}{\A}$ and sends $\READY(\PREP\ \bllt{1}{\A})$ after
  receiving $\VOTE(\PREP\ b_u)$ from a quorum for
  $b_u \in \{ \bllt{1}{\B}, \bllt{1}{\C} \}$ where $b' \in \{ \nbllt \}$ and $b' \lic b_u$
  (lines~\ref{lin:ready-prepare}--\ref{lin:send-ready-prepare} of
  Algorithm~\ref{alg:bunch-voting}).
  In the third row, the nodes $v_1$, $v_2$ and $v_4$ prepare the maximum ballot
  $\bllt{1}{\A}$, as they received $\READY(\PREP\ b_u)$ from a quorum for
  $b_u \in \{ \bllt{1}{\A}, \bllt{1}{\B} \}$ where $b' \in \{ \nbllt \}$ and $b' \lic b_u$
  (lines~\ref{lin:b-prepared}--\ref{lin:send-ready-prepare} of
  Algorithm~\ref{alg:bunch-voting}).
  Now node $v_4$ reaches its candidate value $\bllt{1}{\A}$ and therefore votes for
  it. But at the same time, $v_4$ receives $\READY(\PREP(\bllt{1}{\B}))$ from the
  $v_4$-blocking set $\{v_1,v_2\}$ and sends $\READY(\PREP\ \bllt{1}{\B})$
  (lines~\ref{lin:ready-prepare-v-blocking}--\ref{lin:send-ready-prepare-v-blocking} of
  Algorithm~\ref{alg:bunch-voting}).
  In the fourth, node $v_4$ only votes one commit statement $\CMT\ \bllt{1}{\B}$, as
  opposed to voting $\true$ for the two ballots $\bllt{1}{\A}$ and $\bllt{1}{\B}$ in the
  fourth row of Figure~\ref{fig:trace-abstract}. Similar to Example~\ref{ex:abstract}, the
  correct nodes decide value $\B$, which was not proposed by any correct node.
  As in Figure~\ref{fig:trace-concrete}, at this point no node can decide any value, because
  there is no ballot with a quorum of nodes for it, and the timeouts of all correct nodes
  will expire after $F(1)$ time. Then, in the sixth row of Figure~\ref{fig:trace-concrete},
  nodes $v_1$, $v_2$ and $v_4$ trigger $\timeout$, and since they all prepared ballot
  $\bllt{1}{\B}$, they update their candidate ballot to $\bllt{2}{\B}$.
  Now $v_1, v_2$ and $v_4$ have all the same candidate ballot and analogues to row six to
  nine in of Figure~\ref{fig:trace-abstract} can execute \CFCP\ to decide value~$\B$ and end
  the execution.

  For illustration, the executions in Figure~\ref{fig:trace-concrete}
  and~\ref{fig:trace-abstract} entail concrete and abstract traces $\tau$ and $\rho$
  respectively, which consist of the events on the left of each cell when traversing the
  tables in left-to-right, top-down fashion, and where the network events on the right of
  each cell are intermixed in such a way that the assumptions on atomic and batched
  semantics of \S\ref{sec:abstract-federated-consensus} are met. It is routine to check
  that $\hist(\tau|_{\{1,2,4\}})=\hist(\rho|_{\{1,2,4\}})$ and that
  $\rho|_{\{1,2,4\}} = \sigma(\tau|_{\{1,2,4\}})$.
\end{example}


\section{Proofs in \S\ref{sec:lying}}
\label{ap:lying}

Since in a subjective FBQS $\{\Ss_v\}_{v\in \V_\ok}$ each node $v$ runs the consensus
protocols according to its own view $\Ss_v$, the set of quorums $\Qs$ in the pseudo-code
of Algorithm~\ref{alg:broadcast} and
Algorithm~\ref{alg:abstract-federated-consensus}--\ref{alg:concrete-federated-consensus}
coincides with the set of quorums in the view $\Ss_v$. The notion of $v$-blocking set in
these figures also coincides with the usual notion in the view $\Ss_v$. To fix
terminology, we say that $U$ is a \emph{quorum known by $v$ in a subjective FBQS
  $\{\Ss_v\}_{v\in \V_\ok}$} iff $U$ is a quorum in $v$'s view $\Ss_{v}$, \ie,
$\forall v'\in U.\,\exists q\in \Ss_v(v').\,q\subseteq U$. We say that a set $B$ is
\emph{$v$-blocking in a subjective FBQS $\{\Ss_v\}_{v\in \V_\ok}$} iff $B$ overlaps each
of $v$'s slices in $v$'s view $\Ss_v$, \ie,
$\forall q\in \Ss_v(v).\,q\cap B \not= \emptyset$.

\begin{proof}[Proof of Lemma~\ref{lem:subjective-intact-set}]
  Since the nodes in $I$ are correct, they never equivocate their quorum slices, and thus
  $I$ is a quorum in each view $\Ss_v$ with $v\in I$ iff $I$ is a quorum in
  $\{\Ss_v\}_{v\in \V_\ok}$. For the same reason, all the views $\Ss_v|_I$ with $v\in I$
  coincide with each other. Thus, every node in $I$ is intertwined with each other in the
  all the views $\Ss_v$ with $v\in I$ iff every node in $I$ is intertwined with each other
  in $\{\Ss_v\}_{v\in \V_\ok}$, and the lemma holds.
\end{proof}

\begin{proof}[Proof of Lemma~\ref{lem:subjective-quorum-intersection}]
  Since $U_1$ and $U_2$ are in $\bigcup_{v\in\V_{\ok}}\Ss_v$, there exist correct nodes
  $v_1$ and $v_2$ such that $U_1$ is a quorum known by $v_1$ and $U_2$ is a quorum known
  by $v_2$. By Lemma~\ref{lem:quorum-project} the sets $U_1\cap I$ and $U_2\cap I$ are
  quorums in $\Ss_{v_1}|_I$ and $\Ss_{v_2}|_I$ respectively. Since the nodes in $I$ never
  equivocate their slices, all the views $\Ss_v|_I$ with $v\in I$ coincide with each
  other. Since every node is intertwined with each other in such a view, every two quorums
  have non-empty intersection by Lemma~\ref{lem:quorum-intersection}. Therefore,
  $(U_1\cap I)\cap(U_2\cap I) = (U_1\cap U_2)\cap I \not= \emptyset$, and the intersection
  $U_1\cap U_2$ contains some node in $I$.
\end{proof}

\begin{lemma}
  \label{lem:subjective-empty-I-not-v-blocking}
  Assume that $\V_\ok$ is the set of correct nodes and let $\{\Ss_v\}_{v\in \V_{\ok}}$ be
  a subjective FBQS with some intact set $I$. Let $v\in I$. Then, no $v$-blocking set $B$
  exists such that $B\cap I = \emptyset$.
\end{lemma}
\begin{proof}
  Straightforward by the definition of $v$-blocking in $\{\Ss_v\}_{v\in \V_{\ok}}$ and by
  Lemma~\ref{lem:empty-I-not-v-blocking}.
\end{proof}

\begin{lemma}[Analogous to Lemma~46 in \cite{GG18b} for intact sets]
  \label{lem:subjective-federated-intact-set-ready}
  Let $\{\Ss_v\}_{v\in \V_{\ok}}$ be a subjective FBQS and $t$ be a tag, and consider an
  execution of the instance for $t$ of FV over $\{\Ss_v\}_{v\in \V_{\ok}}$. Let $I$ be an
  intact set in $\{\Ss_v\}_{v\in \V_{\ok}}$.  The first node $v\in I$ that sends a
  $\READY(t,a)$ message first needs to receive a $\VOTE(t,a)$ message from every member of
  a quorum $U$ known by $v$ and to which $v$ belongs.
\end{lemma}
\begin{proof}
  Let $v$ be any node in $I$. By Lemma~\ref{lem:subjective-empty-I-not-v-blocking} no
  $v$-blocking set $B$ exists such that $B\cap I=\emptyset$. Therefore, the first node
  $v\in I$ that sends a $\READY(t,a)$ message does it through
  lines~\ref{lin:ready}--\ref{lin:send-ready} of Algorithm~\ref{alg:broadcast}, which
  means that $v$ received $\VOTE(t,a)$ messages from every member of a quorum $U$ known by
  $v$ and to which $v$ belongs.
\end{proof}

\begin{lemma}
  \label{lem:subjective-federated-not-v-blocking-availability-intact-set}
  Assume that $\V_\ok$ is the set of correct nodes and let $\{\Ss_v\}_{v\in \V_{\ok}}$ be
  a subjective FBQS with some intact set $I$. Consider a set $B$ of nodes. If $B$ is not
  $v$-blocking in $\{\Ss_v\}_{v\in \V_{\ok}}$ for any $v\in I \setminus B$, then either
  $B\supseteq I$ or $I \setminus B$ is a quorum in $\Ss_v|_I$ for each $v\in I$.
\end{lemma}
\begin{proof}
  Assume $B$ is not $v$-blocking in $\{\Ss_v\}_{v\in \V_{\ok}}$ for any
  $v\in I \setminus B$. If $B \supseteq I$ then we are done. Otherwise, for every node $v$
  in $I \setminus B$, there exists a slice $q\in\Ss_v(v)$ such that $q\cap B=\emptyset$.
  We know that $q\cap I\not=\emptyset$ since $v\in q$ by definition of subjective FBQS. We
  also know that $q\cap I \in \Ss_v|_I(v)$ by definition of $\Ss_v|_I$, and since
  $q\cap B=\emptyset$, the intersection $q\cap I$ is a subset of $I\setminus B$. Since for
  each node $v\in I$ there exists a slice $q\in \Ss_v(v)$ such that $q\cap I$ is a subset
  of $I\setminus B$, the set $I \setminus B$ is a quorum in $\Ss_v|_I$. Since the nodes in
  $I$ are correct and never equivocate their slices, we know that
  $\Ss_{v_1}|_I=\Ss_{v_2}|_I$ for every two nodes $v_1$ and $v_2$ in $I$, and the required
  holds.
\end{proof}

\begin{lemma}[Analogous to Lemma~23 of \cite{GG18b} for intact sets]
  \label{lem:subjective-federated-ready-consistent-intact-set}
  Let $\{\Ss_v\}_{v\in \V_{\ok}}$ be a subjective FBQS and consider an execution of \AFCP\
  over $\{\Ss_v\}_{v\in \V_{\ok}}$. Let $I$ be an intact set in
  $\{\Ss_v\}_{v\in \V_{\ok}}$ and $b$ be a ballot. If two nodes in $I$ send respectively
  messages $\READY(b,a)$ and $\READY(b,a')$, then $a = a'$.
\end{lemma}
\begin{proof}
  Assume that two nodes in $I$ send respectively messages $\READY(b,a)$ and
  $\READY(b,a')$. By Lemma~\ref{lem:subjective-federated-intact-set-ready}, some node
  $v\in I$ has received $\VOTE(b,a)$ from a quorum $U$ known by $v$ to which $v$ belongs,
  and some node $v'\in I$ has received $\VOTE(b,a')$ from a quorum $U'$ known by $v'$ to
  which $v'$ belongs. By Lemma~\ref{lem:subjective-quorum-intersection}, the intersection
  $U\cap U'$ contains some node in $I$, so that this node has sent $\VOTE(t,a)$ and
  $\VOTE(t,a')$. But due to the use of the guard variable $\voted$ in
  lines~\ref{lin:bc-init} and \ref{lin:if-echoed}--\ref{lin:echoed} of
  Algorithm~\ref{alg:broadcast}, a node can only vote for one value per tag, and thus it
  cannot vote different values for the same tag. Hence, $a=a'$.
\end{proof}

\begin{lemma}[Analogous to Lemma~24 in \cite{GG18b} for intact sets]
  \label{lem:subjective-federated-totality-intact-set}
  Assume that $\V_\ok$ is the set of correct nodes and let $\{\Ss_v\}_{v\in \V_{\ok}}$ be
  a subjective FBQS with some intact set $I$. Assume that $I = V^+ \uplus V^-$ and for
  some quorum $U$ known by a node $v$ we have $U \cap I \subseteq V^+$. Then either
  $V^-=\emptyset$ or there exists some node $v'\in V^-$ such that $V^+$ is $v'$-blocking
  in $\{\Ss_v\}_{v\in \V_{\ok}}$.
\end{lemma}
\begin{proof}
  Since every one in $I$ is correct, all the views $\Ss_v|_I$ with $v \in I$ coincide with
  each other. Since $V^+$ and $V^-$ only contain nodes in $I$, they both lie within the
  projection $\Ss_v|_I$. Thus, for every $v\in V^+$, $V^-$ is $v$-blocking in
  $\{\Ss_v\}_{v\in\V_\ok}$ iff $V^-$ is $v$-blocking in any $\Ss_v'|_I$ with $v'\in I$.
  Assume that $V^+$ is not $v$-blocking in $\{\Ss_v\}_{v\in \V_{\ok}}$ for any $v\in V^-$.
  By Lemma~\ref{lem:federated-not-v-blocking-availability-intact-set}, either
  $V^-=\emptyset$ or $V^-$ is a quorum in $\Ss|_I$. In the former case we are done, while
  in the latter we get a contradiction as follows. By Lemma~\ref{lem:quorum-project}, the
  intersection $U\cap I$ is a quorum in $\Ss|_I$. Since every two quorums in $\Ss|_I$ have
  non-empty intersection by Lemma~\ref{lem:quorum-intersection}, we have
  $(U\cap I)\cap V^-\not=\emptyset$. But this is impossible, since $U\cap I \subseteq V^+$
  and $V^+\cap V^-=\emptyset$.
\end{proof}

\begin{lemma}
  Let $\{\Ss_v\}_{v\in \V_{\ok}}$ be a subjective FBQS and $t$ be a tag. The instance for
  $t$ of FV over $\{\Ss_v\}_{v\in \V_{\ok}}$ satisfies the specification of reliable
  Byzantine voting for intact sets.
\end{lemma}
\begin{proof}
  We prove that the instance for tag $t$ of FV over $\{\Ss_v\}_{v\in \V_{\ok}}$ enjoys
  each of the properties that define the specification of \emph{reliable Byzantine voting
    for intact sets}.
  \begin{description}
  \item \emph{No duplication:} Straightforward by the use of the guard variable
    $\delivered$ in line~\ref{lin:delivered} of Algorithm~\ref{alg:broadcast}.
  \item \emph{Totality for intact sets:} Assume some node $v$ in $I$ delivers a value $a$
    for tag $t$. By the condition in line~\ref{lin:ready} of
    Algorithm~\ref{alg:broadcast}, the node $v$ has received $\READY(t,a)$ messages from
    every member in a quorum $U$ known by $v$. Since $U\cap I$ contains only correct
    nodes, these nodes send $\READY(t,a)$ messages to every node. By the condition in
    line~\ref{lin:v-blocking} of Algorithm~\ref{alg:broadcast}, any correct node $v'$
    sends $\READY(t,a)$ messages if it receives $\READY(t,a)$ from every member in a
    $v'$-blocking set. Hence, the $\READY(t,a)$ messages from the nodes in $U\cap I$ may
    convince additional correct nodes to send $\READY(t,a)$ messages to every node. Let
    these additional correct nodes send $\READY(t,a)$ messages until a point is reached at
    which no further nodes in $I$ can send $\READY(t,a)$ messages. At this point, let
    $V^+$ be the set of nodes in $I$ that sent $\READY(t,a)$ messages (where
    $U\cap I\subseteq V^+$), and let $V^- = I \setminus V^+$. By
    Lemma~\ref{lem:subjective-federated-ready-consistent-intact-set} the nodes in $V^-$
    did not send any $\READY(t,\_)$ messages at all. The set $V^+$ cannot be $v'$-blocking
    for any node $v'$ in $V^-$, or else more nodes in $I$ could come to send $\READY(t,a)$
    messages. Then by Lemma~\ref{lem:subjective-federated-totality-intact-set} we have
    $V^-=\emptyset$, meaning that every node in $I$ has sent $\READY(t,a)$ messages. Since
    $I$ is a quorum known by every node in $I$, all the nodes in $I$ will eventually
    deliver a Boolean for tag $t$ due to the condition in line~\ref{lin:ready} of
    Algorithm~\ref{alg:broadcast}.
  \item \emph{Consistency for intertwined nodes:} Assume that two nodes $v$ and $v'$ in an
    intact set $I$ deliver values $a$ and $a'$ for tag $t$ respectively. By the condition
    in line~\ref{lin:deliver}, the nodes received $\READY(t,a)$ messages from a quorum
    known by $v$, respectively, $\READY(t,a')$ messages from a quorum known by $v'$. Since
    the two nodes are intertwined, there is a correct node $u$ in the intersection of the
    two quorums, which sent both $\READY(t,a)$ and $\READY(t,a')$. By the use of the guard
    variable $\readied$ in line~\ref{lin:ready} of Algorithm~\ref{alg:broadcast}, node $u$
    can only send one and the same ready message to every other node, and thus $a = a'$ as
    required.
  \item \emph{Validity for intact sets:} Assume every node in an intact set $I$ votes for
    value $a$. Since $I$ is a quorum known by every member of $I$, every node in $I$ will
    eventually send $\READY(t,a)$ by the condition in line~\ref{lin:ready} of
    Algorithm~\ref{alg:broadcast}. By
    Lemma~\ref{lem:subjective-federated-ready-consistent-intact-set}, these messages
    cannot carry a value different from $a$. Then by the condition in
    line~\ref{lin:deliver} of Algorithm~\ref{alg:broadcast} every node in $I$ will
    eventually deliver the value $a$ for tag $t$. Due to \emph{Consistency for intact
      sets}, no node delivers a value different from $a$.\qedhere
  \end{description}
\end{proof}

\begin{lemma}\label{lem:subjective-bounded-totality}
  Let $\{\Ss_v\}_{v\in \V_{\ok}}$ be a subjective FBQS and $t$ be a tag, and consider an
  execution of the instance for $t$ of FV over
  $\{\Ss_v\}_{v\in \V_{\ok}}$. Let $I$ be an intact set in $\{\Ss_v\}_{v\in \V_{\ok}}$ and
  assume that GST has expired. If a node $v\in I$ delivers a voting value then every node
  in $I$ will deliver a voting value within some bounded time.
\end{lemma}
\begin{proof}
  Analogous to the proof of Lemma~\ref{lem:bounded-totality}. %
  %
\end{proof}

Let $\{\Ss_v\}_{v\in \V_{\ok}}$ be a subjective FBQS and $t$ be a tag, and consider an
execution of the instance for tag $t$ of FV over $\{\Ss_v\}_{v\in \V_{\ok}}$. Let $I$ be
an intact set in $\{\Ss_v\}_{v\in \V_{\ok}}$. As we did in \S\ref{sec:stellar-broadcast},
we write $\delta_I$ for the finite time that a node in $I$ takes to deliver some voting
value after GST and provided that some other node in $I$ already delivered some voting
value. Lemma~\ref{lem:subjective-bounded-totality} guarantees that $\delta_I$ is finite.

\begin{lemma}\label{lem:subjective-prepared-before-commit-intact-set}
  Let $\{\Ss_v\}_{v\in \V_{\ok}}$ be a subjective FBQS and consider an execution of \AFCP\
  over $\{\Ss_v\}_{v\in \V_{\ok}}$.  Let $I$ be an intact set in
  $\{\Ss_v\}_{v\in \V_{\ok}}$. If a node $v_1\in I$ commits a ballot $b$, then some node
  $v_2\in I$ prepared $b$.
\end{lemma}
\begin{proof}
  Assume that a node $v_1 \in I$ commits ballot $b$. By line~\ref{lin:ready} of
  Algorithm~\ref{alg:broadcast}, node $v_1$ received $\READY(b,\true)$ from every member
  of a quorum known by $v_1$ and to which $v_1$ belongs. By
  Lemma~\ref{lem:subjective-federated-intact-set-ready} the first node $u\in I$ to do so
  received $\VOTE(b,\true)$ messages from every member of a quorum $U$ known by $u$ and to
  which $u$ belongs. Since $v_1$ is intertwined with every other node in $I$, there exists
  a correct node $v_2$ in the intersection $U \cap I$ that sent $\VOTE(b,\true)$. The node
  $v_2$ can send $\VOTE(b,\true)$ only through line~\ref{lin:send-echo} of
  Algorithm~\ref{alg:broadcast}, which means that $v_2$ triggers $\brs[b].\vote(\true)$ in
  line~\ref{lin:send-commit} of Algorithm~\ref{alg:abstract-federated-consensus}. By
  line~\ref{lin:prepared} of the same figure, this is only possible after $v_2$ has
  aborted every $b' \lic b$, and the lemma holds.
\end{proof}

\begin{lemma}\label{lem:subjective-ready-commit-prepare-intact-set}
  Let $\{\Ss_v\}_{v\in \V_{\ok}}$ be a subjective FBQS and consider an execution of \AFCP\
  over $\{\Ss_v\}_{v\in \V_{\ok}}$.  Let $I$ be an intact set in $\Ss$. Let $v_1$ and
  $v_2$ be nodes in $I$ and $b_1$ and $b_2$ be ballots such that $b_2\lic b_1$. The
  following two things cannot both happen: node $v_1$ prepares $b_1$ and node $v_2$ sends
  $\READY(b_2,\true)$.
\end{lemma}
\begin{proof}
  Assume towards a contradiction that $v_1$ prepares $b_1$, and that $v_2$ sends
  $\READY(b_2,\true)$. By definition of prepare, node $v_1$ aborted every ballot
  $b \lic b_1$. By line~\ref{lin:ready} of Algorithm~\ref{alg:broadcast}, node $v_2$
  received $\READY(b,\false)$ from every member of a quorum $U_b$ known by $v_1$ for each
  ballot $b \lic b_1$.  By assumptions, $b_2\lic b_1$, and therefore $v_1$ received
  $\READY(b_2,\false)$ from every member of the quorum $U_{b_2}$. By
  Lemma~\ref{lem:subjective-empty-I-not-v-blocking}, the first node $u_1\in I$ that sent
  $\READY(b_2,\false)$ received $\VOTE(b_2,\false)$ from a quorum $U_1$ known by $u_1$ and
  to which $u_1$ belongs. Since $v_2$ sent $\READY(b_2,\true)$ and by
  Lemma~\ref{lem:subjective-empty-I-not-v-blocking}, the first node $u_2\in I$ that sent
  $\READY(b_2,\true)$ received $\VOTE(b_2,\true)$ from a quorum $U_2$ known by $u_2$ and
  to which $u_2$ belongs. Since $u_1$ and $u_2$ are intertwined, the intersection
  $U_1\cap U_2$ contains some correct node $v$, which sent both $\VOTE(b_2,\false)$ and
  $\VOTE(b_2,\true)$ messages. By the use of the Boolean $\voted$ in line~\ref{lin:echo}
  of Algorithm~\ref{alg:broadcast} this results in a contradiction and we are done.
\end{proof}

\begin{lemma}\label{lem:subjective-commit-largest-prepared-intact-set}
  Let $\{\Ss_v\}_{v\in \V_{\ok}}$ be a subjective FBQS and consider an execution of \AFCP\
  over $\{\Ss_v\}_{v\in \V_{\ok}}$.  Let $I$ be an intact set in $\Ss$. If a node
  $v_1\in I$ commits a ballot $b_1$, then the largest ballot $b_2$ prepared by any node
  $v_2\in I$ before $v_1$ commits $b_1$ is such that $b_1\sim b_2$.
\end{lemma}
\begin{proof}
  Assume node $v_1$ commits ballot $b_1$. By the guard in line~\ref{lin:deliver} of
  Algorithm~\ref{alg:broadcast}, node $v_1$ received the message $\READY(b_1,\true)$ from
  every member of a quorum known by $v_1$ and to which $v_1$ belongs, which entails that
  node $v_1$ received $\READY(b_1,\true)$ from itself. By
  Lemma~\ref{lem:subjective-federated-intact-set-ready}, the first node $u\in I$ that send
  $\READY(b_1,\true)$ needs to receive an $\VOTE(b_1,\true)$ message from every member of
  some quorum known by $u$ and to which $u$ belongs. Thus, $u$ itself triggered
  $\brs[b_1].\vote(\true)$, which by lines~\ref{lin:send-prepare}
  and~\ref{lin:prepare-increased} of Algorithm~\ref{alg:abstract-federated-consensus}
  means that $u$ prepared ballot $b_1$. Hence, the largest ballot $b_2$ such that there
  exists a node $v_2\in I$ that triggers $\brs[b_2].\vote(\true)$ before $v_1$ commits
  $b_1$, is bigger or equal than $b_1$. If $b_2=b_1$, then $b_2.\val = b_1.\val$ and by
  lines~\ref{lin:prepared}--\ref{lin:send-commit} of
  Algorithm~\ref{alg:abstract-federated-consensus}, node $v_2$ prepares $b_2$ before it
  triggers $\brs[b_2].\vote(\true)$ and the lemma holds.

  If $b_2>b_1$, then we assume towards a contradiction that $b_2.\val\not=b_1.\val$. By
  lines~\ref{lin:prepared}--\ref{lin:send-commit} of
  Algorithm~\ref{alg:abstract-federated-consensus}, node $v_2$ prepared $b_2$. But this
  results in a contradiction by
  Lemma~\ref{lem:subjective-ready-commit-prepare-intact-set}, because $v_1$ and $v_2$ are
  intertwined and $v_1$ sent $\READY(b_1,\true)$, but $b_1\lic b_2$.  Therefore
  $b_2.\val=b_1.\val$, and by lines~\ref{lin:prepared}--\ref{lin:send-commit} of
  Algorithm~\ref{alg:abstract-federated-consensus}, node $v_2$ prepares $b_2$ before it
  triggers $\brs[b_2].\vote(\true)$.
\end{proof}

The definition of \emph{ready-tree for Boolean $a$ and ballot $b$ at node $v$} from
Appendix~\ref{ap:abstract-federated-consensus} can be lifted to the subjective FBQSs
straightaway, since the definition assumes that all nodes are honest, and thus they do not
equivocate their quorum slices and all the views coincide with each other.

\begin{lemma}\label{lem:subjective-cascade}
  Let $\{\Ss_v\}_{v\in \V_{\ok}}$ be a subjective FBQS $b$ be a ballot, and consider an
  execution of the instance for ballot $b$ of FV over $\{\Ss_v\}_{v\in \V_{\ok}}$. Assume
  all nodes are honest. If a node $v$ sends $\READY(b,a)$ then there exists a quorum $U$
  known by every node such that every member of $U$ sent $\VOTE(b,a)$.
\end{lemma}
\begin{proof}
  Straightforward by Lemma~\ref{lem:cascade}, since all nodes are honest and thus all the
  views coincide with each other.
\end{proof}

\begin{lemma}\label{lem:subjective-prepared-propose-intact-set}
  Let $\{\Ss_v\}_{v\in \V_{\ok}}$ be a subjective FBQS and consider an execution of \AFCP\
  over $\{\Ss_v\}_{v\in \V_{\ok}}$. Let $b_1$ be the largest ballot prepared by some node
  $v_1$ at some moment in the execution. If all nodes are honest, then some node $v_2$
  proposed $b_1.\val$.
\end{lemma}
\begin{proof}
  Straightforward by Lemma~\ref{lem:subjective-prepared-propose-intact-set}, since all
  nodes are honest and thus all the views coincide with each other.
\end{proof}

\begin{lemma}\label{lem:subjective-window-n}
  Let $\{\Ss_v\}_{v\in \V_{\ok}}$ be a subjective FBQS and consider an execution of \AFCP\
  over $\{\Ss_v\}_{v\in \V_{\ok}}$.  Let $I$ be a maximal intact set in
  $\{\Ss_v\}_{v\in \V_{\ok}}$ and assume that GST has expired. Let $v$ be a node in $I$
  that prepares some ballot $b$ such that no other node in $I$ has ever prepared a ballot
  with round bigger or equal than $b.\round$. In the interval of duration $\delta_I$ after
  $u$ prepares $b$, every node in $I$ that has not decided any value yet, either decides a
  value or prepares a ballot with round $b.\round$.
\end{lemma}
\begin{proof}
  Since $v$ has prepared $b$, then $v$ has delivered $\false$ for every ballot
  $b_i\lic b$. Let $u\in I$ be a node different from $v$ that has not decided any value
  yet. By assumptions, node $u$ has neither prepared any ballot with round bigger or equal
  than $b.\round$. Since GST has expired and by
  Lemma~\ref{lem:subjective-bounded-totality}, node $u$ will deliver $\false$ for every
  ballot $b_i\lic b$ within $\delta_I$, and the lemma holds.
\end{proof}

Let $\{\Ss_v\}_{v\in \V_{\ok}}$ be a subjective FBQS and consider an execution of \AFCP\
over $\{\Ss_v\}_{v\in \V_{\ok}}$. Let $I$ be an intact set in $\{\Ss_v\}_{v\in \V_{\ok}}$.
As in \S\ref{sec:abstract-federated-consensus}, we say the \emph{window for intact set $I$
  of round $n$} for the interval in which every node in an intact set $I$ that has not
decided any value yet prepares a ballot of round $n$. As in
Appendix~~\ref{ap:abstract-federated-consensus}, we let $v_n$ be the first node in $I$
that ever prepares a ballot $b_n$ with round $n$. The definition of
\emph{prepare-footprint of ballot $b_n$} from
Appendix~\ref{ap:abstract-federated-consensus} can be lifted to the subjective FBQSs
straightaway, since the definition assumes that all faulty nodes have stopped, and thus
the remaining correct nodes agree on the slices of every other correct node, and all the
quorums that are not stopped belong to all the views. We also distinguish the \emph{abort
  interval for intact set $I$ of round $n$} and the duration $\delta_I^{An}$, whose
definitions can be lifted to subjective FBQSs straightaway. As in
Appendix~\ref{ap:abstract-federated-consensus} we may omit the `for intact set $I$'
qualifier when the intact set is clear from the context.

\begin{lemma}\label{lem:subjective-window-no-overlap}
  Let $\{\Ss_v\}_{v\in \V_{\ok}}$ be a subjective FBQS and consider an execution of \AFCP\
  over $\{\Ss_v\}_{v\in \V_{\ok}}$.  Let $I$ be a maximal intact set in
  $\{\Ss_v\}_{v\in \V_{\ok}}$ and assume that all faulty nodes eventually stop. There
  exists a round $n$ such that either every node in $I$ decides some value before reaching
  round $n$, or otherwise the windows of all the rounds $m \geq n$ happen consecutively
  and never overlap with each other, and in each of the consecutive windows of round $m$
  the nodes in $I$ that have not decided any value yet only prepare ballots with round
  $m$.
\end{lemma}
\begin{proof}
  Analogous to the proof of Lemma~\ref{lem:window-no-overlap}.
\end{proof}

\begin{corollary}\label{cor:subjective-b-max}
  Let $m\geq n$ and let $b_{\max}$ be the maximum ballot prepared by any node in $I$
  before the abort interval of round $m+1$ starts. Every node in $I$ prepares $b_{\max}$
  before the abort interval of round $m+1$ starts.
\end{corollary}
\begin{proof}
  Analogous to the proof of Corollary~\ref{cor:b-max}.
\end{proof}

\begin{proof}[Proof of
  Theorem~\ref{thm:non-blocking-byzantine-consensus-intact-sets-lying}]
  Analogous to the proof of Theorem~\ref{thm:non-blocking-byzantine-consensu-intact-sets},
  by using
  Lemmas~\ref{lem:subjective-bounded-totality}--\ref{lem:subjective-window-no-overlap} and
  Corollary~\ref{cor:subjective-b-max}.
\end{proof}


\begin{lemma}\label{lem:coincide-consensus-state-lying} %
  Let $\{\Ss_v\}_{v\in \V_{\ok}}$ be a subjective FBQS with some intact set $I$, $v$ be a
  node with $v \in I$, and $\tau$ be a trace entailed by an execution of \CFCP. If
  $\maptrace(\tau)$ is a trace entailed by an execution of \AFCP, then $v.\n$, $v.\h$, and
  $v.\b$ coincide in both executions.
\end{lemma}
\begin{proof}
  Analogous to the proof of Lemma~\ref{lem:coincide-consensus-state}.
\end{proof}

\begin{lemma}\label{lem:upper-prep-lying}
  Let $\{\Ss_v\}_{v\in \V_{\ok}}$ be a subjective FBQS with some intact set $I$, $v$ be a
  node with $v \in I$, and $\tau$ be a trace entailed by an execution of \CFCP. Then for
  every ballot $b \in v.\Cd$ (respectively, $b \in v.\Cr$) holds $b \leq v.\hd$
  (respectively, $b \leq v.\hd$).
\end{lemma}
\begin{proof}
  Analogous to the proof of Lemma~\ref{lem:upper-prep}.
\end{proof}

\begin{lemma}\label{lem:max-false-lying}
  Let $\{\Ss_v\}_{v\in \V_{\ok}}$ be a subjective FBQS with some intact set $I$, $v$ be a
  node with $v \in I$, and $\tau$ be a trace entailed by an execution of \CFCP. If
  $\maptrace(\tau)$ is a trace entailed by an execution of \AFCP, for every $b > v.\hd$
  holds $v.\brs[b].\delivered$ is false.
\end{lemma}
\begin{proof}
  Analogous to the proof of Lemma~\ref{lem:max-false}, by using
  Lemma~\ref{lem:upper-prep-lying}.
\end{proof}

\begin{lemma}\label{lem:ready-false-lying}
  Let $\{\Ss_v\}_{v\in \V_{\ok}}$ be a subjective FBQS with some intact set $I$, $v$ be a
  node with $v \in I$, and $\tau$ be a trace entailed by an execution of \CFCP. If
  $\maptrace(\tau)$ is a trace entailed by an execution of \AFCP, for every $b > v.\hr$
  holds $v.\brs[b].\ready$ is false.
\end{lemma}
\begin{proof}
  Analogous to the proof of Lemma~\ref{lem:ready-false}, by using
  Lemma~\ref{lem:upper-prep-lying}.
\end{proof}

\begin{lemma}\label{lem:not-cd-delivered-lying}
  Let $\{\Ss_v\}_{v\in \V_{\ok}}$ be a subjective FBQS with some intact set $I$, $v$ be a
  node with $v \in I$, and $\tau$ be a trace entailed by an execution of \CFCP. If
  $\maptrace(\tau)$ is a trace entailed by an execution of \AFCP\ and $b \not\in \Cd$ then
  $b.\delivered$ is false.
\end{lemma}
\begin{proof}
  Analogous to the proof of Lemma~\ref{lem:not-cd-delivered}, by using
  Lemma~\ref{lem:max-false-lying}.
\end{proof}

\begin{lemma}\label{lem:simulate-prefix-lying}
  Let $\{\Ss_v\}_{v\in \V_{\ok}}$ be a subjective FBQS with some intact set $I$ and $\tau$
  be a trace entailed by an execution of \CFCP. For every finite prefix $\tau'$ of the
  projected trace $\tau|_I$, the simulated $\rho' = \maptrace(\tau')$ is the prefix of a
  trace entailed by an execution of \AFCP\ in
  Algorithm~\ref{alg:abstract-federated-consensus}.
\end{lemma}
\begin{proof}
  Analogous to the proof of Lemma~\ref{lem:simulate-prefix}, by using
  Lemmas~\ref{lem:coincide-consensus-state-lying}--\ref{lem:not-cd-delivered-lying}.
\end{proof}

\begin{proof}[Proof of Theorem~\ref{thm:refinement-lying}]
  Analogous to the proof of Theorem~\ref{thm:refinement}, by using
  Lemma~\ref{lem:simulate-prefix-lying}.
\end{proof}

\begin{proof}[Proof of Corollary~\ref{cor:cscp-correct-lying}] %
  Analogous to the proof of Corollary~\ref{cor:cscp-correct}, by using
  Theorem~\ref{thm:refinement-lying}.
\end{proof}



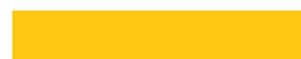
\end{document}